\documentclass[a4paper,UKenglish,cleveref, autoref, thm-restate]{lipics-v2019}

\pdfoutput=1
\title{From Matching Logic To Parallel Imperative Language Verification} 

\titlerunning{From Matching Logic To Parallel Imperative Language Verification} 

\author{ShangBei Wang}{Nanjing University of Aeronautics and Astronautics, Nanjing, China} {wangshangbei123@nuaa.edu.cn}{https://orcid.org/0000-0002-5047-3717}{}

\authorrunning{ShangBei Wang} 

\Copyright{ShangBei Wang} 


\category{} 

\relatedversion{} 

\supplement{}

\acknowledgements{}

\SeriesVolume{42}
\ArticleNo{42}

\begin{document}

\maketitle

\begin{abstract}
Program verification is to develop the program's proof system, and to
prove the proof system soundness with respect to a trusted operational semantics of the program. However, many practical program verifiers are not based on operational semantics and can't seriously validate the program. Matching logic is proposed to make program verification based on operational semantics. In this paper, following Grigore Ro{\c{s}}u et al's work, we consider matching logic for parallel imperative language(PIMP). According to our investigation, this paper is the first study on matching logic for PIMP. In our matching logic, we  redefine "interference-free" to character parallel rule and prove the soundness of matching logic to the operational semantics of PIMP. We also link PIMP's operational semantics and PIMP's verification formally by constructing a matching logic verifier for PIMP which executes
rewriting logic semantics symbolically on configuration patterns and is sound and complete to matching logic for PIMP. That is our matching logic verifier for PIMP is sound to the operational semantics of PIMP. Finally, we also verify the matching logic verifier through an example which is a standard
problem in parallel programming.
\end{abstract}
\section{Introduction}
\label{sec:typesetting-summary}
Operational semantics\cite{klin2011bialgebras}\cite{plotkin1981structural}\cite{aceto2001structural} is called ``transition semantics'' whose basic idea is to use a sequence $\gamma_0\to \gamma_1\to\gamma_i\to\cdots$ of configurations to formalize the execution of a program. $\gamma_i$ is either a terminal configuration or a nonterminal configuration.
The development of semantics engineering frameworks $\mathbb{K}$\cite{DBLP:journals/cacm/rosu10}\cite{DBLP:journals/cacm/rosu14}, Ott\cite{DBLP:journals/cacm/Sewell10} and PLT-Redex\cite{DBLP:journals/cacm/klein12} make it very easy to define the operational semantics of a programming language.
Consequently, the operational semantics of C\cite{DBLP:conf/focs/Hathhorn15}\cite{DBLP:journals/cacm/ellison12}, Java\cite{DBLP:conf/focs/Bogdanas15}, Python\cite{DBLP:journals/cacm/politz13}, CAML\cite{DBLP:conf/focs/Owens08}, JavaScript\cite{DBLP:conf/focs/Park15}\cite{DBLP:journals/cacm/Bodin14} have been proposed. The advantages of operational semantics such as easy to define and understand, being executable and being tested, make it suitable as trusted reference model for language. The ideal program  verification should use such operational semantics, unchanged, to produce proof certificates. However, program verification rarely use operational semantics directly, because the proof based on operational semantics directly involve the corresponding  transition system, which is generally considered low-level. Hoare\cite{DBLP:journals/cacm/Apt19}\cite{DBLP:journals/cacm/Hoare69} and dynamic logic\cite{harel1984} are typically used because their reasoning seems higher level. However, the set of abstract proof rules to define language semantics in Hoare and dynamic logic are hard to understand and trust. The essence of program verification is to develop a program's proof system, and to prove the proof system sound with respect to a trusted operational semantics of the program. However, instead of being based on a formal semantics, many practical program verifiers\cite{calcagno2015moving}\cite{bjorner2015}\cite{domenica2015method}\cite{DBLP:journals/cacm/Sasse07} convert the target program language to an intermediate validation language or simply implement ad-hoc verification condition. If a program verification is not based on the formal semantics of the program, the program verification does not seriously validate the program and the result can't be trusted.\\
Matching logic\cite{rosu2017matching} is proposed to make program verification based on operational semantics. To reason about program, we first need to define program configurations. Matching logic configuration patterns consist of variables, symbols in signature, first-order logical connectives and existential quantifiers.
For example $\exists z(o=<\!\!\mathrm{x}:=\mathrm{x}*\mathrm{y}\!\!>_k<\!\!\mathrm{x}\mapsto x,\mathrm{y}\mapsto z,\rho\!\!>_{\mathit{env}}<\!\!m\!\!>_{\mathit{mem}}\land x \ne 0)$ is matching logic configuration pattern where $o$ is a distinguished variable of sort $\mathit{Cfg}$ and $<\!\!\cdots\!\!>_k$ holds code fragment and$<\!\!\cdots\!\!>_{env}<\!\!\cdots\!\!>_{mem}$ holds program state and $x \ne 0$ is the constraint, an arbitrary first-order logical formula. There are two reasons why matching logic is particularly suitable for program reasoning:
\begin{itemize}
\item By matching logic configuration patterns, we can get access to any detail in the program and hide irrelevant details using existential quantization;
\item Both the operational semantics and reachability properties of a program can be described as matching logic rules between configuration patterns.
\end{itemize}
Like Hoare logic, correctness pair $\exists X(o=<\!\!c\!\!>_k<\!\!\rho\!\!>_{\mathit{env}}<\!\!m\!\!>_{\mathit{mem}}\land \varphi)\Downarrow \exists X(o=<\!\!\cdot\!\!>_k<\!\!\rho'\!\!>_{\mathit{env}}<\!\!m'\!\!>_{\mathit{mem}}\land\varphi')$ in matching logic relates configuration before the execution of $c$ to configuration after it's execution. Correctness pair should not be viewed as an independent object but as the result of a proof outline which carries intermediate proof information.
Matching logic has a lot of achievements\cite{rocsu2010matching}\cite{rocsu2012hoare}\cite{rosu2009rewriting}\cite{rosu2011matching}\cite{rosu2012checking}\cite{rocsu2020matching} in sequential imperative language (IMP). In \cite{rosu2009rewriting}, Grigore Ro{\c{s}}u et al presented matching logic proof system of IMP and proved the soundness of the matching logic proof system w.r.t. the operational semantics of IMP. However, when multiple processes execute in parallel, the results are complex and difficult to handle because the execution order of actions in different processes is unpredictable. A number of proof systems of PIMP have been proposed, such as temporal logic\cite{pnueli1977the} and Hoare logic of parallel program \cite{BF00268134}. Unfortunately, according to our investigation, there is no research on matching logic proof system for PIMP. \\
In this paper, following Grigore Ro{\c{s}}u et al's work\cite{rosu2009rewriting},  we consider matching logic for PIMP, which include parallelism in matching logic proof system of IMP. Matching logic for PIMP provides a simple and understandable way to deal with parallelism. More importantly, it is intuitive and suitable as a basis for a reliable proof outline. In \cite{BF00268134}, the definition of "interference-free" is key in characterizing parallel rule. However, this definition can't be used directly in matching logic. We had to redefine it so that it could be used in matching logic. Rewriting logic is the theoretical basis of $\mathbb{K}$, Ott and PLT-Redex.
we use rewriting logic semantics\cite{marti2002rewriting}\cite{meseguer1992conditional}\cite{meseguer2012twenty} to give operational semantics of PIMP and define
the operational semantics of PIMP as a rewrite theory. We also try to link PIMP's operational semantics and PIMP's verification formally. First, we prove the soundness of matching logic with respect to operational semantics of PIMP, and then matching logic verifier for PIMP is given which execute rewriting logic semantics symbolically on configuration patterns and is sound and complete for matching logic system of PIMP. Figure 1 shows the relationship among operational semantics, matching logic and matching logic verifier.\\
\begin{figure}
  \centering
  \includegraphics[width=3.5in,height=11in,clip,keepaspectratio]{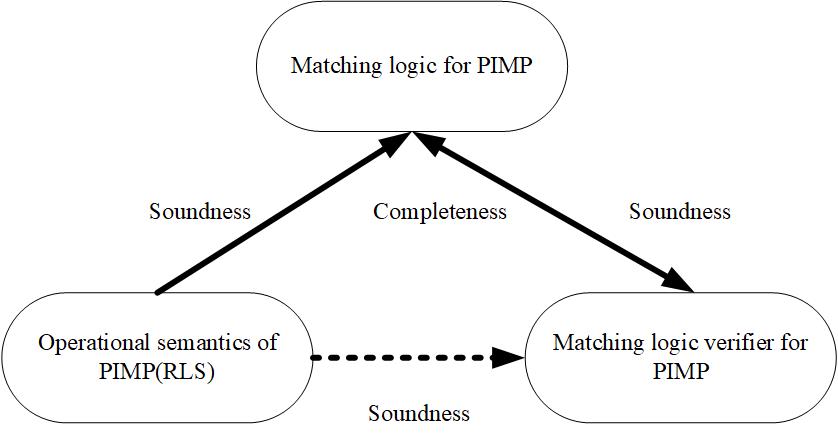}\\
  \caption{Relationship among operational semantics, matching logic and matching logic verifier}\label{fig1}
\end{figure}
The paper is organized as follows. The next section presents PIMP, a parallel imperative language with PAR and AWAIT operations to describing cooperation between processes-synchronization, mutual exclusion; with Env operation to describing a process is executed in an arbitrary "state", that is, in parallel with other processes. In addition, we define operational semantics of PIMP by rewriting logic and prove several useful properties. Section 3 introduces matching logic proof system for PIMP and prove the soundness of matching logic with respect to operational semantics of PIMP. Section 4 gives the matching logic verifier for PIMP and shows that it is sound and complete for matching logic proof system for PIMP. That is the matching logic verifier for PIMP is sound to the operational semantics of PIMP. We also verify the matching logic verifier through an example which is a standard problem in parallel programming.
\section{Operational semantics of PIMP}
In this section, we introduce a simple parallel imperative language(PIMP) by adding await operation, parallel operation and array operations to the IMP\cite{rosu2009rewriting}.
The operational semantics of PIMP is defined as a rewrite logic theory $(\Sigma_\mathrm{PIMP},\mathcal{E}_\mathrm{PIMP},\mathcal{R}_\mathrm{PIMP})$\cite{cserbuanuctua2009rewriting}. Figure 2 shows the complete rewrite theory. The signature $\Sigma_\mathrm{PIMP}$ consists of the PIMP's syntax and the syntax of configurations.
The configurations of PIMP have the form $<\!\!\cdots\!\!>_k<\!\!\cdots\!\!>_{env}<\!\!\cdots\!\!>_{mem}$, containing: a computation, an environment and a memory. The $\mathit{Env}$ sort is a partial mapping and has the form $x_1\mapsto i_1,x_2\mapsto i_2,\cdots,x_n\mapsto i_n$ with $\mathit{\cdot}$ representing the empty environment. Like $\mathit{Env}$ sort, the $\mathit{Mem}$ sort is also a partial map structure but from positive naturals to integers.
The $\mathit{C}$ sort is computation and $\mathit{\cdot}$ is empty computation with its usual properties : $\mathit{\cdot}\;;c=c$ and $\mathit{\cdot}\;\Vert c=c=c\Vert\;\mathit{\cdot}$. $\mathit{await }\;\mathit{b}\;\mathit{ then }\;\mathit{cc}$ computation acts as a flexible but primitive tool for mutual exclusion. Only when the condition $\mathit{b}$ is true, the process can execute $\mathit{await }\;\mathit{b}\;\mathit{ then }\;\mathit{cc}$, otherwise, the process is blocked. The computation $\mathit{cc}$ is an indivisible action. During its execution, other processes are blocked. Hence, it is very desirable to not contain $\Vert$, $\mathit{while}$ and $\mathit{await}$ computations in $\mathit{CC}$ sort.
$\mathit{A}:=\mathit{Array}(e_1, e_2, \cdots, e_n)$ evaluates $e_1, e_2, \cdots, e_n$ to $i_1, i_2, \cdots, i_n$, allocates a contiguous space with size $n$ and starting address as positive integer $p$ in memory, writes values $i_1, i_2, \cdots, i_n$ in order in that space and assigns $p$ to $A$. $x:=A[e]$ evaluates $e$ to positive integer $p$, evaluates $A$ to positive integer $q$ and assigns to $x$ the value at location $p+q$ in memory ($p+q$ must be allocated). $A[e_1]:=e_2$ evaluates $e_1$ to
positive integer $p$, $A$ to positive integer $q$ and $e_2$ to value $i$, writes $i$ to location $p+q$ in memory ($p+q$ must be allocated).\\
\begin{figure}
  \scriptsize
  \begin{tabular}{l}
    \\[-2mm]
    \hline
    \hline\\[-2mm]
    {\bf \small Rewrite Theory Of PIMP}\\
    \hline\\\\[-2mm]
    \vspace{1mm}
   \textbf{Abstract Syntax:}\\\\
$\mathit{E}::=0\mid1\mid2\mid\cdots\mid\mathit{Var}\mid-\mathit{E}\mid\mathit{E_1}+\mathit{E_2}\mid\mathit{E_1}-\mathit{E_2}
    \mid\mathit{E_1}*\mathit{E_2}\mid\mathit{E_1}\div\mathit{E_2}\mid\mathit{E_1} \;mod\;\mathit{E_2}$\\\\
 $\mathit{B}::=true\mid false\mid\mathit{E_1}=\mathit{E_2}\mid\mathit{E_1}\neq\mathit{E_2}\mid\mathit{E_1}<\mathit{E_2}\mid\mathit{E_1}>\mathit{E_2} $\\\\
$\mathit{CC}::=\mathit{Var}:=\mathit{E}\mid\mathit{if}\;(\mathit{B})\;\mathit{CC_1}\;\mathit{else }\;\mathit{CC_2}\mid \mathit{CC_1};\mathit{CC_2}\mid\mathit{Var}:= \mathit{Array}(\mathit{Seq}^{-,-}[E])\mid$\\
$\mathit{Var}[E_1]:=E_2\mid \mathit{Var_1}:=\mathit{Var_2}[E]\mid \mathit{skip}$\\\\
$\mathit{C}::=\mathit{Var}:=\mathit{E}\mid\mathit{if}\;(\mathit{B})\;\mathit{C_1}\;\mathit{else }\;\mathit{C_2}\mid
\mathit{C_1};\mathit{C_2}\mid\mathit{while}\; \mathit{B}\; \mathit{do}\; \mathit{C}\mid\mathit{C_1}\parallel\mathit{C_2}\mid\mathit{await }\;\mathit{B}\;\mathit{then}\;\mathit{CC}\mid$\\
$\mathit{skip}\mid\mathit{Var}:= \mathit{Array}(\mathit{Seq}^{-,-}[E])\mid \mathit{Var}[E_1]:=E_2\mid \mathit{Var_1}:=\mathit{Var_2}[E]$\\\\
\textbf{Configurtation:}\\\\
$\mathit{Cfg}::=<\!\!\mathit{C}\!\!>_k<\!\!\mathit{Env}\!\!>_{env}<\!\!\mathit{Mem}\!\!>_{mem}$ \\\\
$\mathit{Env}::=\mathit{Map}_.^{-,-}[\mathit{Var},\mathit{Int}]$ \\\\
$\mathit{Mem}::=\mathit{Map}_.^{-,-}[\mathit{Nat^+},\mathit{Int}]$ \\\\
\textbf{Semantic Rules:}\\\\
$\textrm{SKIP:}\dfrac{\cdot}{<\!\!skip\!\!>_k<\!\!\rho\!\!>_{env}<\!\!m\!\!>_{mem}\xrightarrow{P}<\!\!\cdot\!\!>_k<\!\!\rho\!\!>_{env}<\!\!m\!\!>_{mem}}$ \\\\
$\textrm{SEQ:}\dfrac{<\!\!c_1\!\!>_k<\!\!\rho\!\!>_{env}<\!\!m\!\!>_{mem}\xrightarrow{P}<\!\!c_1'\!\!>_k<\!\!\rho'\!\!>_{env}<\!\!m'\!\!>_{mem}}{<\!\!c_1;c_2\!\!>_k<\!\!\rho\!\!>_{env}<\!\!m\!\!>_{mem}\xrightarrow{P}<\!\!c_1';c_2\!\!>_k<\!\!\rho'\!\!>_{env}<\!\!m'\!\!>_{mem}}$\\
\\
$\textrm{ASGN1:}\dfrac{\cdot}{<\!\!x:=e\!\!>_k<\!\!\rho\!\!>_{env}<\!\!m\!\!>_{mem}\xrightarrow{P}<\!\!\cdot\!\!>_k<\!\!\rho[\rho(e)/x]\!\!>_{env}<\!\!m\!\!>_{mem}}$\\
\\
$\textrm{ASGN2:}\dfrac{\cdot}{<\!\!x:=A[e]\!\!>_k<\!\!\rho\!\!>_{env}<\!\!m\!\!>_{mem}\xrightarrow{P}<\!\!\cdot\!\!>_k<\!\!\rho[m(\rho(A)+_\mathit{Int}\rho(e))/x]\!\!>_{env}<\!\!m\!\!>_{mem}}$\\
\\
$\textrm{ASGN3:}\dfrac{\cdot}{<\!\!A[e_1]:=e_2\!\!>_k<\!\!\rho\!\!>_{env}<\!\!m\!\!>_{mem}\xrightarrow{P}<\!\!\cdot\!\!>_k<\!\!\rho\!\!>_{env}<\!\!m[\rho(e_2)/(\rho(A)+_\mathit{Int}\rho(e_1))]\!\!>_{mem}}$\\
\\
$\textrm{ARRAY:}\dfrac{\cdot}{<\!\!A:=\mathit{Array}(\overline{e})\!\!>_k<\!\!\rho\!\!>_{env}<\!\!m\!\!>_{mem}\xrightarrow{P}<\!\!\cdot\!\!>_k<\!\!\rho[p/A]\!\!>_{env}<\!\!p\mapsto[\rho(\overline{e})],m\!\!>_{mem}}$\\
\\
$\textrm{IF1:}\dfrac{\rho(b)\;\;is\;\;true}{<\!\!\mathit{if}\;(b)\;c_1\;\mathit{else}\;c_2\!\!>_k<\!\!\rho\!\!>_{env}<\!\!m\!\!>_{mem}\xrightarrow{P}<\!\!c_1\!\!>_k<\!\!\rho\!\!>_{env}<\!\!m\!\!>_{mem}}$\\
\\
$\textrm{IF2:}\dfrac{\rho(b)\;\;is\;\;false}{<\!\!\mathit{if }\;(b)\;c_1\;\mathit{else}\;c_2\!\!>_k<\!\!\rho\!\!>_{env}<\!\!m\!\!>_{mem}\xrightarrow{P}<\!\!c_2\!\!>_k<\!\!\rho\!\!>_{env}<\!\!m\!\!>_{mem}}$\\
\\
\textrm{WHILE1:}$\dfrac{\rho(b)\;\;is\;\;true}{<\!\!\mathit{while }\;b\;\mathit{ do }\;c\!\!>_k<\!\!\rho\!\!>_{env}<\!\!m\!\!>_{mem}\xrightarrow{P}<\!\!c;\mathit{while }\;b\;\mathit{ do }\;c\!\!>_k<\!\!\rho\!\!>_{env}<\!\!m\!\!>_{mem}}$\\
\\
\textrm{WHILE2:}$\dfrac{\rho(b)\;\;is\;\;false}{<\!\!\mathit{while }\;b\;\mathit{ do }\;c\!\!>_k<\!\!\rho\!\!>_{env}<\!\!m\!\!>_{mem}\xrightarrow{P}<\!\!\cdot\!\!>_k<\!\!\rho\!\!>_{env}<\!\!m\!\!>_{mem}}$\\
\\
\textrm{AWAIT:}$\dfrac{\rho(b)\;\;is\;\;true,<\!\!cc\!\!>_k<\!\!\rho\!\!>_{env}<\!\!m\!\!>_{mem}\xrightarrow{P}^*<\!\!\cdot\!\!>_k<\!\!\rho'\!\!>_{env}<\!\!m'\!\!>_{mem}}{<\!\!\mathit{await }\;b\;\mathit{ then } \; cc\!\!>_k<\!\!\rho\!\!>_{env}<\!\!m\!\!>_{mem}\xrightarrow{P}<\!\!\cdot\!\!>_k<\!\!\rho'\!\!>_{env}<\!\!m'\!\!>_{mem}}$\\
\\
\textrm{PAR1}$\dfrac{<\!\!c_1\!\!>_k<\!\!\rho\!\!>_{env}<\!\!m\!\!>_{mem}\xrightarrow{P}<\!\!c_1'\!\!>_k<\!\!\rho'\!\!>_{env}<\!\!m'\!\!>_{mem}}{<\!\!c_1\parallel c_2\!\!>_k<\!\!\rho\!\!>_{env}<\!\!m\!\!>_{mem}\xrightarrow{P}<\!\!c_1'\parallel c_2\!\!>_k<\!\!\rho'\!\!>_{env}<\!\!m'\!\!>_{mem}}$\\
\\
\textrm{PAR2}$\dfrac{<\!\!c_2\!\!>_k<\!\!\rho\!\!>_{env}<\!\!m\!\!>_{mem}\xrightarrow{P}<\!\!c_2'\!\!>_k<\!\!\rho'\!\!>_{env}<\!\!m'\!\!>_{mem}}{<\!\!c_1\parallel c_2\!\!>_k<\!\!\rho\!\!>_{env}<\!\!m\!\!>_{mem}\xrightarrow{P}<\!\!c_1\parallel c_2'\!\!>_k<\!\!\rho'\!\!>_{env}<\!\!m'\!\!>_{mem}}$\\
\\
\textrm{ENV}$\dfrac{\cdot}{<\!\!c\!\!>_k<\!\!\rho\!\!>_{env}<\!\!m\!\!>_{mem}\xrightarrow{E}<\!\!c\!\!>_k<\!\!\rho'\!\!>_{env}<\!\!m'\!\!>_{mem}}$\\
 \\
    \hline
    \hline
  \end{tabular}
  \caption{Rewrite Theory Of PIMP}\label{Fig1}
\end{figure}
The $\mathcal{E}_\mathrm{PIMP}$ contains equations which define bags, sequences and maps. We do not list these equations explicitly in Figure 2 because our main goal is to give PIMP an operational semantics in the rewriting logic framework.
The semantic rules\cite{reynolds2009theories} are the core of a rewrite semantics, and usually each language construct has at least one semantic rule.
$+, -, \ast, \div,\;\mathit{mod}$ constructs are reduced to the domain  $+_\mathrm{Int}, -_\mathrm{Int}, \ast_\mathrm{Int}, \div_\mathrm{Int}$, $mod_\mathrm{Int}$ when its arguments become integers and the $=, \neq, <, >$ constructs are reduced to the domain $=_\mathrm{Bool}, \neq_\mathrm{Bool}, <_\mathrm{Bool}, >_\mathrm{Bool}$ when its arguments become integers. $\mathit{PAR1}$ and $\mathit{PAR2}$ semantic rules use non-determinism to simulate parallelism, but they are defined in such a way that the results are equivalent to those which would be obtained using true parallelism. $\mathit{ARRAY}$ semantic rule chose some arbitrary positive integer $p$ such that ($p\mapsto[\rho(\overline{e})],m$) is a well-formed map and update
the environment and the memory accordingly where $p\mapsto[\rho(\overline{e})]$ is a shorthand for $p\mapsto\rho(e_1),p+1\mapsto\rho(e_2),\cdots,p+n-1\mapsto\rho(e_n)$ and $\overline{e}\in \mathit{Seq}^{-,-}[E]$ and $\rho(\overline{e})\in \mathit{Seq}^{-,-}[Int]$. \\
It has been suggested\cite{barringer1984now}\cite{abrahamson1979modal} that a computation should be thought of as being executed in an arbitrary "state", that is, in parallel with other computations. Therefore, there are two types of semantic rules, $<\!\!c\!\!>_k<\!\!\rho\!\!>_{env}<\!\!m\!\!>_{mem}\xrightarrow{P}<\!\!c'\!\!>_k<\!\!\rho'\!\!>_{env}<\!\!m'\!\!>_{mem}$ represents update of the environment $\rho$ to $\rho'$ and memory $m$ to $m'$ because $c$ is executed; $<\!\!c\!\!>_k<\!\!\rho\!\!>_{env}<\!\!m\!\!>_{mem}\xrightarrow{E}<\!\!c\!\!>_k<\!\!\rho'\!\!>_{env}<\!\!m'\!\!>_{mem}$ represents update of the environment $\rho$ to $\rho'$ and memory $m$ to $m'$ because other computation is executed which is in parallel with $c$. We can now formally define the operational semantics of PIMP as a rewrite logic theory.
\begin{definition}
Let PIMP denotes the rewriting logic theory $(\Sigma_\mathrm{PIMP},\mathcal{E}_\mathrm{PIMP},\mathcal{R}_\mathrm{PIMP})$ in Figure 2, $L=P$ or $L=E$, $\mathrm{PIMP}\models t\xrightarrow{L}t'$ indicates that $t\xrightarrow{L}t'$ can be derived. $\mathrm{PIMP}\models t\xrightarrow{L}^* t'$ indicates that the rule $t\xrightarrow{L} t'$ can be derived in zero or more steps.
\end{definition}
Computation $k$ is \emph{well-terminated} iff it is equal to an integer value or to "$\cdot$". Let $\mathrm{PIMP}^o$ be the algebraic specification $(\Sigma_\mathrm{PIMP},\mathcal{E}_\mathrm{PIMP}^o)$ where $\mathcal{E}_\mathrm{PIMP}^o\subseteq \mathcal{E}_\mathrm{PIMP}$ contains equations defining bags, sequences and maps.
Let $\mathcal{T}^o$ be the initial $\mathrm{PIMP}^o$ algebra. Terms $<\!\!k\!\!>_k<\!\!\rho\!\!>_{env}<\!\!m\!\!>_{mem}$ in $\mathcal{T}^o$ of sort $\mathit{Cfg}$ are called \emph{concrete configurations}. If $k$ is well-terminated, concrete configurations $<\!\!k\!\!>_k<\!\!\rho\!\!>_{env}<\!\!m\!\!>_{mem}$ are called \emph{final configurations}.
\begin{definition}
An execution of concrete configuration $<\!\!k_0\!\!>_k<\!\!\rho_0\!\!>_{env}<\!\!m_0\!\!>_{mem}$ is any finite or infinite sequence of the form: $\sigma\equiv<\!\!k_0\!\!>_k<\!\!\rho_0\!\!>_{env}<\!\!m_0\!\!>_{mem}\xrightarrow{L}<\!\!k_1\!\!>_k<\!\!\rho_1\!\!>_{env}<\!\!m_1\!\!>_{mem}\xrightarrow{L}\cdots\xrightarrow{L}<\!\!k_i\!\!>_k<\!\!\rho_i\!\!>_{env}<\!\!m_i\!\!>_{mem}\xrightarrow{L}<\!\!k_{i+1}\!\!>_k<\!\!\rho_{i+1}\!\!>_{env}<\!\!m_{i+1}\!\!>_{mem}\xrightarrow{L}\cdots$.
If the sequence is finite and there exits $j\in Int$, for all $j'\geq j$, $<\!\!k_{j'}\!\!>_k<\!\!\rho_{j'}\!\!>_{env}<\!\!m_{j'}\!\!>_{mem}$ is final configuration, we call $\sigma$ terminable, otherwise, $\sigma$ no terminable.
\end{definition}
For example $c_1\equiv x:=3; \mathit{if}(x>0)\; x:=x+1\;\mathit{else}\; x:=x-1$ and $c_2\equiv \mathit{await}\; x>0 \; \mathit{then}\; x:=10$, then a terminable execution of $<\!\!c_1\Vert c_2\!\!>_k<\!\!x\mapsto 0\!\!>_{env}<\!\!\cdot\!\!>_{mem}$ is :
$$<\!\!c_1\Vert c_2\!\!>_k<\!\!x\mapsto 0\!\!>_{env}<\!\!\cdot\!\!>_{mem}\xrightarrow{P}$$
$$<\!\!\mathit{if}(x>0) x:=x+1\;\mathit{else}\; x:=x-1\Vert\; \mathit{await}\; x>0 \; \mathit{then}\;x:=10\!\!>_k<\!\!x\mapsto 3\!\!>_{env}<\!\!\cdot\!\!>_{mem}\xrightarrow{P}$$
$$<\!\!\mathit{if}(x>0)\; x:=x+1\;\mathit{else}\; x:=x-1\Vert \cdot\!\!>_k<\!\!x\mapsto 10\!\!>_{env}<\!\!\cdot\!\!>_{mem}\xrightarrow{P}$$
$$<\!\! x:=x+1\!\!>_k<\!\!x\mapsto 10\!\!>_{env}<\!\!\cdot\!\!>_{mem}\xrightarrow{P}<\!\!\cdot\!\!>_k<\!\!x\mapsto 11\!\!>_{env}<\!\!\cdot\!\!>_{mem}$$
which can be broken down into two terminable executions:
$$<\!\!c_1\!\!>_k<\!\!x\mapsto 0\!\!>_{env}<\!\!\cdot\!\!>_{mem}\xrightarrow{P}<\!\!\mathit{if}(x>0) x:=x+1\;\mathit{else}\; x:=x-1\!\!>_k<\!\!x\mapsto 3\!\!>_{env}<\!\!\cdot\!\!>_{mem}$$
$$\xrightarrow{E}<\!\!\mathit{if}(x>0) x:=x+1\;\mathit{else}\; x:=x-1\!\!>_k<\!\!x\mapsto 10\!\!>_{env}<\!\!\cdot\!\!>_{mem}\xrightarrow{P}$$
$$<\!\! x:=x+1\!\!>_k<\!\!x\mapsto 10\!\!>_{env}<\!\!\cdot\!\!>_{mem}\xrightarrow{P}<\!\! \cdot\!\!>_k<\!\!x\mapsto 11\!\!>_{env}<\!\!\cdot\!\!>_{mem}$$
and
$$<\!\!c_2\!\!>_k<\!\!x\mapsto 0\!\!>_{env}<\!\!\cdot\!\!>_{mem}\xrightarrow{E}<\!\!c_2\!\!>_k<\!\!x\mapsto 3\!\!>_{env}<\!\!\cdot\!\!>_{mem}\xrightarrow{P}<\!\!\cdot\!\!>_k<\!\!x\mapsto 10\!\!>_{env}<\!\!\cdot\!\!>_{mem}$$
$$\xrightarrow{E}<\!\!\cdot\!\!>_k<\!\!x\mapsto 10\!\!>_{env}<\!\!\cdot\!\!>_{mem}\xrightarrow{E}<\!\!\cdot\!\!>_k<\!\!x\mapsto 11\!\!>_{env}<\!\!\cdot\!\!>_{mem}
$$
\begin{proposition}\label{prop1}
For any $c_1, c_2\in \mathit{C}$, if $\sigma$ is an execution of concrete configuration $<\!\!c_1\Vert c_2\!\!>_k<\!\!\rho\!\!>_{env}<\!\!m\!\!>_{mem}$ and $\sigma$ is terminable, then $\sigma$ can be broken down into two terminable executions $\sigma_1, \sigma_2$, which are executions of concrete configurations $<\!\!c_1\!\!>_k<\!\!\rho\!\!>_{env}<\!\!m\!\!>_{mem}$ and $<\!\! c_2\!\!>_k<\!\!\rho\!\!>_{env}<\!\!m\!\!>_{mem}$ respectively.
\end{proposition}
\begin{proof}
Suppose
$$\sigma\equiv <\!\!c_{10}\Vert c_{20}\!\!>_k<\!\!\rho_0\!\!>_{env}<\!\!m_0\!\!>_{mem}\xrightarrow{L}<\!\!c_{11}\Vert c_{21}\!\!>_k<\!\!\rho_1\!\!>_{env}<\!\!m_1\!\!>_{mem}\xrightarrow{L}\cdots\xrightarrow{L}$$
$$<\!\!c_{1i}\Vert c_{2i}\!\!>_k<\!\!\rho_i\!\!>_{env}<\!\!m_i\!\!>_{mem}\xrightarrow{L}\cdots\xrightarrow{L}<\!\!c_{1n}\Vert c_{2n}\!\!>_k<\!\!\rho_{n}\!\!>_{env}<\!\!m_n\!\!>_{mem}
$$
where $c_{10}=c_1$ and $c_{20}=c_2$ and $\rho_0=\rho$ and $m_0=m$.\\
For $1\leq i\leq n$, $<\!\!c_{1(i-1)}\Vert c_{2(i-1)}\!\!>_k<\!\!\rho_{i-1}\!\!>_{env}<\!\!m_{i-1}\!\!>_{mem}\xrightarrow{L}<\!\!c_{1i}\Vert c_{2i}\!\!>_k<\!\!\rho_i\!\!>_{env}<\!\!m_i\!\!>_{mem}$, if $L=P$, there are two ways in which the inference is done.\\
\textbf{Case1:} By PAR1 rule,
$$
\dfrac{<\!\!c_{1(i-1)}\!\!>_k<\!\!\rho_{i-1}\!\!>_{env}<\!\!m_{i-1}\!\!>_{mem}\xrightarrow{P}<\!\!c_{1i}\!\!>_k<\!\!\rho_i\!\!>_{env}<\!\!m_i\!\!>_{mem}}{<\!\!c_{1(i-1)}\Vert c_{2(i-1)}\!\!>_k<\!\!\rho_{i-1}\!\!>_{env}<\!\!m_{i-1}\!\!>_{mem}\xrightarrow{P}<\!\!c_{1i}\parallel c_{2(i-1)}\!\!>_k<\!\!\rho_i\!\!>_{env}<\!\!m_i\!\!>_{mem}}
$$
Set
$$<\!\!c_{1(i-1)}\!\!>_k<\!\!\rho_{i-1}\!\!>_{env}<\!\!m_{i-1}\!\!>_{mem}\xrightarrow{P}<\!\!c_{1i}\!\!>_k<\!\!\rho_i\!\!>_{env}<\!\!m_i\!\!>_{mem}$$
$$<\!\!c_{2(i-1)}\!\!>_k<\!\!\rho_{i-1}\!\!>_{env}<\!\!m_{i-1}\!\!>_{mem}\xrightarrow{E}<\!\!c_{2(i-1)}\!\!>_k<\!\!\rho_i\!\!>_{env}<\!\!m_i\!\!>_{mem}$$
\textbf{Case2:} By PAR2 rule,
$$
\dfrac{<\!\!c_{2(i-1)}\!\!>_k<\!\!\rho_{i-1}\!\!>_{env}<\!\!m_{i-1}\!\!>_{mem}\xrightarrow{P}<\!\!c_{2i}\!\!>_k<\!\!\rho_i\!\!>_{env}<\!\!m_i\!\!>_{mem}}{<\!\!c_{1(i-1)}\Vert c_{2(i-1)}\!\!>_k<\!\!\rho_{i-1}\!\!>_{env}<\!\!m_{i-1}\!\!>_{mem}\xrightarrow{P}<\!\!c_{1(i-1)}\parallel c_{2i}\!\!>_k<\!\!\rho_i\!\!>_{env}<\!\!m_i\!\!>_{mem}}
$$
Set
$$<\!\!c_{1(i-1)}\!\!>_k<\!\!\rho_{i-1}\!\!>_{env}<\!\!m_{i-1}\!\!>_{mem}\xrightarrow{E}<\!\!c_{1(i-1)}\!\!>_k<\!\!\rho_i\!\!>_{env}<\!\!m_i\!\!>_{mem}$$
$$<\!\!c_{2(i-1)}\!\!>_k<\!\!\rho_{i-1}\!\!>_{env}<\!\!m_{i-1}\!\!>_{mem}\xrightarrow{P}<\!\!c_{2i}\!\!>_k<\!\!\rho_i\!\!>_{env}<\!\!m_i\!\!>_{mem}$$
If $L=E$, set
$$<\!\!c_{1(i-1)}\!\!>_k<\!\!\rho_{i-1}\!\!>_{env}<\!\!m_{i-1}\!\!>_{mem}\xrightarrow{E}<\!\!c_{1(i-1)}\!\!>_k<\!\!\rho_i\!\!>_{env}<\!\!m_i\!\!>_{mem}$$
$$<\!\!c_{2(i-1)}\!\!>_k<\!\!\rho_{i-1}\!\!>_{env}<\!\!m_{i-1}\!\!>_{mem}\xrightarrow{E}<\!\!c_{2(i-1)}\!\!>_k<\!\!\rho_i\!\!>_{env}<\!\!m_i\!\!>_{mem}
$$
Hence, $\sigma$ can be broken down into two terminable executions $\sigma_1, \sigma_2$, which are executions of concrete configurations $<\!\!c_1\!\!>_k<\!\!\rho\!\!>_{env}<\!\!m\!\!>_{mem}$ and $<\!\! c_2\!\!>_k<\!\!\rho\!\!>_{env}<\!\!m\!\!>_{mem}$ respectively.
\end{proof}
\begin{lemma}
If $<\!\!c_1\!\!>_k<\!\!\rho\!\!>_{\mathit{env}}<\!\!m\!\!>_{\mathit{mem}}\xrightarrow{L}^n<\!\!c_1'\!\!>_k<\!\!\rho'\!\!>_{\mathit{env}}<\!\!m'\!\!>_{\mathit{mem}}$, then $<\!\!c_1;c_2\!\!>_k<\!\!\rho\!\!>_{\mathit{env}}<\!\!m\!\!>_{\mathit{mem}}\xrightarrow{L}^n<\!\!c_1';c_2\!\!>_k<\!\!\rho'\!\!>_{\mathit{env}}<\!\!m'\!\!>_{\mathit{mem}}$, $n\in Int $.
\end{lemma}
\begin{proof}
Suppose
$$
<\!\!c_1\!\!>_k<\!\!\rho\!\!>_{\mathit{env}}<\!\!m\!\!>_{\mathit{mem}}\xrightarrow{L}^n<\!\!c_1'\!\!>_k<\!\!\rho'\!\!>_{\mathit{env}}<\!\!m'\!\!>_{\mathit{mem}}, n\in Int
$$
If $n=0$, then
$$
<\!\!c_1\!\!>_k<\!\!\rho\!\!>_{\mathit{env}}<\!\!m\!\!>_{\mathit{mem}}\xrightarrow{L}^0<\!\!c_1\!\!>_k<\!\!\rho\!\!>_{\mathit{env}}<\!\!m\!\!>_{\mathit{mem}}
$$
Obviously,
$$
<\!\!c_1;c_2\!\!>_k<\!\!\rho\!\!>_{\mathit{env}}<\!\!m\!\!>_{\mathit{mem}}\xrightarrow{L}^0<\!\!c_1;c_2\!\!>_k<\!\!\rho\!\!>_{\mathit{env}}<\!\!m\!\!>_{\mathit{mem}}
$$
If $n\neq0$, we prove by induction on the length of execution path, that is induction on $n$.
$$
<\!\!c_1\!\!>_k<\!\!\rho\!\!>_{\mathit{env}}<\!\!m\!\!>_{\mathit{mem}}\xrightarrow{L}<\!\!c_1''\!\!>_k<\!\!\rho''\!\!>_{\mathit{env}}<\!\!m''\!\!>_{\mathit{mem}}\xrightarrow{L}^{n-1}<\!\!c_1'\!\!>_k<\!\!\rho'\!\!>_{\mathit{env}}<\!\!m'\!\!>_{\mathit{mem}}
$$
By the induction hypothesis of
$$
<\!\!c_1''\!\!>_k<\!\!\rho''\!\!>_{\mathit{env}}<\!\!m''\!\!>_{\mathit{mem}}\xrightarrow{L}^{n-1}<\!\!c_1'\!\!>_k<\!\!\rho'\!\!>_{\mathit{env}}<\!\!m'\!\!>_{\mathit{mem}}
$$
we conclude
$$
<\!\!c_1'';c_2\!\!>_k<\!\!\rho''\!\!>_{\mathit{env}}<\!\!m''\!\!>_{\mathit{mem}}\xrightarrow{L}^{n-1}<\!\!c_1';c_2\!\!>_k<\!\!\rho'\!\!>_{\mathit{env}}<\!\!m'\!\!>_{\mathit{mem}}
$$
If $L$ in $<\!\!c_1\!\!>_k<\!\!\rho\!\!>_{\mathit{env}}<\!\!m\!\!>_{\mathit{mem}}\xrightarrow{L}<\!\!c_1''\!\!>_k<\!\!\rho''\!\!>_{\mathit{env}}<\!\!m''\!\!>_{\mathit{mem}}$ is $P$, then SEQ rule implies
$$
<\!\!c_1;c_2\!\!>_k<\!\!\rho\!\!>_{\mathit{env}}<\!\!m\!\!>_{\mathit{mem}}\xrightarrow{L}<\!\!c_1'';c_2\!\!>_k<\!\!\rho''\!\!>_{\mathit{env}}<\!\!m''\!\!>_{\mathit{mem}}
$$
If $L$ in $<\!\!c_1\!\!>_k<\!\!\rho\!\!>_{\mathit{env}}<\!\!m\!\!>_{\mathit{mem}}\xrightarrow{L}<\!\!c_1''\!\!>_k<\!\!\rho''\!\!>_{\mathit{env}}<\!\!m''\!\!>_{\mathit{mem}}$ is $E$, then $c_1''=c_1$ and
$$
<\!\!c_1;c_2\!\!>_k<\!\!\rho\!\!>_{\mathit{env}}<\!\!m\!\!>_{\mathit{mem}}\xrightarrow{L}<\!\!c_1'';c_2\!\!>_k<\!\!\rho''\!\!>_{\mathit{env}}<\!\!m''\!\!>_{\mathit{mem}}
$$
Hence, $<\!\!c_1;c_2\!\!>_k<\!\!\rho\!\!>_{\mathit{env}}<\!\!m\!\!>_{\mathit{mem}}\xrightarrow{L}^n<\!\!c_1';c_2\!\!>_k<\!\!\rho'\!\!>_{\mathit{env}}<\!\!m'\!\!>_{\mathit{mem}}$.
\end{proof}
\begin{proposition} $<\!\!c_1;c_2\!\!>_k<\!\!\rho\!\!>_{\mathit{env}}<\!\!m\!\!>_{\mathit{mem}}\xrightarrow{L}^*<\!\!\cdot\!\!>_k<\!\!\rho'\!\!>_{\mathit{env}}<\!\!m'\!\!>_{\mathit{mem}}$ for some final configuration $<\!\!\cdot\!\!>_k<\!\!\rho'\!\!>_{\mathit{env}}<\!\!m'\!\!>_{\mathit{mem}}$ iff there exits some final configuration $<\!\!\cdot\!\!>_k<\!\!\rho''\!\!>_{\mathit{env}}<\!\!m''\!\!>_{\mathit{mem}}$ such that $<\!\!c_1\!\!>_k<\!\!\rho\!\!>_{\mathit{env}}<\!\!m\!\!>_{\mathit{mem}}\xrightarrow{L}^*<\!\!\cdot\!\!>_k<\!\!\rho''\!\!>_{\mathit{env}}<\!\!m''\!\!>_{\mathit{mem}}$ and $<\!\!c_2\!\!>_k<\!\!\rho''\!\!>_{\mathit{env}}<\!\!m''\!\!>_{\mathit{mem}}\xrightarrow{L}^*<\!\!\cdot\!\!>_k<\!\!\rho'\!\!>_{\mathit{env}}<\!\!m'\!\!>_{\mathit{mem}}$.
\end{proposition}
\begin{proof}
$(\Leftarrow)$\\
If there exit some final configuration $<\!\!\cdot\!\!>_k<\!\!\rho''\!\!>_{\mathit{env}}<\!\!m''\!\!>_{\mathit{mem}}$ such that
$$
<\!\!c_1\!\!>_k<\!\!\rho\!\!>_{\mathit{env}}<\!\!m\!\!>_{\mathit{mem}}\xrightarrow{L}^*<\!\!\cdot\!\!>_k<\!\!\rho''\!\!>_{\mathit{env}}<\!\!m''\!\!>_{\mathit{mem}}
$$
and
$$
<\!\!c_2\!\!>_k<\!\!\rho''\!\!>_{\mathit{env}}<\!\!m''\!\!>_{\mathit{mem}}\xrightarrow{L}^*<\!\!\cdot\!\!>_k<\!\!\rho'\!\!>_{\mathit{env}}<\!\!m'\!\!>_{\mathit{mem}}
$$
Since $<\!\!\cdot\!\!>_k<\!\!\rho''\!\!>_{\mathit{env}}<\!\!m''\!\!>_{\mathit{mem}}$ is final configuration, there exits $n\in Int$ such that $<\!\!c_1\!\!>_k<\!\!\rho\!\!>_{\mathit{env}}<\!\!m\!\!>_{\mathit{mem}}\xrightarrow{L}^n<\!\!\cdot\!\!>_k<\!\!\rho''\!\!>_{\mathit{env}}<\!\!m''\!\!>_{\mathit{mem}}$. By Lemma 1, we conclude
$$
<\!\!c_1;c_2\!\!>_k<\!\!\rho\!\!>_{\mathit{env}}<\!\!m\!\!>_{\mathit{mem}}\xrightarrow{L}^n<\!\!c_2\!\!>_k<\!\!\rho''\!\!>_{\mathit{env}}<\!\!m''\!\!>_{\mathit{mem}}
$$
Hence, $<\!\!c_1;c_2\!\!>_k<\!\!\rho\!\!>_{\mathit{env}}<\!\!m\!\!>_{\mathit{mem}}\xrightarrow{L}^*<\!\!\cdot\!\!>_k<\!\!\rho'\!\!>_{\mathit{env}}<\!\!m'\!\!>_{\mathit{mem}}$.\\
$(\Rightarrow)$
Suppose
$$<\!\!c_1;c_2\!\!>_k<\!\!\rho\!\!>_{\mathit{env}}<\!\!m\!\!>_{\mathit{mem}}\xrightarrow{E}^*<\!\!c_1;c_2\!\!>_k<\!\!\rho_1\!\!>_{\mathit{env}}<\!\!m_1\!\!>_{\mathit{mem}}\xrightarrow{P}\cdots\xrightarrow{L}^*<\!\!\cdot\!\!>_k<\!\!\rho'\!\!>_{\mathit{env}}<\!\!m'\!\!>_{\mathit{mem}}$$
our goal is to find $<\!\!\cdot\!\!>_k<\!\!\rho''\!\!>_{\mathit{env}}<\!\!m''\!\!>_{\mathit{mem}}$ which satisfy the property. we prove by structural induction on $c_1$.\\
\textbf{Case1:} $c_1\equiv \mathrm{x}:=e$, by SEQ and ASGN1 rules,
$$
<\!\!\mathrm{x}:=e;c_2\!\!>_k<\!\!\rho_1\!\!>_{\mathit{env}}<\!\!m_1\!\!>_{\mathit{mem}}\xrightarrow{P}<\!\!c_2\!\!>_k<\!\!\rho_1[\rho_1(e)/\mathrm{x}]\!\!>_{\mathit{env}}<\!\!m_1\!\!>_{\mathit{mem}}$$
$$\xrightarrow{L}^*
<\!\!\cdot\!\!>_k<\!\!\rho'\!\!>_{\mathit{env}}<\!\!m'\!\!>_{\mathit{mem}}
$$
Set $\rho''=\rho_1[\rho_1(e)/\mathrm{x}]$ and $m''=m_1$.\\
\textbf{Case2:} $c_1\equiv \mathrm{x}:=A[e]$, by SEQ and ASGN2 rules,
$$
<\!\!\mathrm{x}:=A[e];c_2\!\!>_k<\!\!\rho_1\!\!>_{\mathit{env}}<\!\!m_1\!\!>_{\mathit{mem}}\xrightarrow{P}$$
$$<\!\!c_2\!\!>_k<\!\!\rho_1[m_1(\rho_1(A)+_{Int}\rho_1(e))/\mathrm{x}]\!\!>_{\mathit{env}}<\!\!m_1\!\!>_{\mathit{mem}}\xrightarrow{L}^*<\!\!\cdot\!\!>_k<\!\!\rho'\!\!>_{\mathit{env}}<\!\!m'\!\!>_{\mathit{mem}}
$$
Set $\rho''=\rho_1[m_1(\rho_1(A)+_{Int}\rho_1(e))/\mathrm{x}]$ and $m''=m_1$.\\
\textbf{Case3:} $c_1\equiv A[e_1]:=e_2$, by SEQ and ASGN3 rules,
$$
<\!\!A[e_1]:=e_2;c_2\!\!>_k<\!\!\rho_1\!\!>_{\mathit{env}}<\!\!m_1\!\!>_{\mathit{mem}}\xrightarrow{P}$$
$$<\!\!c_2\!\!>_k<\!\!\rho_1\!\!>_{\mathit{env}}<\!\!m_1[\rho_1(e_2)/(\rho_1(A)+_{Int}\rho_1(e_1))]\!\!>_{\mathit{mem}}\xrightarrow{L}^*<\!\!\cdot\!\!>_k<\!\!\rho'\!\!>_{\mathit{env}}<\!\!m'\!\!>_{\mathit{mem}}
$$
Set $\rho''=\rho_1$ and $m''=m_1[\rho_1(e_2)/(\rho_1(A)+_{Int}\rho_1(e_1))]$.\\
\textbf{Case4:} $c_1\equiv A:=Array(\overline{e})$, by SEQ and ARRAY rules,
$$
<\!\!A:=Array(\overline{e});c_2\!\!>_k<\!\!\rho_1\!\!>_{\mathit{env}}<\!\!m_1\!\!>_{\mathit{mem}}\xrightarrow{P}<\!\!c_2\!\!>_k<\!\!\rho_1[p/A]\!\!>_{\mathit{env}}<\!\!p\mapsto[\rho_1(\overline{e})],m_1\!\!>_{\mathit{mem}}$$
$$\xrightarrow{L}^*<\!\!\cdot\!\!>_k<\!\!\rho'\!\!>_{\mathit{env}}<\!\!m'\!\!>_{\mathit{mem}}
$$
Set $\rho''=\rho_1[p/A]$ and $m''=p\mapsto[\rho_1(\overline{e})],m_1$.\\
\textbf{Case5:} $c_1\equiv \mathit{if}\;(\mathit{b})\;c_3\;\mathit{else}\;c_4$. If $\rho_1(b)$ is true, then by IF1 and SEQ rules,
$$
<\!\!\mathit{if}\;(\mathit{b})\;c_3\;\mathit{else}\;c_4;c_2\!\!>_k<\!\!\rho_1\!\!>_{\mathit{env}}<\!\!m_1\!\!>_{\mathit{mem}}\xrightarrow{P}<\!\!c_3;c_2\!\!>_k<\!\!\rho_1\!\!>_{\mathit{env}}<\!\!m_1\!\!>_{\mathit{mem}}\xrightarrow{L}^*$$
$$<\!\!\cdot\!\!>_k<\!\!\rho'\!\!>_{\mathit{env}}<\!\!m'\!\!>_{\mathit{mem}}
$$
Since $c_3$ is substructure of $c_1$, by the induction of
$$
<\!\!c_3;c_2\!\!>_k<\!\!\rho_1\!\!>_{\mathit{env}}<\!\!m_1\!\!>_{\mathit{mem}}\xrightarrow{L}^*<\!\!\cdot\!\!>_k<\!\!\rho'\!\!>_{\mathit{env}}<\!\!m'\!\!>_{\mathit{mem}}
$$
there exits some final configuration $<\!\!\cdot\!\!>_k<\!\!\rho_2\!\!>_{\mathit{env}}<\!\!m_2\!\!>_{\mathit{mem}}$ such that
$$
<\!\!c_3\!\!>_k<\!\!\rho_1\!\!>_{\mathit{env}}<\!\!m_1\!\!>_{\mathit{mem}}\xrightarrow{L}^*<\!\!\cdot\!\!>_k<\!\!\rho_2\!\!>_{\mathit{env}}<\!\!m_2\!\!>_{\mathit{mem}}$$
$$<\!\!c_2\!\!>_k<\!\!\rho_2\!\!>_{\mathit{env}}<\!\!m_2\!\!>_{\mathit{mem}}\xrightarrow{L}^*<\!\!\cdot\!\!>_k<\!\!\rho'\!\!>_{\mathit{env}}<\!\!m'\!\!>_{\mathit{mem}}
$$
Set $\rho''=\rho_2$ and $m''=m_2$.\\
If $\rho_1(b)$ is false, then by IF2 and SEQ rules,
$$
<\!\!\mathit{if}\;(\mathit{b})\;c_3\;\mathit{else}\;c_4;c_2\!\!>_k<\!\!\rho_1\!\!>_{\mathit{env}}<\!\!m_1\!\!>_{\mathit{mem}}\xrightarrow{P}<\!\!c_4;c_2\!\!>_k<\!\!\rho_1\!\!>_{\mathit{env}}<\!\!m_1\!\!>_{\mathit{mem}}\xrightarrow{L}^*$$
$$<\!\!\cdot\!\!>_k<\!\!\rho'\!\!>_{\mathit{env}}<\!\!m'\!\!>_{\mathit{mem}}
$$
Since $c_4$ is substructure of $c_1$, by the induction of
$$
<\!\!c_4;c_2\!\!>_k<\!\!\rho_1\!\!>_{\mathit{env}}<\!\!m_1\!\!>_{\mathit{mem}}\xrightarrow{L}^*<\!\!\cdot\!\!>_k<\!\!\rho'\!\!>_{\mathit{env}}<\!\!m'\!\!>_{\mathit{mem}}
$$
there exits some final configuration $<\!\!\cdot\!\!>_k<\!\!\rho_2\!\!>_{\mathit{env}}<\!\!m_2\!\!>_{\mathit{mem}}$ such that
$$
<\!\!c_4\!\!>_k<\!\!\rho_1\!\!>_{\mathit{env}}<\!\!m_1\!\!>_{\mathit{mem}}\xrightarrow{L}^*<\!\!\cdot\!\!>_k<\!\!\rho_2\!\!>_{\mathit{env}}<\!\!m_2\!\!>_{\mathit{mem}}$$
$$<\!\!c_2\!\!>_k<\!\!\rho_2\!\!>_{\mathit{env}}<\!\!m_2\!\!>_{\mathit{mem}}\xrightarrow{L}^*<\!\!\cdot\!\!>_k<\!\!\rho'\!\!>_{\mathit{env}}<\!\!m'\!\!>_{\mathit{mem}}
$$
Set $\rho''=\rho_2$ and $m''=m_2$.\\
\textbf{Case6:} $c_1\equiv \mathit{while }\;\mathit{b}\;\mathit{ do }\;c_3$. Since $<\!\!\cdot\!\!>_k<\!\!\rho'\!\!>_{\mathit{env}}<\!\!m'\!\!>_{\mathit{mem}}$ is final configuration, there exits $n\in Int$ such that
$$
<\!\!\mathit{while }\;\mathit{b}\;\mathit{ do }\;c_3;c_2\!\!>_k<\!\!\rho\!\!>_{\mathit{env}}<\!\!m\!\!>_{\mathit{mem}}\xrightarrow{L}^n<\!\!\cdot\!\!>_k<\!\!\rho'\!\!>_{\mathit{env}}<\!\!m'\!\!>_{\mathit{mem}}
$$
We prove the property by well-founded induction on $n$ for any $\rho$ and $m$. If $\rho_1(b)$ is true, by WHILE1 and SEQ rule,
$$
<\!\!\mathit{while }\;\mathit{b}\;\mathit{ do }\;c_3;c_2\!\!>_k<\!\!\rho_1\!\!>_{\mathit{env}}<\!\!m_1\!\!>_{\mathit{mem}}\xrightarrow{P}<\!\!c_3;\mathit{while }\;\mathit{b}\;\mathit{ do }\;c_3 ;c_2\!\!>_k<\!\!\rho_1\!\!>_{\mathit{env}}<\!\!m_1\!\!>_{\mathit{mem}}$$
$$\xrightarrow{L}^{n_1}<\!\!\cdot\!\!>_k<\!\!\rho'\!\!>_{\mathit{env}}<\!\!m'\!\!>_{\mathit{mem}}
$$
with $n_1<n$. Since $c_3$ is substructure of $c_1$, by the induction of
$$
<\!\!c_3;\mathit{while }\;\mathit{b}\;\mathit{ do }\;c_3 ;c_2\!\!>_k<\!\!\rho_1\!\!>_{\mathit{env}}<\!\!m_1\!\!>_{\mathit{mem}}\xrightarrow{L}^{n_1}<\!\!\cdot\!\!>_k<\!\!\rho'\!\!>_{\mathit{env}}<\!\!m'\!\!>_{\mathit{mem}}
$$
there exits some final configuration $<\!\!\cdot\!\!>_k<\!\!\rho_2\!\!>_{\mathit{env}}<\!\!m_2\!\!>_{\mathit{mem}}$ such that
$$
<\!\!c_3\!\!>_k<\!\!\rho_1\!\!>_{\mathit{env}}<\!\!m_1\!\!>_{\mathit{mem}}\xrightarrow{L}^*<\!\!\cdot\!\!>_k<\!\!\rho_2\!\!>_{\mathit{env}}<\!\!m_2\!\!>_{\mathit{mem}}$$
$$<\!\!\mathit{while }\;\mathit{b}\;\mathit{ do }\;c_3 ;c_2\!\!>_k<\!\!\rho_2\!\!>_{\mathit{env}}<\!\!m_2\!\!>_{\mathit{mem}}\xrightarrow{L}^{n_2}<\!\!\cdot\!\!>_k<\!\!\rho'\!\!>_{\mathit{env}}<\!\!m'\!\!>_{\mathit{mem}}
$$
with $n_2<n$. By the inner induction hypothesis ($n_2<n$), there exits some final configuration
$<\!\!\cdot\!\!>_k<\!\!\rho_3\!\!>_{\mathit{env}}<\!\!m_3\!\!>_{\mathit{mem}}$ such that
$$
<\!\!\mathit{while }\;\mathit{b}\;\mathit{ do }\;c_3\!\!>_k<\!\!\rho_2\!\!>_{\mathit{env}}<\!\!m_2\!\!>_{\mathit{mem}}\xrightarrow{L}^*<\!\!\cdot\!\!>_k<\!\!\rho_3\!\!>_{\mathit{env}}<\!\!m_3\!\!>_{\mathit{mem}}$$
$$<\!\!c_2\!\!>_k<\!\!\rho_3\!\!>_{\mathit{env}}<\!\!m_3\!\!>_{\mathit{mem}}\xrightarrow{L}^*<\!\!\cdot\!\!>_k<\!\!\rho'\!\!>_{\mathit{env}}<\!\!m'\!\!>_{\mathit{mem}}
$$
Set $\rho''=\rho_3$ and $m''=m_3$.\\
If $\rho_1(b)$ is false, then by WHILE2 and SEQ rule,
$$
<\!\!\mathit{while }\;\mathit{b}\;\mathit{ do }\;c_3;c_2\!\!>_k<\!\!\rho_1\!\!>_{\mathit{env}}<\!\!m_1\!\!>_{\mathit{mem}}\xrightarrow{P}<\!\!c_2\!\!>_k<\!\!\rho_1\!\!>_{\mathit{env}}<\!\!m_1\!\!>_{\mathit{mem}}\xrightarrow{L}^{n_1}
<\!\!\cdot\!\!>_k<\!\!\rho'\!\!>_{\mathit{env}}<\!\!m'\!\!>_{\mathit{mem}}
$$
Set $\rho''=\rho_1$ and $m''=m_1$.\\
\textbf{Case7:} $c_1\equiv \mathit{await}\;\mathit{b}\;\mathit{then}\;cc $. $<\!\!\cdot\!\!>_k<\!\!\rho'\!\!>_{\mathit{env}}<\!\!m'\!\!>_{\mathit{mem}}$ is final configuration, there exits $\rho_1,m_1$ such that $\rho_1(b)$ is true, then by SEQ and AWAIT rules,
$$
<\!\!\mathit{await}\;\mathit{b}\;\mathit{then}\;cc;c_2\!\!>_k<\!\!\rho_1\!\!>_{\mathit{env}}<\!\!m_1\!\!>_{\mathit{mem}}\xrightarrow{P}<\!\!c_2\!\!>_k<\!\!\rho_2\!\!>_{\mathit{env}}<\!\!m_2\!\!>_{\mathit{mem}}\xrightarrow{L}^*$$
$$<\!\!\cdot\!\!>_k<\!\!\rho'\!\!>_{\mathit{env}}<\!\!m'\!\!>_{\mathit{mem}}
$$
where $<\!\!cc\!\!>_k<\!\!\rho_1\!\!>_{\mathit{env}}<\!\!m_1\!\!>_{\mathit{mem}}\xrightarrow{P}^*<\!\!\cdot\!\!>_k<\!\!\rho_2\!\!>_{\mathit{env}}<\!\!m_2\!\!>_{\mathit{mem}}$.
Set $\rho''=\rho_2$ and $m''=m_2$.\\
\textbf{Case8:} $c_1\equiv c_3;c_4 $. Since $c_3$ is substructure of $c_1$,  by the induction of
$$
<\!\!c_3;(c_4;c_2)\!\!>_k<\!\!\rho\!\!>_{\mathit{env}}<\!\!m\!\!>_{\mathit{mem}}\xrightarrow{L}^*<\!\!\cdot\!\!>_k<\!\!\rho'\!\!>_{\mathit{env}}<\!\!m'\!\!>_{\mathit{mem}}
$$
there exits some  final configuration $<\!\!\cdot\!\!>_k<\!\!\rho_2\!\!>_{\mathit{env}}<\!\!m_2\!\!>_{\mathit{mem}}$ such that
$$
<\!\!c_3\!\!>_k<\!\!\rho\!\!>_{\mathit{env}}<\!\!m\!\!>_{\mathit{mem}}\xrightarrow{L}^*<\!\!\cdot\!\!>_k<\!\!\rho_2\!\!>_{\mathit{env}}<\!\!m_2\!\!>_{\mathit{mem}}$$
$$<\!\!c_4;c_2\!\!>_k<\!\!\rho_2\!\!>_{\mathit{env}}<\!\!m_2\!\!>_{\mathit{mem}}\xrightarrow{L}^*<\!\!\cdot\!\!>_k<\!\!\rho'\!\!>_{\mathit{env}}<\!\!m'\!\!>_{\mathit{mem}}
$$
Since $c_4$ is substructure of $c_1$, by the induction of
$$
<\!\!c_4;c_2\!\!>_k<\!\!\rho_2\!\!>_{\mathit{env}}<\!\!m_2\!\!>_{\mathit{mem}}\xrightarrow{L}^*<\!\!\cdot\!\!>_k<\!\!\rho'\!\!>_{\mathit{env}}<\!\!m'\!\!>_{\mathit{mem}}
$$
there exits some  final configuration $<\!\!\cdot\!\!>_k<\!\!\rho_3\!\!>_{\mathit{env}}<\!\!m_3\!\!>_{\mathit{mem}}$ such that
$$
<\!\!c_4\!\!>_k<\!\!\rho_2\!\!>_{\mathit{env}}<\!\!m_2\!\!>_{\mathit{mem}}\xrightarrow{L}^*<\!\!\cdot\!\!>_k<\!\!\rho_3\!\!>_{\mathit{env}}<\!\!m_3\!\!>_{\mathit{mem}}$$
$$<\!\!c_2\!\!>_k<\!\!\rho_3\!\!>_{\mathit{env}}<\!\!m_3\!\!>_{\mathit{mem}}\xrightarrow{L}^*<\!\!\cdot\!\!>_k<\!\!\rho'\!\!>_{\mathit{env}}<\!\!m'\!\!>_{\mathit{mem}}
$$
Set $\rho''=\rho_3$ and $m''=m_3$.\\
\textbf{Case9:} $c_1\equiv c_3\parallel c_4 $. According to the rule used in
$<\!\!(c_3\parallel c_4);c_2\!\!>_k<\!\!\rho_1\!\!>_{\mathit{env}}<\!\!m_1\!\!>_{\mathit{mem}}\xrightarrow{P}\gamma$, there are two cases:\\
\textbf{Case9.1:} By SEQ and PAR1 rules,
$$
<\!\!(c_3\parallel c_4);c_2\!\!>_k<\!\!\rho_1\!\!>_{\mathit{env}}<\!\!m_1\!\!>_{\mathit{mem}}\xrightarrow{P}<\!\!(c_3'\parallel c_4);c_2\!\!>_k<\!\!\rho_2\!\!>_{\mathit{env}}<\!\!m_2\!\!>_{\mathit{mem}}\xrightarrow{L}^*$$
$$<\!\!\cdot\!\!>_k<\!\!\rho'\!\!>_{\mathit{env}}<\!\!m'\!\!>_{\mathit{mem}}
$$
If $c_3\neq ((\mathit{while}\;\mathit{b}\;\mathit{do}\;c_5);c_6)\parallel c_7$, then $c_3'\parallel c_4$ is substructure of $c_1$, by the induction of
$$
<\!\!(c_3'\parallel c_4);c_2\!\!>_k<\!\!\rho_2\!\!>_{\mathit{env}}<\!\!m_2\!\!>_{\mathit{mem}}\xrightarrow{L}^*<\!\!\cdot\!\!>_k<\!\!\rho'\!\!>_{\mathit{env}}<\!\!m'\!\!>_{\mathit{mem}}
$$
there exits some final configuration $<\!\!\cdot\!\!>_k<\!\!\rho_3\!\!>_{\mathit{env}}<\!\!m_3\!\!>_{\mathit{mem}}$ such that
$$
<\!\!(c_3'\parallel c_4)\!\!>_k<\!\!\rho_2\!\!>_{\mathit{env}}<\!\!m_2\!\!>_{\mathit{mem}}\xrightarrow{L}^*<\!\!\cdot\!\!>_k<\!\!\rho_3\!\!>_{\mathit{env}}<\!\!m_3\!\!>_{\mathit{mem}}$$
$$<\!\!c_2\!\!>_k<\!\!\rho_3\!\!>_{\mathit{env}}<\!\!m_3\!\!>_{\mathit{mem}}\xrightarrow{L}^*<\!\!\cdot\!\!>_k<\!\!\rho'\!\!>_{\mathit{env}}<\!\!m'\!\!>_{\mathit{mem}}
$$
Set $\rho''=\rho_3$ and $m''=m_3$.\\
If $c_3=((\mathit{while}\;\mathit{b}\;\mathit{do}\;c_5);c_6)\parallel c_7$, then $c_3'=((c_5;\mathit{while}\;\mathit{b}\;\mathit{do}\;c_5);c_6)\parallel c_7$ and $\rho_2=\rho_1$ and $m_2=m_1$. Since $<\!\!\cdot\!\!>_k<\!\!\rho'\!\!>_{\mathit{env}}<\!\!m'\!\!>_{\mathit{mem}}$ is final configuration, there exits $<\!\!(c_3\parallel c_4');c_2 \!\!>_k<\!\!\rho_3\!\!>_{\mathit{env}}<\!\!m_3\!\!>_{\mathit{mem}}$ such that $c_4'$ is substructure of $c_4$ and
$$<\!\!(c_3'\parallel c_4);c_2\!\!>_k<\!\!\rho_2\!\!>_{\mathit{env}}<\!\!m_2\!\!>_{\mathit{mem}}\xrightarrow{L}^*<\!\!(c_3\parallel c_4');c_2\!\!>_k<\!\!\rho_3\!\!>_{\mathit{env}}<\!\!m_3\!\!>_{\mathit{mem}}\xrightarrow{L}^*$$
$$<\!\!\cdot\!\!>_k<\!\!\rho'\!\!>_{\mathit{env}}<\!\!m'\!\!>_{\mathit{mem}}
$$
By the induction of
$$
<\!\!(c_3\parallel c_4');c_2\!\!>_k<\!\!\rho_3\!\!>_{\mathit{env}}<\!\!m_3\!\!>_{\mathit{mem}}\xrightarrow{L}^*<\!\!\cdot\!\!>_k<\!\!\rho'\!\!>_{\mathit{env}}<\!\!m'\!\!>_{\mathit{mem}}
$$
there exits some final configuration $<\!\!\cdot\!\!>_k<\!\!\rho_4\!\!>_{\mathit{env}}<\!\!m_4\!\!>_{\mathit{mem}}$ such that
$$
<\!\!(c_3\parallel c_4')\!\!>_k<\!\!\rho_3\!\!>_{\mathit{env}}<\!\!m_3\!\!>_{\mathit{mem}}\xrightarrow{L}^*<\!\!\cdot\!\!>_k<\!\!\rho_4\!\!>_{\mathit{env}}<\!\!m_4\!\!>_{\mathit{mem}}$$
$$<\!\!c_2\!\!>_k<\!\!\rho_4\!\!>_{\mathit{env}}<\!\!m_4\!\!>_{\mathit{mem}}\xrightarrow{L}^*<\!\!\cdot\!\!>_k<\!\!\rho'\!\!>_{\mathit{env}}<\!\!m'\!\!>_{\mathit{mem}}
$$
Set $\rho''=\rho_4$ and $m''=m_4$.\\
\textbf{Case9.2:} By SEQ and PAR2 rules,
$$
<\!\!(c_3\parallel c_4);c_2\!\!>_k<\!\!\rho_1\!\!>_{\mathit{env}}<\!\!m_1\!\!>_{\mathit{mem}}\xrightarrow{P}<\!\!(c_3\parallel c_4');c_2\!\!>_k<\!\!\rho_2\!\!>_{\mathit{env}}<\!\!m_2\!\!>_{\mathit{mem}}\xrightarrow{L}^*$$
$$<\!\!\cdot\!\!>_k<\!\!\rho'\!\!>_{\mathit{env}}<\!\!m'\!\!>_{\mathit{mem}}
$$
If $c_4\neq ((\mathit{while}\;\mathit{b}\;\mathit{do}\;c_5);c_6)\parallel c_7$, then $c_3\parallel c_4'$ is substructure of $c_1$, by the induction of
$$
<\!\!(c_3\parallel c_4');c_2\!\!>_k<\!\!\rho_2\!\!>_{\mathit{env}}<\!\!m_2\!\!>_{\mathit{mem}}\xrightarrow{L}^*<\!\!\cdot\!\!>_k<\!\!\rho'\!\!>_{\mathit{env}}<\!\!m'\!\!>_{\mathit{mem}}
$$
there exits some final configuration $<\!\!\cdot\!\!>_k<\!\!\rho_3\!\!>_{\mathit{env}}<\!\!m_3\!\!>_{\mathit{mem}}$ such that
$$
<\!\!(c_3\parallel c_4')\!\!>_k<\!\!\rho_2\!\!>_{\mathit{env}}<\!\!m_2\!\!>_{\mathit{mem}}\xrightarrow{L}^*<\!\!\cdot\!\!>_k<\!\!\rho_3\!\!>_{\mathit{env}}<\!\!m_3\!\!>_{\mathit{mem}}$$
$$<\!\!c_2\!\!>_k<\!\!\rho_3\!\!>_{\mathit{env}}<\!\!m_3\!\!>_{\mathit{mem}}\xrightarrow{L}^*<\!\!\cdot\!\!>_k<\!\!\rho'\!\!>_{\mathit{env}}<\!\!m'\!\!>_{\mathit{mem}}
$$
Set $\rho''=\rho_3$ and $m''=m_3$.\\
If $c_4=((\mathit{while}\;\mathit{b}\;\mathit{do}\;c_5);c_6)\parallel c_7$, then $c_4'=((c_5;\mathit{while}\;\mathit{b}\;\mathit{do}\;c_5);c_6)\parallel c_7$ and $\rho_2=\rho_1$ and $m_2=m_1$. Since $<\!\!\cdot\!\!>_k<\!\!\rho'\!\!>_{\mathit{env}}<\!\!m'\!\!>_{\mathit{mem}}$ is final configuration, there exits $<\!\!(c_3'\parallel c_4);c_2 \!\!>_k<\!\!\rho_3\!\!>_{\mathit{env}}<\!\!m_3\!\!>_{\mathit{mem}}$ such that $c_3'$ is substructure of $c_3$ and
$$<\!\!(c_3\parallel c_4');c_2\!\!>_k<\!\!\rho_2\!\!>_{\mathit{env}}<\!\!m_2\!\!>_{\mathit{mem}}\xrightarrow{L}^*<\!\!(c_3'\parallel c_4);c_2\!\!>_k<\!\!\rho_3\!\!>_{\mathit{env}}<\!\!m_3\!\!>_{\mathit{mem}}\xrightarrow{L}^*$$
$$<\!\!\cdot\!\!>_k<\!\!\rho'\!\!>_{\mathit{env}}<\!\!m'\!\!>_{\mathit{mem}}
$$
By the induction of
$$
<\!\!(c_3'\parallel c_4);c_2\!\!>_k<\!\!\rho_3\!\!>_{\mathit{env}}<\!\!m_3\!\!>_{\mathit{mem}}\xrightarrow{L}^*<\!\!\cdot\!\!>_k<\!\!\rho'\!\!>_{\mathit{env}}<\!\!m'\!\!>_{\mathit{mem}}
$$
there exits some final configuration $<\!\!\cdot\!\!>_k<\!\!\rho_4\!\!>_{\mathit{env}}<\!\!m_4\!\!>_{\mathit{mem}}$ such that
$$
<\!\!(c_3\parallel c_4')\!\!>_k<\!\!\rho_3\!\!>_{\mathit{env}}<\!\!m_3\!\!>_{\mathit{mem}}\xrightarrow{L}^*<\!\!\cdot\!\!>_k<\!\!\rho_4\!\!>_{\mathit{env}}<\!\!m_4\!\!>_{\mathit{mem}}$$
$$<\!\!c_2\!\!>_k<\!\!\rho_4\!\!>_{\mathit{env}}<\!\!m_4\!\!>_{\mathit{mem}}\xrightarrow{L}^*<\!\!\cdot\!\!>_k<\!\!\rho'\!\!>_{\mathit{env}}<\!\!m'\!\!>_{\mathit{mem}}
$$
Set $\rho''=\rho_4$ and $m''=m_4$.\\
\textbf{Case10:} $c_1\equiv skip $, by SEQ and SKIP rules
$$<\!\!skip;c_2\!\!>_k<\!\!\rho\!\!>_{\mathit{env}}<\!\!m\!\!>_{\mathit{mem}}\xrightarrow{P}<\!\!c_2\!\!>_k<\!\!\rho\!\!>_{\mathit{env}}<\!\!m\!\!>_{\mathit{mem}}\xrightarrow{L}^*<\!\!\cdot\!\!>_k<\!\!\rho'\!\!>_{\mathit{env}}<\!\!m'\!\!>_{\mathit{mem}}$$
Set $\rho''=\rho$ and $m''=m$.
\end{proof}
\section{ Matching logic for PIMP}
In matching logic, one can't quantify program variables because they are syntactic constants rather than logical variables.
Suppose $\mathit{SVar}$ is an infinite set of logical or semantic variables, and it also contain a special variable named "$\mathit{o}$" of sort $\mathit{Cfg}$ which serves as a place holder in the matching logic pattern.
\begin{definition}
\cite{rosu2009rewriting}Matching logic patterns, are $\mathrm{FOL}_=$ forumlae $\exists X((o=c)\land \varphi)$, where: $X \subset \mathit{SVar}$ is the set of bound variables; $c$ is the pattern structure and is a term of sort $\mathit{Cfg}$; $\varphi$ is the constraint, an arbitrary $\mathit{FOL}_=$ formula.
\end{definition}
Let $\mathcal{T}$ is the initial model of $\mathrm{PIMP}$ and $\mathit{SVar}^o=\mathit{SVar}\cup \{o\}$. Valuation $(\gamma,\tau):\mathit{SVar}^o \to \mathcal{T}$ includes a concrete configuration $\gamma$ and a map $\tau:\mathit{SVar}\to \mathcal{T}$. $(\gamma,\tau)\models \exists X((o=c)\land \varphi)$ iff there exists $\theta_\tau:\mathit{SVar}\to \mathcal{T}$ with $\theta_\tau\!\!\upharpoonright_{\mathit{SVar}/X}=\tau\!\!\upharpoonright_{\mathit{SVar}/X}$ such that $\gamma=\theta_\tau(c)$ and $\theta_\tau \models \varphi$. Let $\Gamma,\Gamma'$ are matching logic patterns, $\Gamma\Downarrow\Gamma'$ is called matching logic correctness pair.
\begin{figure}
  \scriptsize
  \begin{tabular}{l}
    \\[-2mm]
    \hline
    \hline\\[-2mm]
    {\bf \small Matching logic proof system of PIMP}\\
    \hline\\[-2mm]\\
    \textrm{M-SKIP:}\\
    $\dfrac{\cdot}{\exists X(o=<\!\!skip\!\!>_k<\!\!\rho\!\!>_{\mathit{env}}<\!\!m\!\!>_{\mathit{mem}}\land \varphi)\Downarrow\exists X(o=<\!\!\cdot\!\!>_k<\!\!\rho\!\!>_{\mathit{env}}<\!\!m\!\!>_{\mathit{mem}}\land \varphi)}$\\\\\\
    \textrm{M-ASGN1:}\\
    $\dfrac{\cdot}{\exists X(o=<\!\!\mathrm{x}:=e\!\!>_k<\!\!\rho\!\!>_{\mathit{env}}<\!\!m\!\!>_{\mathit{mem}}\land \varphi)\Downarrow\exists  X(o=\!<\!\!\cdot\!\!>_k<\!\!\rho[\rho(e)/\mathrm{x}]\!\!>_{\mathit{env}}<\!\!m\!\!>_{\mathit{mem}}\land \varphi)}$\\\\\\
    \textrm{M-ASGN2:}\\
    $\dfrac{\cdot}{\exists X(o=<\!\!\mathrm{x}:=A[e]\!\!>_k<\!\!\rho\!\!>_{\mathit{env}}<\!\!m\!\!>_{\mathit{mem}}\land \varphi)\Downarrow\exists  X(o=\!<\!\!\cdot\!\!>_k<\!\!\rho[m(\rho(A)\!+_{Int}\!\rho(e))/\mathrm{x}]\!\!>_{\mathit{env}}<\!\!m\!\!>_{\mathit{mem}}\land \varphi)}$\\
    \\\\
    \textrm{M-ASGN3:}\\
    $\dfrac{\cdot}{\exists X(o=<\!\!A[e_1]:=e_2\!\!>_k<\!\!\rho\!\!>_{\mathit{e\!n\!v}}<\!\!m\!\!>_{\mathit{m\!e\!m}}\!\land \varphi)\!\Downarrow\!\exists  X(o=\!<\!\!\cdot\!\!>_k<\!\!\rho\!\!>_{\mathit{e\!n\!v}}<\!\!m[\rho(e_2)/(\rho(A)\!+_{I\!n\!t}\!\rho(e_1))]\!\!>_{\mathit{m\!e\!m}}\!\land \varphi)}$\\\\\\
    \textrm{M-ARRAY:}\\
    $\dfrac{\cdot}{\exists X(o\!=\!<\!\!A\!:=\!\!Array(\overline{e})\!\!>_k<\!\!\rho\!\!>_{\mathit{e\!n\!v}}\!<\!\!m\!\!>_{\mathit{m\!e\!m}}\!\land\! \varphi)\!\!\Downarrow\!\exists  (X\!\cup\! \{p\})(o\!=\!<\!\!\cdot\!\!>_k\!<\!\!\rho[p/A]\!\!>_{\mathit{e\!n\!v}}\!<\!\!p\!\mapsto\![\rho(\overline{e})],m\!\!>_{\mathit{m\!e\!m}}\!\land \!\varphi)}$\\\\\\
    \textrm{M-SEQ:}\\\\
    $\;\;\;\exists X_1(o=<\!\!c_1\!\!>_k<\!\!\rho_1\!\!>_{\mathit{env}}<\!\!m_1\!\!>_{\mathit{mem}}\land \varphi_1)\Downarrow\exists X_2(o=<\!\!\cdot\!\!>_k<\!\!\rho_2\!\!>_{\mathit{env}}<\!\!m_2\!\!>_{\mathit{mem}}\land \varphi_2)$\\
$\dfrac{\exists X_2(o=<\!\!c_2\!\!>_k<\!\!\rho_2\!\!>_{\mathit{env}}<\!\!m_2\!\!>_{\mathit{mem}}\land \varphi_2)\Downarrow\exists  X_3(o=<\!\!\cdot\!\!>_k<\!\!\rho_3\!\!>_{\mathit{env}}<\!\!m_3\!\!>_{\mathit{mem}}\land \varphi_3)}{\exists X_1(o=<\!\!c_1;c_2\!\!>_k<\!\!\rho_1\!\!>_{\mathit{env}}<\!\!m_1\!\!>_{\mathit{mem}}\land \varphi_1)\Downarrow\exists X_3(o=<\!\!\cdot\!\!>_k<\!\!\rho_3\!\!>_{\mathit{env}}<\!\!m_3\!\!>_{\mathit{mem}}\land \varphi_3)}$\\\\\\
\textrm{M-IF:}\\\\
$\;\exists X_1(o=<\!\!c_1\!\!>_k<\!\!\rho\!\!>_{\mathit{env}}<\!\!m\!\!>_{\mathit{mem}}\land \varphi \land (\rho(b)\;\;is\;\; true))\Downarrow\exists X_2(o=<\!\!\cdot\!\!>_k<\!\!\rho'\!\!>_{\mathit{env}}<\!\!m'\!\!>_{\mathit{mem}}\land \varphi')$\\
$\dfrac{\exists X_1(o=<\!\!c_2\!\!>_k<\!\!\rho\!\!>_{\mathit{env}}<\!\!m\!\!>_{\mathit{mem}}\land \varphi \land(\rho(b)\;\;is\;\; false))\Downarrow\exists X_2(o=<\!\!\cdot\!\!>_k<\!\!\rho'\!\!>_{\mathit{env}}<\!\!m'\!\!>_{\mathit{mem}}\land \varphi')}{\exists X_1(o=<\!\!\textrm{if }(b)c_1\textrm{ else }c_2 \!\!>_k<\!\!\rho\!\!>_{\mathit{env}}<\!\!m\!\!>_{\mathit{mem}}\land \varphi)\Downarrow\exists X_2(o=<\!\!\cdot\!\!>_k<\!\!\rho'\!\!>_{\mathit{env}}<\!\!m'\!\!>_{\mathit{mem}}\land \varphi')}$\\\\\\
\textrm{M-CONS:}\\\\
$\dfrac{\models \Gamma \Rightarrow \Gamma_1, \Gamma_1\Downarrow \Gamma_1',\models \Gamma_1' \Rightarrow \Gamma'}{\Gamma\Downarrow \Gamma'}$\\\\\\
\textrm{M-CASE:}\\\\
$\dfrac{\models \Gamma \Rightarrow \Gamma_1\vee \Gamma_2, \Gamma_1\Downarrow \Gamma', \Gamma_2\Downarrow \Gamma'}{\Gamma\Downarrow \Gamma'}$\\\\\\
\textrm{M-AWAIT:}\\\\
$\dfrac{\exists  X_1(o=<\!\!cc\!\!>_k<\!\!\rho\!\!>_{\mathit{env}}<\!\!m\!\!>_{\mathit{mem}}\land \varphi  \land(\rho(b)\;\;is\;\; true))\Downarrow\exists X_2(o=<\!\!\cdot\!\!>_k<\!\!\rho'\!\!>_{\mathit{env}}<\!\!m'\!\!>_{\mathit{mem}}\land  \varphi')}{\exists X_1(o=<\!\!\textrm{await }b\textrm{ then }\mathit{cc}\!\! >_k<\!\!\rho\!\!>_{\mathit{env}}<\!\!m\!\!>_{\mathit{mem}}\land \varphi)\Downarrow\exists X_2(o=<\!\!\cdot\!\!>_k<\!\!\rho'\!\!>_{\mathit{env}}<\!\!m'\!\!>_{\mathit{mem}}\land \varphi')}$\\\\\\
\textrm{M-WHILE:}\\\\
$\dfrac{\exists X(o=<\!\!c\!\!>_k<\!\!\rho\!\!>_{\mathit{env}}<\!\!m\!\!>_{\mathit{mem}}\land \varphi \land(\rho(b)\;\;is\;\; true))\Downarrow\exists X(o=<\!\!\cdot\!\!>_k<\!\!\rho\!\!>_{\mathit{env}}<\!\!m\!\!>_{\mathit{mem}}\land \varphi)}{\exists X(o=<\!\!\textrm{while }b\textrm{ do } c\!\!>_k<\!\!\rho\!\!>_{\mathit{env}}<\!\!m\!\!>_{\mathit{mem}}\land \varphi)\Downarrow\exists X(o=<\!\!\cdot\!\!>_k<\!\!\rho\!\!>_{\mathit{env}}<\!\!m\!\!>_{\mathit{mem}}\land \varphi\land(\rho(b)\;\;is\;\; false))}$  \\\\\\
\textrm{M-PAR:}\\\\
$\exists X_1(o=<\!\!c_1\!\!>_k<\!\!\rho\!\!>_{\mathit{env}}<\!\!m\!\!>_{\mathit{mem}}\land \varphi_1)\Downarrow\exists X_2(o=<\!\!\cdot\!\!>_k<\!\!\rho'\!\!>_{\mathit{env}}<\!\!m'\!\!>_{\mathit{mem}}\land \varphi_1')$\\
$\exists X_1(o=<\!\!c_2\!\!>_k<\!\!\rho\!\!>_{\mathit{env}}<\!\!m\!\!>_{\mathit{mem}}\land \varphi_2)\Downarrow\exists X_2(o=<\!\!\cdot\!\!>_k<\!\!\rho'\!\!>_{\mathit{env}}<\!\!m'\!\!>_{\mathit{mem}}\land\varphi_2')$\\
$c_1,\;c_2\; \textrm{are interference-free}$\\
$\overline{\exists X_1(o=<\!\!c_1\parallel c_2\!\!>_k<\!\!\rho\!\!>_{env}<\!\!m\!\!>_{\mathit{mem}}\land \varphi_1 \land \varphi_2)\Downarrow\exists X_2(o=<\!\!\cdot\!\!>_k<\!\!\rho'\!\!>_{env}<\!\!m'\!\!>_{\mathit{mem}}\land\varphi_1'\land \varphi_2')}$\\\\\\
\hline
\hline
\end{tabular}
  \caption{ Matching logic proof system of PIMP}\label{fig2}
\end{figure}
Informally matching logic correctness pair $\exists X(o=<\!\!c\!\!>_k<\!\!\rho\!\!>_{\mathit{env}}<\!\!m\!\!>_{mem}\land \varphi)\Downarrow \exists X(o=<\!\!\cdot\!\!>_k<\!\!\rho'\!\!>_{\mathit{env}}<\!\!m'\!\!>_{mem}\land\varphi')$ means: for any valuation $(\gamma,\tau):\mathit{SVar}^o \to \mathcal{T}$, if $(\gamma,\tau) \models \exists X(o=<\!\!c\!\!>_k<\!\!\rho\!\!>_{\mathit{env}}<\!\!m\!\!>_{mem}\land \varphi)$ before execution of $\gamma$, then $(\gamma',\tau) \models \exists X(o=<\!\!\cdot\!\!>_k<\!\!\rho'\!\!>_{\mathit{env}}<\!\!m'\!\!>_{mem}\land\varphi')$ after execution of $\gamma$, where  $\mathit{PIMP} \models \gamma\to^*\gamma'$ and $\gamma'$ is a final configuration.
We introduce syntax shorthand notations for configuration pattern and correctness pair:
$$<\!\!c\!\!>_k<\!\!\rho\!\!>_{env}<\!\!m\!\!>_{mem}<\!\!X\!\!>_{bnd}<\!\!\varphi\!\!>_{form}$$
instead of
$$\exists X((o=<\!\!c\!\!>_k<\!\!\rho\!\!>_{env}<\!\!m\!\!>_{mem})\land \varphi)$$
and
$$<\!\!\rho\!\!>_{env}<\!\!m\!\!>_{mem}<\!\!X\!\!>_{bnd}<\!\!\varphi\!\!>_{form}\;c\;<\!\!\rho'\!\!>_{env}<\!\!m'\!\!>_{mem}<\!\!X'\!\!>_{bnd}<\!\!\varphi'\!\!>_{form}$$ instead of
$$\exists X(o=<\!\!c\!\!>_k<\!\!\rho\!\!>_{\mathit{env}}<\!\!m\!\!>_{mem}\land \varphi)\Downarrow \exists X'(o=<\!\!\cdot\!\!>_k<\!\!\rho'\!\!>_{\mathit{env}}<\!\!m'\!\!>_{mem}\land\varphi')$$
we call $<\!\!\rho\!\!>_{env}<\!\!m\!\!>_{mem}<\!\!X\!\!>_{bnd}<\!\!\varphi\!\!>_{form}$ and $<\!\!\rho'\!\!>_{env}<\!\!m'\!\!>_{mem}<\!\!X'\!\!>_{bnd}<\!\!\varphi'\!\!>_{form}$ assertions.
Figure 3 gives the matching logic proof system for PIMP.
Notice that M-PAR rule in Figure 3 says that as long as $c_1$, $c_2$ don't interfere with each other, the effect of executing $c_1$ and $c_2$ in parallel is the same as executing $c_1$ and $c_2$ separately. The easiest way to get "interference-free" is not to allow shared variables, but this is too restrictive to handle the synchronization of producer and consumer which is a standard problem in parallel programming literature.\\
Suppose $<\!\!\rho_1\!\!>_{env}<\!\!m_1\!\!>_{mem}<\!\!X_1\!\!>_{bnd}<\!\!\varphi_1\!\!>_{form}\;c_1\;<\!\!\rho_1'\!\!>_{env}<\!\!m_1'\!\!>_{mem}<\!\!X_1'\!\!>_{bnd}<\!\!\varphi_1'\!\!>_{form}$ , set $\mathit{pre}(c_1)=<\!\!\rho_1\!\!>_{env}<\!\!m_1\!\!>_{mem}<\!\!X_1\!\!>_{bnd}<\!\!\varphi_1\!\!>_{form}$ and $\mathit{post}(c_1)=<\!\!\rho_1'\!\!>_{env}<\!\!m_1'\!\!>_{mem}<\!\!X_1'\!\!>_{bnd}<\!\!\varphi_1'\!\!>_{form}$, we now define "interference-free".
\begin{definition}
Given a proof outline of $\mathit{pre}(c_1)\;c_1\;\mathit{post}(c_1)$ and a computation $k$ with $\mathit{pre}(k)\;k\;\mathit{post}(k)$, we say that $k$ don't interfere with $c_1$ if the following two
conditions hold:
\begin{enumerate}
  \item For any concrete configuration $\gamma$, if $(\gamma,\tau)\models <\!\!k\!\!>_k \mathit{pre}(k)$ and $(\gamma,\tau)\models <\!\!k\!\!>_k \mathit{post}(c_1)$ and $\gamma\rightarrow^*\gamma'$ and $\gamma'$ is a final configuration, then $(\gamma',\tau)\models <\!\!\cdot\!\!>_k \mathit{post}(c_1)$;
  \item Let $c_1'$ is any sub computation within $c_1$ but not within any $\mathit{await}$, for any concrete configuration $\gamma$, if $(\gamma,\tau)\models <\!\!k\!\!>_k \mathit{pre}(k)$ and $(\gamma,\tau)\models <\!\!k\!\!>_k \mathit{pre}(c_1')$ and $\gamma\rightarrow^*\gamma'$ and $\gamma'$ is a final configuration, then $(\gamma',\tau)\models <\!\!\cdot\!\!>_k \mathit{pre}(c_1')$.
\end{enumerate}
\end{definition}
\begin{definition}
Given proof outlines of $\mathit{pre}(c_1)\;c_1\;\mathit{post}(c_1)$ and $\mathit{pre}(c_2)\;c_2\;\mathit{post}(c_2)$, we say that $c_1$, $c_2$ are "interference-free" if the following two conditions hold:
\begin{enumerate}
  \item Let $c_1'$ be an $\mathit{await}$ or $:=$ sub computation (which don't
appear in an $\mathit{await}$) of $c_1$, then $c_1'$ don't interfere with $c_2$;
  \item Let $c_2'$ be an $\mathit{await}$ or $:=$ sub computation (which don't
appear in an $\mathit{await}$) of $c_2$, then $c_2'$ don't interfere with $c_1$.
\end{enumerate}
\end{definition}
\begin{definition}
Given a proof outline of $pre(c)\; c\; post(c)$ and any concrete configuration $\gamma_0$  with $(\gamma_0,\tau) \models <\!\!c\!\!>_k pre(c)$, there are infinite number of terminable executions $\sigma=\gamma_0\xrightarrow{L}\gamma_1\xrightarrow{L}\cdots\xrightarrow{L}\gamma_i\xrightarrow{L}\gamma_{i+1}\xrightarrow{L}\cdots\xrightarrow{L}\gamma_n$ because of ENV rule. For any $\gamma_i\xrightarrow{L}\gamma_{i+1},0\leq i < n$, if $L=E$, it means other computation $c'$ which is in parallel with $c$ updates environment and memory. Let $S$ represents a set of computations executed in parallel with $c$. A terminable execution $\sigma$ is called \textbf{\emph{actual}} execution if $c', c$ are "interference-free" for any $c'\in S$.\\
\end{definition}
We now formally give the proof that matching logic is soundness to the
operational semantics of PIMP. Let's make the assumption that the original PIMP program don't contain variables in $\mathit{SVar}$.
\begin{theorem}
For any actual execution $\sigma=\gamma_0\xrightarrow{L}\gamma_1\xrightarrow{L}\cdots\xrightarrow{L}\gamma_i\xrightarrow{L}\gamma_{i+1}\xrightarrow{L}\cdots\xrightarrow{L}\gamma_n$, if $(\gamma_0,\tau) \models \Gamma$ and $\Gamma\Downarrow \Gamma'$ is derivable, then $(\gamma_n,\tau) \models \Gamma'$.
\end{theorem}
\begin{proof}
We prove by the induction on the depth of inference of $\Gamma\Downarrow \Gamma'$. We consider the different ways in which the last step of the inference is done:\\
\textbf{Case1:} By M-ASGN1 rule
$$\dfrac{\cdot}{\exists X(o=<\!\!\mathrm{x}:=e\!\!>_k<\!\!\rho\!\!>_{\mathit{env}}<\!\!m\!\!>_{\mathit{mem}}\land \varphi)\Downarrow\exists  X(o=\!<\!\!\cdot\!\!>_k<\!\!\rho[\rho(e)/\mathrm{x}]\!\!>_{\mathit{env}}<\!\!m\!\!>_{\mathit{mem}}\land \varphi)}
$$
$\Gamma\equiv\exists X(o=<\!\!\mathrm{x}:=e\!\!>_k<\!\!\rho\!\!>_{\mathit{env}}<\!\!m\!\!>_{\mathit{mem}}\land \varphi)$ and $\Gamma'\equiv\exists  X(o=\!<\!\!\cdot\!\!>_k<\!\!\rho[\rho(e)/\mathrm{x}]\!\!>_{\mathit{env}}<\!\!m\!\!>_{\mathit{mem}}\land \varphi)$. Suppose there exists $\theta_\tau:\mathit{SVar}\to \mathcal{T}$ with $\theta_\tau\!\!\upharpoonright_{\mathit{SVar}/X}=\tau\!\!\upharpoonright_{\mathit{SVar}/X}$ such that $\gamma_0=<\!\!\mathrm{x}:=e\!\!>_k<\!\!\theta_\tau(\rho)\!\!>_{\mathit{env}}<\!\!\theta_\tau(m)\!\!>_{\mathit{mem}}$  and $\theta_\tau \models \varphi$ and
$$<\!\!\mathrm{x}:=e\!\!>_k<\!\!\theta_\tau(\rho)\!\!>_{\mathit{env}}<\!\!\theta_\tau(m)\!\!>_{\mathit{mem}}\xrightarrow{L}\gamma_1\xrightarrow{L}\cdots\xrightarrow{L}\gamma_i\xrightarrow{L}\gamma_{i+1}\xrightarrow{L}\cdots\xrightarrow{L}\gamma_n
$$
If $<\!\!\mathrm{x}:=e\!\!>_k<\!\!\theta_\tau(\rho)\!\!>_{\mathit{env}}<\!\!\theta_\tau(m)\!\!>_{\mathit{mem}}\xrightarrow{P}\gamma_1$, by ASGN1 rule, $$\gamma_1=<\!\!\cdot\!\!>_k<\!\!\theta_\tau(\rho)[\theta_\tau(\rho)(e)/\mathrm{x}]\!\!>_{env}<\!\!\theta_\tau(m)\!\!>_{mem}$$ and
$$
<\!\!\cdot\!\!>_k<\!\!\theta_\tau(\rho)[\theta_\tau(\rho)(e)/\mathrm{x}]\!\!>_{\mathit{env}}<\!\!\theta_\tau(m)\!\!>_{\mathit{mem}}\xrightarrow{E}\cdots\xrightarrow{E}\gamma_i\xrightarrow{E}\gamma_{i+1}\xrightarrow{E}\cdots\xrightarrow{E}\gamma_n
$$
Due to $\theta_\tau \models \varphi$, $(\gamma_1,\tau)\models\Gamma'$. Since $\sigma$ is an actual execution, we conclude $(\gamma_i,\tau)\models\Gamma'$, $2\leq i\leq n$.
If $<\!\!\mathrm{x}:=e\!\!>_k<\!\!\theta_\tau(\rho)\!\!>_{\mathit{env}}<\!\!\theta_\tau(m)\!\!>_{\mathit{mem}}\xrightarrow{E}\gamma_1$, since $\sigma$ is an actual execution, $(\gamma_1,\tau)\models\Gamma$ and there exits $j$ ($1< j\leq n-1$) such that
$\gamma_1\xrightarrow{E}\cdots\xrightarrow{E}\gamma_j\xrightarrow{P}\gamma_{j+1}\xrightarrow{E}\cdots\xrightarrow{E}\gamma_n$
with $(\gamma_i,\tau)\models\Gamma$, $1\leq i\leq j$.
Suppose there exists $\theta_\tau^1:\mathit{SVar}\to \mathcal{T}$ with $\theta_\tau^1\!\!\upharpoonright_{\mathit{SVar}/X}=\tau\!\!\upharpoonright_{\mathit{SVar}/X}$ such that $\gamma_j=<\!\!\mathrm{x}:=e\!\!>_k<\!\!\theta_\tau^1(\rho_j)\!\!>_{\mathit{env}}<\!\!\theta_\tau^1(m_j)\!\!>_{\mathit{mem}}$  and $\theta_\tau^1 \models \varphi$. By ASGN1 rule, we conclude $\gamma_{j+1}=<\!\!\cdot\!\!>_k<\!\!\theta_\tau^1(\rho_j)[\theta_\tau^1(\rho_j)(e)/\mathrm{x}]\!\!>_{\mathit{env}}<\!\!\theta_\tau^1(m_j)\!\!>_{\mathit{mem}}$. Due to $\theta_\tau^1 \models \varphi$, $(\gamma_{j+1},\tau)\models\Gamma'$. Since $\sigma$ is an actual execution, we conclude $(\gamma_i,\tau)\models\Gamma'$, $j+1\leq i\leq n$.\\
\textbf{Case2:} By M-ASGN2 rule
$$
\dfrac{\cdot}{\exists X\!(o\!=\!<\!\!\mathrm{x}\!:=\!\!A[e]\!\!>_k<\!\!\rho\!\!>_{\mathit{e\!n\!v}}<\!\!m\!\!>_{\mathit{m\!e\!m}}\!\land \!\varphi)\!\Downarrow\!\exists  X(o\!=\!<\!\!\cdot\!\!>_k<\!\!\rho[m(\rho(A)\!+_{I\!n\!t}\!\rho(e))/\mathrm{x}]\!\!>_{\mathit{e\!n\!v}}<\!\!m\!\!>_{\mathit{m\!e\!m}}\!\!\land \!\varphi)}$$
$\Gamma\equiv \exists X(o=<\!\!\mathrm{x}:=A[e]\!\!>_k<\!\!\rho\!\!>_{\mathit{env}}<\!\!m\!\!>_{\mathit{mem}}\land \varphi)$ and $\Gamma'\equiv \exists  X(o=\!<\!\!\cdot\!\!>_k<\!\!\rho[m(\rho(A)+_{Int}\rho(e))/\mathrm{x}]\!\!>_{\mathit{env}}<\!\!m\!\!>_{\mathit{mem}}\land \varphi)$. Suppose there exists $\theta_\tau:\mathit{SVar}\to \mathcal{T}$ with $\theta_\tau\!\!\upharpoonright_{\mathit{SVar}/X}=\tau\!\!\upharpoonright_{\mathit{SVar}/X}$ such that $\gamma_0=<\!\!\mathrm{x}:=A[e]\!\!>_k<\!\!\theta_\tau(\rho)\!\!>_{\mathit{env}}<\!\!\theta_\tau(m)\!\!>_{\mathit{mem}}$  and $\theta_\tau \models \varphi$ and
$$<\!\!\mathrm{x}:=A[e]\!\!>_k<\!\!\theta_\tau(\rho)\!\!>_{\mathit{env}}<\!\!\theta_\tau(m)\!\!>_{\mathit{mem}}\xrightarrow{L}\gamma_1\xrightarrow{L}\cdots\xrightarrow{L}\gamma_i\xrightarrow{L}\gamma_{i+1}\xrightarrow{L}\cdots\xrightarrow{L}\gamma_n
$$
If $<\!\!\mathrm{x}:=A[e]\!\!>_k<\!\!\theta_\tau(\rho)\!\!>_{\mathit{env}}<\!\!\theta_\tau(m)\!\!>_{\mathit{mem}}\xrightarrow{P}\gamma_1$, by ASGN2 rule, $$\gamma_1=<\!\!\cdot\!\!>_k<\!\!\theta_\tau(\rho)[\theta_\tau(m)(\theta_\tau(\rho)
(A)+_{Int}\theta_\tau(\rho)(e))/\mathrm{x}]\!\!>_{\mathit{env}}<\!\!\theta_\tau(m)\!\!>_{\mathit{mem}}$$
and
$$ <\!\!\cdot\!\!>_k<\!\!\theta_\tau(\rho)[\theta_\tau(m)(\theta_\tau(\rho)(A)+_{Int}\theta_\tau(\rho)(e))/\mathrm{x}]\!\!>_{\mathit{env}}<\!\!\theta_\tau(m)\!\!>_{\mathit{mem}}\xrightarrow{E}\cdots\xrightarrow{E}\gamma_n$$
Due to $\theta_\tau \models \varphi$, $(\gamma_1,\tau)\models\Gamma'$. Since $\sigma$ is actual execution, we conclude $(\gamma_i,\tau)\models\Gamma'$, $2\leq i\leq n$.\\
If $<\!\!\mathrm{x}:=A[e]\!\!>_k<\!\!\theta_\tau(\rho)\!\!>_{\mathit{env}}<\!\!\theta_\tau(m)\!\!>_{\mathit{mem}}\xrightarrow{E}\gamma_1$, since $\sigma$ is actual execution, $(\gamma_1,\tau)\models\Gamma$ and there exits $j$ ($1< j\leq n-1$) such that
$$\gamma_1\xrightarrow{E}\cdots\xrightarrow{E}\gamma_j\xrightarrow{P}\gamma_{j+1}\xrightarrow{E}\cdots\xrightarrow{E}\gamma_n$$
with $(\gamma_i,\tau)\models\Gamma$, $1\leq i\leq j$.
Suppose there exists $\theta_\tau^1:\mathit{SVar}\to \mathcal{T}$ with $\theta_\tau^1\!\!\upharpoonright_{\mathit{SVar}/X}=\tau\!\!\upharpoonright_{\mathit{SVar}/X}$ such that $\gamma_j=<\!\!\mathrm{x}:=A[e]\!\!>_k<\!\!\theta_\tau^1(\rho_j)\!\!>_{\mathit{env}}<\!\!\theta_\tau^1(m_j)\!\!>_{\mathit{mem}}$  and $\theta_\tau^1 \models \varphi$. By ASGN2 rule, we conclude $$\gamma_{j+1}=<\!\!\cdot\!\!>_k<\!\!\theta_\tau^1(\rho_j)[\theta_\tau^1(m_j)(\theta_\tau^1(\rho_j)(A)+_{Int}\theta_\tau^1(\rho_j)(e))/\mathrm{x}]\!\!>_{\mathit{env}}<\!\!\theta_\tau^1(m_j)\!\!>_{\mathit{mem}}$$ Due to $\theta_\tau^1 \models \varphi$, $(\gamma_{j+1},\tau)\models\Gamma'$. Since $\sigma$ is actual execution, we conclude $(\gamma_i,\tau)\models\Gamma'$, $j+1\leq i\leq n$.\\
\textbf{Case3:} By M-ASGN3 rule
$$\dfrac{\cdot}{\exists X\!(o\!\!=\!<\!\!A[e_1\!]\!:=\!e_2\!\!>_k<\!\!\rho\!\!>_{\mathit{e\!n\!v}}<\!\!m\!\!>_{\mathit{m\!e\!m}}\!\!\land \!\varphi)\!\Downarrow\!\exists  X\!(o\!\!=\!<\!\!\cdot\!\!>_k<\!\!\rho\!\!>_{\mathit{e\!n\!v}}<\!\!m[\rho(\!e_2\!)/(\rho(\!A\!)\!+_{I\!n\!t}\!\rho(\!e_1\!)\!)]\!\!>_{\mathit{m\!e\!m}}\!\!\land\! \varphi)}$$
$\Gamma\equiv \exists X(o=<\!\!A[e_1]:=e_2\!\!>_k<\!\!\rho\!\!>_{\mathit{env}}<\!\!m\!\!>_{\mathit{mem}}\land \varphi)$ and $\Gamma'\equiv \exists  X(o=\!<\!\!\cdot\!\!>_k<\!\!\rho\!\!>_{\mathit{env}}<\!\!m[\rho(e_2)/(\rho(A)+_{Int}\rho(e_1))]\!\!>_{\mathit{mem}}\land \varphi)$.
Suppose there exists $\theta_\tau:\mathit{SVar}\to \mathcal{T}$ with $\theta_\tau\!\!\upharpoonright_{\mathit{SVar}/X}=\tau\!\!\upharpoonright_{\mathit{SVar}/X}$ such that $\gamma_0=<\!\!A[e_1]:=e_2\!\!>_k<\!\!\theta_\tau(\rho)\!\!>_{\mathit{env}}<\!\!\theta_\tau(m)\!\!>_{\mathit{mem}}$  and $\theta_\tau \models \varphi$ and
$$
<\!\!A[e_1]:=e_2\!\!>_k<\!\!\theta_\tau(\rho)\!\!>_{\mathit{env}}<\!\!\theta_\tau(m)\!\!>_{\mathit{mem}}\xrightarrow{L}\gamma_1\xrightarrow{L}\cdots\xrightarrow{L}\gamma_i\xrightarrow{L}\gamma_{i+1}\xrightarrow{L}\cdots\xrightarrow{L}\gamma_n
$$
If $<\!\!A[e_1]:=e_2\!\!>_k<\!\!\theta_\tau(\rho)\!\!>_{\mathit{env}}<\!\!\theta_\tau(m)\!\!>_{\mathit{mem}}\xrightarrow{P}\gamma_1$, by ASGN3 rule,
$$\gamma_1=<\!\!\cdot\!\!>_k<\!\!\theta_\tau(\rho)\!\!>_{\mathit{env}}<\!\!\theta_\tau(m)[\theta_\tau(\rho)(e_2)/(\theta_\tau(\rho)(A)+_{Int}\theta_\tau(\rho)(e_1))]\!\!>_{\mathit{mem}}$$ and
$$
<\!\!\cdot\!\!>_k<\!\!\theta_\tau(\rho)\!\!>_{\mathit{env}}<\!\!\theta_\tau(m)[\theta_\tau(\rho)(e_2)/(\theta_\tau(\rho)(A)+_{Int}\theta_\tau(\rho)(e_1))]\!\!>_{\mathit{mem}}\xrightarrow{E}\cdots\xrightarrow{E}\gamma_n$$
Due to $\theta_\tau \models \varphi$, $(\gamma_1,\tau)\models\Gamma'$. Since $\sigma$ is an actual execution, we conclude $(\gamma_i,\tau)\models\Gamma'$, $2\leq i\leq n$.
If $<\!\!A[e_1]:=e_2\!\!>_k<\!\!\theta_\tau(\rho)\!\!>_{\mathit{env}}<\!\!\theta_\tau(m)\!\!>_{\mathit{mem}}\xrightarrow{E}\gamma_1$, since $\sigma$ is an actual execution, $(\gamma_1,\tau)\models\Gamma$ and there exits $j$ ($1< j\leq n-1$) such that
$$
\gamma_1\xrightarrow{E}\cdots\xrightarrow{E}\gamma_j\xrightarrow{P}\gamma_{j+1}\xrightarrow{E}\cdots\xrightarrow{E}\gamma_n
$$
with $(\gamma_i,\tau)\models\Gamma$, $1\leq i\leq j$.
Suppose there exists $\theta_\tau^1:\mathit{SVar}\to \mathcal{T}$ with $\theta_\tau^1\!\!\upharpoonright_{\mathit{SVar}/X}=\tau\!\!\upharpoonright_{\mathit{SVar}/X}$ such that $\gamma_j=<\!\!A[e_1]:=e_2\!\!>_k<\!\!\theta_\tau^1(\rho_j)\!\!>_{\mathit{env}}<\!\!\theta_\tau^1(m_j)\!\!>_{\mathit{mem}}$  and $\theta_\tau^1 \models \varphi$. By ASGN3 rule, we conclude
$$
\gamma_{j+1}=<\!\!\cdot\!\!>_k<\!\!\theta_\tau^1(\rho_j)\!\!>_{\mathit{env}}<\!\!\theta_\tau^1(m_j)[\theta_\tau^1(\rho_j)(e_2)/(\theta_\tau^1(\rho_j)(A)+_{Int}\theta_\tau^1(\rho_j)(e_1))]\!\!>_{\mathit{mem}}
$$
Due to $\theta_\tau^1 \models \varphi$, $(\gamma_{j+1},\tau)\models\Gamma'$. Since $\sigma$ is an actual execution, we conclude $(\gamma_i,\tau)\models\Gamma'$, $j+1\leq i\leq n$.\\
\textbf{Case4:} By M-ARRAY rule
$$\dfrac{\cdot}{\exists X\!(o\!\!=\!<\!\!A\!\!:=\!\!A\!r\!r\!a\!y(\overline{e})\!\!>_k<\!\!\rho\!\!>_{\mathit{e\!n\!v}}<\!\!m\!\!>_{\mathit{m\!e\!m}}\!\land \!\varphi)\!\Downarrow\!\exists (\!X\!\cup \!\{p\})(o\!\!=\!<\!\!\cdot\!\!>_k<\!\!\rho[p/A]\!\!>_{\mathit{e\!n\!v}}<\!\!p\!\mapsto\![\rho(\overline{e})],m\!\!>_{\mathit{m\!e\!m}}\!\!\land \!\varphi)}$$
$\Gamma\equiv \exists X(o=<\!\!A:=Array(\overline{e})\!\!>_k<\!\!\rho\!\!>_{\mathit{env}}<\!\!m\!\!>_{\mathit{mem}}\land \varphi)$ and $\Gamma'\equiv \exists  (X\cup \{p\})(o=\!<\!\!\cdot\!\!>_k<\!\!\rho[p/A]\!\!>_{\mathit{env}}<\!\!p\mapsto[\rho(\overline{e})],m\!\!>_{\mathit{mem}}\land \varphi)$.
Suppose there exists $\theta_\tau:\mathit{SVar}\to \mathcal{T}$ with $\theta_\tau\!\!\upharpoonright_{\mathit{SVar}/X}=\tau\!\!\upharpoonright_{\mathit{SVar}/X}$ such that $\gamma_0=<\!\!A:=Array(\overline{e})\!\!>_k<\!\!\theta_\tau(\rho)\!\!>_{\mathit{env}}<\!\!\theta_\tau(m)\!\!>_{\mathit{mem}}$  and $\theta_\tau \models \varphi$ and
$$
<\!\!A:=Array(\overline{e})\!\!>_k<\!\!\theta_\tau(\rho)\!\!>_{\mathit{env}}<\!\!\theta_\tau(m)\!\!>_{\mathit{mem}}\xrightarrow{L}\gamma_1\xrightarrow{L}\cdots\xrightarrow{L}\gamma_i\xrightarrow{L}\gamma_{i+1}\xrightarrow{L}\cdots\xrightarrow{L}\gamma_n
$$
If $<\!\!A:=Array(\overline{e})\!\!>_k<\!\!\theta_\tau(\rho)\!\!>_{\mathit{env}}<\!\!\theta_\tau(m)\!\!>_{\mathit{mem}}\xrightarrow{P}\gamma_1$, by ARRAY rule,
$$\gamma_1=<\!\!\cdot\!\!>_k<\!\!\theta_\tau(\rho)[\theta_\tau(p)/A]\!\!>_{\mathit{env}}<\!\!\theta_\tau(p)\mapsto[\theta_\tau(\rho)(\overline{e})],\theta_\tau(m)\!\!>_{\mathit{mem}}$$ and
$$
<\!\!\cdot\!\!>_k<\!\!\theta_\tau(\rho)[\theta_\tau(p)/A]\!\!>_{\mathit{env}}<\!\!\theta_\tau(p)\mapsto[\theta_\tau(\rho)(\overline{e})],\theta_\tau(m)\!\!>_{\mathit{mem}}\xrightarrow{E}\cdots\xrightarrow{E}\gamma_i\xrightarrow{E}\gamma_{i+1}\xrightarrow{E}\cdots
\xrightarrow{E}\gamma_n$$
Set $\theta_\tau^1\!\!\upharpoonright_{\mathit{SVar}/(X\cup\{p\})}=\tau\!\!\upharpoonright_{\mathit{SVar}/(X\cup\{p\})}$ and $\theta_\tau^1\!\!\upharpoonright_{X}=\theta_\tau\!\!\upharpoonright_{X}$ and $\theta_\tau^1\!\!\upharpoonright_{p}=\tau\!\!\upharpoonright_{p}$. Due to $\theta_\tau \models \varphi$, $\theta_\tau^1 \models \varphi$ and $(\gamma_1,\tau)\models\Gamma'$. Since $\sigma$ is an actual execution, we conclude $(\gamma_i,\tau)\models\Gamma'$, $2\leq i\leq n$.\\
If $<\!\!A:=Array(\overline{e})\!\!>_k<\!\!\theta_\tau(\rho)\!\!>_{\mathit{env}}<\!\!\theta_\tau(m)\!\!>_{\mathit{mem}}\xrightarrow{E}\gamma_1$, since $\sigma$ is an actual execution, $(\gamma_1,\tau)\models\Gamma$ and there exits $j$ ($1< j\leq n-1$) such that
$$
\gamma_1\xrightarrow{E}\cdots\xrightarrow{E}\gamma_j\xrightarrow{P}\gamma_{j+1}\xrightarrow{E}\cdots\xrightarrow{E}\gamma_n
$$
with $(\gamma_i,\tau)\models\Gamma$, $1\leq i\leq j$.
Suppose there exists $\theta_\tau^2:\mathit{SVar}\to \mathcal{T}$ with $\theta_\tau^2\!\!\upharpoonright_{\mathit{SVar}/X}=\tau\!\!\upharpoonright_{\mathit{SVar}/X}$ such that $\gamma_j=<\!\!A:=Array(\overline{e})\!\!>_k<\!\!\theta_\tau^2(\rho_j)\!\!>_{\mathit{env}}<\!\!\theta_\tau^2(m_j)\!\!>_{\mathit{mem}}$  and $\theta_\tau^2 \models \varphi$. By  ARRAY rule, we conclude
$$
\gamma_{j+1}=<\!\!\cdot\!\!>_k<\!\!\theta_\tau^2(\rho_j)[\theta_\tau^2(p)/A]\!\!>_{\mathit{env}}<\!\!\theta_\tau^2(p)\mapsto[\theta_\tau^2(\rho_j)(\overline{e})],\theta_\tau^2(m_j)\!\!>_{\mathit{mem}}
$$
Set $\theta_\tau^3\!\!\upharpoonright_{\mathit{SVar}/(X\cup\{p\})}=\tau\!\!\upharpoonright_{\mathit{SVar}/(X\cup\{p\})}$ and $\theta_\tau^3\!\!\upharpoonright_{X}=\theta_\tau^2\!\!\upharpoonright_{X}$ and $\theta_\tau^3\!\!\upharpoonright_{p}=\tau\!\!\upharpoonright_{p}$.
Due to $\theta_\tau^2 \models \varphi$, $\theta_\tau^3 \models \varphi$ and $(\gamma_{j+1},\tau)\models\Gamma'$. Since $\sigma$ is an actual execution, we conclude $(\gamma_i,\tau)\models\Gamma'$, $j+1\leq i\leq n$.\\
\textbf{Case5:} By M-SEQ rule
\begin{align*}
  &\exists X_1(o=<\!\!c_1\!\!>_k<\!\!\rho_1\!\!>_{\mathit{env}}<\!\!m_1\!\!>_{\mathit{mem}}\land \varphi_1)\Downarrow \exists X_2(o=<\!\!\cdot\!\!>_k<\!\!\rho_2\!\!>_{\mathit{env}}<\!\!m_2\!\!>_{\mathit{mem}}\land \varphi_2)\\
  &\exists X_2(o=<\!\!c_2\!\!>_k<\!\!\rho_2\!\!>_{\mathit{env}}<\!\!m_2\!\!>_{\mathit{mem}}\land \varphi_2)\Downarrow\exists  X_3(o=<\!\!\cdot\!\!>_k<\!\!\rho_3\!\!>_{\mathit{env}}<\!\!m_3\!\!>_{\mathit{mem}}\land \varphi_3)\\
  &\overline{\exists X_1(o=<\!\!c_1;c_2\!\!>_k<\!\!\rho_1\!\!>_{\mathit{env}}<\!\!m_1\!\!>_{\mathit{mem}}\land \varphi_1)\Downarrow\exists X_3(o=<\!\!\cdot\!\!>_k<\!\!\rho_3\!\!>_{\mathit{env}}<\!\!m_3\!\!>_{\mathit{mem}}\land \varphi_3)}
\end{align*}
$\Gamma\equiv \exists X_1(o=<\!\!c_1;c_2\!\!>_k<\!\!\rho_1\!\!>_{\mathit{env}}<\!\!m_1\!\!>_{\mathit{mem}}\land \varphi_1)$ and
$\Gamma'\equiv \exists X_3(o=<\!\!\cdot\!\!>_k<\!\!\rho_3\!\!>_{\mathit{env}}<\!\!m_3\!\!>_{\mathit{mem}}\land \varphi_3)$.
Suppose there exists $\theta_\tau:\mathit{SVar}\to \mathcal{T}$ with $\theta_\tau\!\!\upharpoonright_{\mathit{SVar}/X_1}=\tau\!\!\upharpoonright_{\mathit{SVar}/X_1}$ such that $\gamma_0=<\!\!c_1;c_2\!\!>_k<\!\!\theta_\tau(\rho_1)\!\!>_{\mathit{env}}<\!\!\theta_\tau(m_1)\!\!>_{\mathit{mem}}$  and $\theta_\tau \models \varphi_1$ and
$$
<\!\!c_1;c_2\!\!>_k<\!\!\theta_\tau(\rho)\!\!>_{\mathit{env}}<\!\!\theta_\tau(m)\!\!>_{\mathit{mem}}\xrightarrow{L}\gamma_1\xrightarrow{L}\cdots\xrightarrow{L}\gamma_i\xrightarrow{L}\gamma_{i+1}\xrightarrow{L}\cdots\xrightarrow{L}\gamma_n
$$
Proposition 5 implies there exits some final configuration $<\!\!\cdot\!\!>_k<\!\!\rho'\!\!>_{\mathit{env}}<\!\!m'\!\!>_{\mathit{mem}}$ such that
$$<\!\!c_1\!\!>_k<\!\!\theta_\tau(\rho_1)\!\!>_{\mathit{env}}<\!\!\theta_\tau(m_1)\!\!>_{\mathit{mem}}\xrightarrow{L}^*<\!\!\cdot\!\!>_k<\!\!\rho'\!\!>_{\mathit{env}}<\!\!m'\!\!>_{\mathit{mem}}$$ and
$$<\!\!c_2\!\!>_k<\!\!\rho'\!\!>_{\mathit{env}}<\!\!m'\!\!>_{\mathit{mem}}\xrightarrow{L}^*\gamma_n$$
By the induction hypothesis of
$$
\exists X_1(o=<\!\!c_1\!\!>_k<\!\!\rho_1\!\!>_{\mathit{env}}<\!\!m_1\!\!>_{\mathit{mem}}\land \varphi_1)\Downarrow \exists X_2(o=<\!\!\cdot\!\!>_k<\!\!\rho_2\!\!>_{\mathit{env}}<\!\!m_2\!\!>_{\mathit{mem}}\land \varphi_2)
$$
we conclude
$(<\!\!\cdot\!\!>_k<\!\!\rho'\!\!>_{\mathit{env}}<\!\!m'\!\!>_{\mathit{mem}},\tau)\models \exists X_2(o=<\!\!\cdot\!\!>_k<\!\!\rho_2\!\!>_{\mathit{env}}<\!\!m_2\!\!>_{\mathit{mem}}\land \varphi_2)$.
Since $c_2$  is ground,
$(<\!\!c_2\!\!>_k<\!\!\rho'\!\!>_{\mathit{env}}<\!\!m'\!\!>_{\mathit{mem}},\tau)\models \exists X_2(o=<\!\!c_2\!\!>_k<\!\!\rho_2\!\!>_{\mathit{env}}<\!\!m_2\!\!>_{\mathit{mem}}\land \varphi_2)$. By the induction hypothesis of
$$
\exists X_2(o=<\!\!c_2\!\!>_k<\!\!\rho_2\!\!>_{\mathit{env}}<\!\!m_2\!\!>_{\mathit{mem}}\land \varphi_2)\Downarrow\exists  X_3(o=<\!\!\cdot\!\!>_k<\!\!\rho_3\!\!>_{\mathit{env}}<\!\!m_3\!\!>_{\mathit{mem}}\land \varphi_3)
$$
we conclude $(\gamma_n,\tau)\models \exists  X_3(o=<\!\!\cdot\!\!>_k<\!\!\rho_3\!\!>_{\mathit{env}}<\!\!m_3\!\!>_{\mathit{mem}}\land \varphi_3)$.\\
\textbf{Case6:} By M-IF rule
\begin{align*}
&\exists X_1(o\!=\!<\!\!c_1\!\!>_k<\!\!\rho\!\!>_{\mathit{e\!n\!v}}<\!\!m\!\!>_{\mathit{m\!e\!m}}\!\land \!\varphi \!\land \! (\rho(b)\;\;is\;\; true))\!\Downarrow\!\exists X_2(o\!=\!<\!\!\cdot\!\!>_k<\!\!\rho'\!\!>_{\mathit{e\!n\!v}}<\!\!m'\!\!>_{\mathit{m\!e\!m}}\!\land \!\varphi')\\
&\exists X_1(o\!=\!<\!\!c_2\!\!>_k<\!\!\rho\!\!>_{\mathit{e\!n\!v}}<\!\!m\!\!>_{\mathit{m\!e\!m}}\!\land \!\varphi \! \land\!(\rho(b)\;\;is\;\; false))\!\Downarrow\!\exists X_2(o\!=\!<\!\!\cdot\!\!>_k<\!\!\rho'\!\!>_{\mathit{e\!n\!v}}<\!\!m'\!\!>_{\mathit{m\!e\!m}}\!\land\! \varphi')\\
&\overline{\exists X_1(o=<\!\!\textrm{if }(b)c_1\textrm{ else }c_2 \!\!>_k<\!\!\rho\!\!>_{\mathit{env}}<\!\!m\!\!>_{\mathit{mem}}\land \varphi)\Downarrow\exists X_2(o=<\!\!\cdot\!\!>_k<\!\!\rho'\!\!>_{\mathit{env}}<\!\!m'\!\!>_{\mathit{mem}}\land \varphi')}
\end{align*}
$\Gamma\equiv\exists X_1(o=<\!\!\textrm{if }(b)c_1\textrm{ else }c_2 \!\!>_k<\!\!\rho\!\!>_{\mathit{env}}<\!\!m\!\!>_{\mathit{mem}}\land \varphi)$ and $\Gamma'\equiv\exists X_2(o=<\!\!\cdot\!\!>_k<\!\!\rho'\!\!>_{\mathit{env}}<\!\!m'\!\!>_{\mathit{mem}}\land \varphi')$.
Suppose there exists $\theta_\tau:\mathit{SVar}\to \mathcal{T}$ with $\theta_\tau\!\!\upharpoonright_{\mathit{SVar}/X_1}=\tau\!\!\upharpoonright_{\mathit{SVar}/X_1}$ such that $\gamma_0=<\!\!\textrm{if }(b)c_1\textrm{ else }c_2\!\!>_k<\!\!\theta_\tau(\rho)\!\!>_{\mathit{env}}<\!\!\theta_\tau(m)\!\!>_{\mathit{mem}}$  and $\theta_\tau \models \varphi$ and
$$
<\!\!\textrm{if }(b)c_1\textrm{ else }c_2\!\!>_k<\!\!\theta_\tau(\rho)\!\!>_{\mathit{env}}<\!\!\theta_\tau(m)\!\!>_{\mathit{mem}}\xrightarrow{L}\gamma_1\xrightarrow{L}\cdots\xrightarrow{L}\gamma_i\xrightarrow{L}\gamma_{i+1}\xrightarrow{L}\cdots\xrightarrow{L}\gamma_n
$$
If $<\!\!\textrm{if }(b)c_1\textrm{ else }c_2\!\!>_k<\!\!\theta_\tau(\rho)\!\!>_{\mathit{env}}<\!\!\theta_\tau(m)\!\!>_{\mathit{mem}}\xrightarrow{P}\gamma_1$, we distinguish two cases according to whether $\theta_\tau(\rho)(b)$ is true or not.\\
\textbf{Case6.1:} $\theta_\tau(\rho)(b)$ is true; by IF1 rule,
$$
<\!\!\textrm{if }(b)c_1\textrm{ else }c_2\!\!>_k<\!\!\theta_\tau(\rho)\!\!>_{\mathit{env}}<\!\!\theta_\tau(m)\!\!>_{\mathit{mem}}\xrightarrow{P}<\!\!c_1\!\!>_k<\!\!\theta_\tau(\rho)\!\!>_{\mathit{env}}<\!\!\theta_\tau(m)\!\!>_{\mathit{mem}}\xrightarrow{L}^*\gamma_n
$$
By the induction hypothesis of
$$
\exists X_1(o\!=\!<\!\!c_1\!\!>_k<\!\!\rho\!\!>_{\mathit{e\!n\!v}}<\!\!m\!\!>_{\mathit{m\!e\!m}}\!\land \!\varphi\! \land \! (\rho(b)\;\;is\;\; true))\!\Downarrow\!\exists X_2(o\!=\!<\!\!\cdot\!\!>_k<\!\!\rho'\!\!>_{\mathit{e\!n\!v}}<\!\!m'\!\!>_{\mathit{m\!e\!m}}\!\land\! \varphi')
$$
we conclude
$(\gamma_n,\tau)\models\exists X_2(o=<\!\!\cdot\!\!>_k<\!\!\rho'\!\!>_{\mathit{env}}<\!\!m'\!\!>_{\mathit{mem}}\land \varphi')$.\\
\textbf{Case6.2:} $\theta_\tau(\rho)(b)$ is false; by IF2 rule,
$$
<\!\!\textrm{if }(b)c_1\textrm{ else }c_2\!\!>_k<\!\!\theta_\tau(\rho)\!\!>_{\mathit{env}}<\!\!\theta_\tau(m)\!\!>_{\mathit{mem}}\xrightarrow{P}<\!\!c_2\!\!>_k<\!\!\theta_\tau(\rho)\!\!>_{\mathit{env}}<\!\!\theta_\tau(m)\!\!>_{\mathit{mem}}\xrightarrow{L}^*\gamma_n
$$
By the induction hypothesis of
$$
\exists X_1(o\!=\!<\!\!c_2\!\!>_k<\!\!\rho\!\!>_{\mathit{e\!n\!v}}<\!\!m\!\!>_{\mathit{m\!e\!m}}\!\land \!\varphi\! \land \! (\rho(b)\;\;is\;\; false))\!\Downarrow\!\exists X_2(o=<\!\!\cdot\!\!>_k<\!\!\rho'\!\!>_{\mathit{e\!n\!v}}<\!\!m'\!\!>_{\mathit{m\!e\!m}}\!\land\! \varphi')
$$
we conclude
$(\gamma_n,\tau)\models\exists X_2(o=<\!\!\cdot\!\!>_k<\!\!\rho'\!\!>_{\mathit{env}}<\!\!m'\!\!>_{\mathit{mem}}\land \varphi')$.\\
If $<\!\!\textrm{if }(b)c_1\textrm{ else }c_2\!\!>_k<\!\!\theta_\tau(\rho)\!\!>_{\mathit{env}}<\!\!\theta_\tau(m)\!\!>_{\mathit{mem}}\xrightarrow{E}\gamma_1$,  since $\sigma$ is an actual execution, $(\gamma_1,\tau)\models\Gamma$ and there exits $j$ ($1< j\leq n-1$) such that
$$
\gamma_1\xrightarrow{E}\cdots\xrightarrow{E}\gamma_j\xrightarrow{P}\gamma_{j+1}\xrightarrow{E}\cdots\xrightarrow{E}\gamma_n
$$
with $(\gamma_i,\tau)\models\Gamma$, $1\leq i\leq j$. Suppose there exists $\theta_\tau^1:\mathit{SVar}\to \mathcal{T}$ with $\theta_\tau^1\!\!\upharpoonright_{\mathit{SVar}/X_1}=\tau\!\!\upharpoonright_{\mathit{SVar}/X_1}$ such that $\gamma_j=<\!\!\textrm{if }(b)c_1\textrm{ else }c_2\!\!>_k<\!\!\theta_\tau^1(\rho_j)\!\!>_{\mathit{env}}<\!\!\theta_\tau^1(m_j)\!\!>_{\mathit{mem}}$  and $\theta_\tau^1 \models \varphi$. We also distinguish two cases according to whether $\theta_\tau^1(\rho_j)(b)$ is true or not.\\
\textbf{Case6.3:} $\theta_\tau^1(\rho_j)(b)$ is true; by IF1 rule,
$$
<\!\!\textrm{if }(b)c_1\textrm{ else }c_2\!\!>_k<\!\!\theta_\tau^1(\rho_j)\!\!>_{\mathit{env}}<\!\!\theta_\tau^1(m_j)\!\!>_{\mathit{mem}}\xrightarrow{P}<\!\!c_1\!\!>_k<\!\!\theta_\tau^1(\rho_j)\!\!>_{\mathit{env}}<\!\!\theta_\tau^1(m_j)\!\!>_{\mathit{mem}}\xrightarrow{L}^*\gamma_n
$$
Since $\textrm{if }(b)c_1\textrm{ else }c_2$ is ground,
$$
(<\!\!c_1\!\!>_k<\!\!\theta_\tau^1(\rho_j)\!\!>_{\mathit{env}}<\!\!\theta_\tau^1(m_j)\!\!>_{\mathit{mem}},\tau)\models \exists X_1(o\!=\!<\!\!c_1\!\!>_k<\!\!\rho\!\!>_{\mathit{e\!n\!v}}<\!\!m\!\!>_{\mathit{m\!e\!m}}\!\land \!\varphi\! \land \! (\rho(b)\;\;is\;\; true))
$$
By the induction hypothesis of
$$
\exists X_1(o\!=\!<\!\!c_1\!\!>_k<\!\!\rho\!\!>_{\mathit{e\!n\!v}}<\!\!m\!\!>_{\mathit{m\!e\!m}}\!\land\! \varphi\! \land\! (\rho(b)\;\;is\;\; true))\!\Downarrow\!\exists X_2(o\!=\!<\!\!\cdot\!\!>_k<\!\!\rho'\!\!>_{\mathit{e\!n\!v}}<\!\!m'\!\!>_{\mathit{m\!e\!m}}\!\land \!\varphi')
$$
we conclude
$(\gamma_n,\tau)\models\exists X_2(o=<\!\!\cdot\!\!>_k<\!\!\rho'\!\!>_{\mathit{env}}<\!\!m'\!\!>_{\mathit{mem}}\land \varphi')$.\\
\textbf{Case6.4:} $\theta_\tau^1(\rho_j)(b)$ is false; by IF2 rule,
$$
<\!\!\textrm{if }(b)c_1\textrm{ else }c_2\!\!>_k<\!\!\theta_\tau^1(\rho_j)\!\!>_{\mathit{env}}<\!\!\theta_\tau^1(m_j)\!\!>_{\mathit{mem}}\xrightarrow{P}<\!\!c_2\!\!>_k<\!\!\theta_\tau^1(\rho_j)\!\!>_{\mathit{env}}<\!\!\theta_\tau^1(m_j)\!\!>_{\mathit{mem}}\xrightarrow{L}^*\gamma_n
$$
Since $\textrm{if }(b)c_1\textrm{ else }c_2$ is ground,
$$
(<\!\!c_2\!\!>_k<\!\!\theta_\tau^1(\rho_j)\!\!>_{\mathit{e\!n\!v}}<\!\!\theta_\tau^1(m_j)\!\!>_{\mathit{m\!e\!m}},\tau)\!\models \! \exists X_1(o\!=\!<\!\!c_2\!\!>_k<\!\!\rho\!\!>_{\mathit{e\!n\!v}}<\!\!m\!\!>_{\mathit{m\!e\!m}}\!\land\! \varphi\! \land \! (\rho(b)\;is\; false))
$$
By the induction hypothesis of
$$
\exists X_1(o\!=\!<\!\!c_2\!\!>_k<\!\!\rho\!\!>_{\mathit{e\!n\!v}}<\!\!m\!\!>_{\mathit{m\!e\!m}}\!\land\! \varphi\! \land \! (\rho(b)\;\;is\;\; false))\!\Downarrow\!\exists X_2(o\!=\!<\!\!\cdot\!\!>_k<\!\!\rho'\!\!>_{\mathit{e\!n\!v}}<\!\!m'\!\!>_{\mathit{m\!e\!m}}\!\land \!\varphi')
$$
we conclude
$(\gamma_n,\tau)\models\exists X_2(o=<\!\!\cdot\!\!>_k<\!\!\rho'\!\!>_{\mathit{env}}<\!\!m'\!\!>_{\mathit{mem}}\land \varphi')$.\\
\textbf{Case7:} By M-CONS rule
$$\dfrac{\models \Gamma \Rightarrow \Gamma_1, \Gamma_1\Downarrow \Gamma_1',\models \Gamma_1' \Rightarrow \Gamma'}{\Gamma\Downarrow \Gamma'}$$
Suppose $(\gamma_0,\tau)\models \Gamma$ and
$\gamma_0\xrightarrow{L}\gamma_1\xrightarrow{L}\cdots\xrightarrow{L}\gamma_i\xrightarrow{L}\gamma_{i+1}\xrightarrow{L}\cdots\xrightarrow{L}\gamma_n$.
Due to $\models \Gamma \Rightarrow \Gamma_1$, $(\gamma_0,\tau)\models \Gamma_1$. By the induction hypothesis of $\Gamma_1\Downarrow \Gamma_1'$, we conclude $(\gamma_n,\tau)\models \Gamma_1'$. Due to $\models \Gamma_1' \Rightarrow \Gamma'$,  $(\gamma_n,\tau)\models \Gamma'$.\\
\textbf{Case8:} By M-CASE rule
$$\dfrac{\models \Gamma \Rightarrow \Gamma_1\vee \Gamma_2, \Gamma_1\Downarrow \Gamma', \Gamma_2\Downarrow \Gamma'}{\Gamma\Downarrow \Gamma'}$$
Suppose $(\gamma_0,\tau)\models \Gamma$ and
$\gamma_0\xrightarrow{L}\gamma_1\xrightarrow{L}\cdots\xrightarrow{L}\gamma_i\xrightarrow{L}\gamma_{i+1}\xrightarrow{L}\cdots\xrightarrow{L}\gamma_n$.
Due to $\models \Gamma \Rightarrow \Gamma_1\vee \Gamma_2$, $(\gamma_0,\tau)\models \Gamma_1$ or $(\gamma_0,\tau)\models \Gamma_2$.
If $(\gamma_0,\tau)\models \Gamma_1$,
by the induction hypothesis of $\Gamma_1\Downarrow \Gamma'$, we conclude $(\gamma_n,\tau)\models \Gamma'$.
If $(\gamma_0,\tau)\models \Gamma_2$,
by the induction hypothesis of $\Gamma_2\Downarrow \Gamma'$, we conclude $(\gamma_n,\tau)\models \Gamma'$. Thus, $(\gamma_n,\tau)\models \Gamma'$.\\
\textbf{Case9:} By M-AWAIT rule
$$
\dfrac{\exists  X_1(o\!=\!<\!\!cc\!\!>_k<\!\!\rho\!\!>_{\mathit{e\!n\!v}}<\!\!m\!\!>_{\mathit{m\!e\!m}}\!\land \!\varphi \! \land\!(\rho(b)\;\;is\;\; true))\!\Downarrow\!\exists X_2(o\!=\!<\!\!\cdot\!\!>_k<\!\!\rho'\!\!>_{\mathit{e\!n\!v}}<\!\!m'\!\!>_{\mathit{m\!e\!m}}\!\land \! \varphi')}{\exists X_1(o\!=\!<\!\!\textrm{await }b\textrm{ then }\mathit{cc}\!\! >_k<\!\!\rho\!\!>_{\mathit{e\!n\!v}}<\!\!m\!\!>_{\mathit{m\!e\!m}}\!\land \!\varphi)\!\Downarrow\!\exists X_2(o\!=\!<\!\!\cdot\!\!>_k<\!\!\rho'\!\!>_{\mathit{e\!n\!v}}<\!\!m'\!\!>_{\mathit{m\!e\!m}}\!\land\! \varphi')}
$$
$\Gamma\equiv \exists X_1(o=<\!\!\textrm{await }b\textrm{ then }\mathit{cc}\!\! >_k<\!\!\rho\!\!>_{\mathit{env}}<\!\!m\!\!>_{\mathit{mem}}\land \varphi)$ and $\Gamma'\equiv \exists X_2(o=<\!\!\cdot\!\!>_k<\!\!\rho'\!\!>_{\mathit{env}}<\!\!m'\!\!>_{\mathit{mem}}\land \varphi')$.
Suppose there exists $\theta_\tau:\mathit{SVar}\to \mathcal{T}$ with $\theta_\tau\!\!\upharpoonright_{\mathit{SVar}/X_1}=\tau\!\!\upharpoonright_{\mathit{SVar}/X_1}$ such that $\gamma_0=<\!\!\textrm{await }b\textrm{ then }\mathit{cc}\!\!>_k<\!\!\theta_\tau(\rho)\!\!>_{\mathit{env}}<\!\!\theta_\tau(m)\!\!>_{\mathit{mem}}$  and $\theta_\tau \models \varphi$ and
$$
<\!\!\textrm{await }b\textrm{ then }\mathit{cc}\!\!>_k<\!\!\theta_\tau(\rho)\!\!>_{\mathit{env}}<\!\!\theta_\tau(m)\!\!>_{\mathit{mem}}\xrightarrow{L}\gamma_1\xrightarrow{L}\cdots\xrightarrow{L}\gamma_i\xrightarrow{L}\gamma_{i+1}\xrightarrow{L}\cdots\xrightarrow{L}\gamma_n
$$
If $<\!\!\textrm{await }b\textrm{ then }\mathit{cc}\!\!>_k<\!\!\theta_\tau(\rho)\!\!>_{\mathit{env}}<\!\!\theta_\tau(m)\!\!>_{\mathit{mem}}\xrightarrow{P}\gamma_1$, by AWAIT rule, we conclude $\gamma_1$ is a final configuration and $\theta_\tau(\rho)(b)$ is true and
$<\!\!\mathit{cc}\!\!>_k<\!\!\theta_\tau(\rho)\!\!>_{\mathit{env}}<\!\!\theta_\tau(m)\!\!>_{\mathit{mem}}\xrightarrow{L}^*\gamma_1$. By the induction hypothesis of
$$
\exists  X_1(o\!=\!<\!\!cc\!\!>_k<\!\!\rho\!\!>_{\mathit{e\!n\!v}}<\!\!m\!\!>_{\mathit{m\!e\!m}}\!\land\! \varphi \! \land\!(\rho(b)\;\;is\;\; true))\!\Downarrow\!\exists X_2(o\!=\!<\!\!\cdot\!\!>_k<\!\!\rho'\!\!>_{\mathit{e\!n\!v}}<\!\!m'\!\!>_{\mathit{m\!e\!m}}\!\land\!  \varphi')
$$
we conclude $(\gamma_1,\tau)\models\exists X_2(o=<\!\!\cdot\!\!>_k<\!\!\rho'\!\!>_{\mathit{env}}<\!\!m'\!\!>_{\mathit{mem}}\land  \varphi')$. Since $\sigma$ is an actual execution, $(\gamma_i,\tau)\models\Gamma'$, $2\leq i\leq n$.\\
If $<\!\!\textrm{await }b\textrm{ then }\mathit{cc}\!\!>_k<\!\!\theta_\tau(\rho)\!\!>_{\mathit{env}}<\!\!\theta_\tau(m)\!\!>_{\mathit{mem}}\xrightarrow{E}\gamma_1$,
since $\sigma$ is an actual execution, there exits $j$ ($1< j\leq n-1$) such that
$$\gamma_1\xrightarrow{E}\cdots\xrightarrow{E}\gamma_j\xrightarrow{P}\gamma_{j+1}\xrightarrow{E}\cdots\xrightarrow{E}\gamma_n$$
with $(\gamma_i,\tau)\models\Gamma$, $1\leq i\leq j$. Suppose there exists $\theta_\tau^1:\mathit{SVar}\to \mathcal{T}$ with $\theta_\tau^1\!\!\upharpoonright_{\mathit{SVar}/X_1}=\tau\!\!\upharpoonright_{\mathit{SVar}/X_1}$ such that $\gamma_j=<\!\!\textrm{await }b\textrm{ then }\mathit{cc}\!\!>_k<\!\!\theta_\tau^1(\rho_j)\!\!>_{\mathit{env}}<\!\!\theta_\tau^1(m_j)\!\!>_{\mathit{mem}}$  and $\theta_\tau^1 \models \varphi$. By AWAIT rule, we conclude $\gamma_{j+1}$ is a final configuration and $\theta_\tau^1(\rho_j)(b)$ is true and
$<\!\!\mathit{cc}\!\!>_k<\!\!\theta_\tau^1(\rho_j)\!\!>_{\mathit{env}}<\!\!\theta_\tau^1(m_j)\!\!>_{\mathit{mem}}\xrightarrow{L}^*\gamma_{j+1}$.
By the induction hypothesis of
$$\exists  X_1(o=<\!\!cc\!\!>_k<\!\!\rho\!\!>_{\mathit{env}}<\!\!m\!\!>_{\mathit{mem}}\land \varphi  \land(\rho(b)\;\;is\;\; true))\Downarrow\exists X_2(o=<\!\!\cdot\!\!>_k<\!\!\rho'\!\!>_{\mathit{env}}<\!\!m'\!\!>_{\mathit{mem}}\land  \varphi')$$
we conclude
$(\gamma_{j+1},\tau)\models\exists X_2(o=<\!\!\cdot\!\!>_k<\!\!\rho'\!\!>_{\mathit{env}}<\!\!m'\!\!>_{\mathit{mem}}\land  \varphi')$. Since $\sigma$ is an actual execution, $(\gamma_i,\tau)\models\Gamma'$, $j+1\leq i\leq n$.\\
\textbf{Case10:} By M-WHILE rule
$$\dfrac{\exists X(o\!=\!<\!\!c\!\!>_k<\!\!\rho\!\!>_{\mathit{e\!n\!v}}<\!\!m\!\!>_{\mathit{m\!e\!m}}\!\land\! \varphi \! \land\!(\rho(b)\;\;is\;\; true))\!\Downarrow\!\exists X(o\!=\!<\!\!\cdot\!\!>_k<\!\!\rho\!\!>_{\mathit{e\!n\!v}}<\!\!m\!\!>_{\mathit{m\!e\!m}}\!\land \!\varphi)}{\exists X(o\!=\!<\!\!\textrm{while }b\textrm{ do } c\!\!>_k<\!\!\rho\!\!>_{\mathit{e\!n\!v}}<\!\!m\!\!>_{\mathit{m\!e\!m}}\!\land \!\varphi)\!\Downarrow\!\exists X(o\!=\!<\!\!\cdot\!\!>_k<\!\!\rho\!\!>_{\mathit{e\!n\!v}}<\!\!m\!\!>_{\mathit{m\!e\!m}}\!\land \varphi\!\land\!(\rho(b)\;\;is\;\; false))}$$
$\Gamma\equiv \exists X(o=<\!\!\textrm{while }b\textrm{ do } c\!\!>_k<\!\!\rho\!\!>_{\mathit{env}}<\!\!m\!\!>_{\mathit{mem}}\land \varphi)$ and $\Gamma'\equiv \exists X(o=<\!\!\cdot\!\!>_k<\!\!\rho\!\!>_{\mathit{env}}<\!\!m\!\!>_{\mathit{mem}}\land \varphi\land(\rho(b)\;\;is\;\; false))$. Suppose there exists $\theta_\tau:\mathit{SVar}\to \mathcal{T}$ with $\theta_\tau\!\!\upharpoonright_{\mathit{SVar}/X}=\tau\!\!\upharpoonright_{\mathit{SVar}/X}$ such that $\gamma_0=<\!\!\textrm{while }b\textrm{ do } c\!\!>_k<\!\!\theta_\tau(\rho)\!\!>_{\mathit{env}}<\!\!\theta_\tau(m)\!\!>_{\mathit{mem}}$  and $\theta_\tau \models \varphi$ and
$$<\!\!\textrm{while }b\textrm{ do } c\!\!>_k<\!\!\theta_\tau(\rho)\!\!>_{\mathit{env}}<\!\!\theta_\tau(m)\!\!>_{\mathit{mem}}\xrightarrow{L}^n\gamma_n$$
We prove by well-founded induction on $n$ that
$$(\gamma_n,\tau)\models \exists X(o=<\!\!\cdot\!\!>_k<\!\!\rho\!\!>_{\mathit{env}}<\!\!m\!\!>_{\mathit{mem}}\land \varphi\land(\rho(b)\;\;is\;\; false))$$
There exit $n_1,n_2$ such that
$$<\!\!\textrm{while }b\textrm{ do } c\!\!>_k<\!\!\theta_\tau(\rho)\!\!>_{\mathit{e\!n\!v}}<\!\!\theta_\tau(m)\!\!>_{\mathit{m\!e\!m}}\xrightarrow{E}^{n_1}<\!\!\textrm{while }b\textrm{ do } c\!\!>_k<\!\!\rho'\!\!>_{\mathit{e\!n\!v}}<\!\!m'\!\!>_{\mathit{m\!e\!m}}\xrightarrow{P}\gamma'\xrightarrow{L}^{n_2}\gamma_n$$
with $n_1+n_2+1=n$.\\
If $\rho'(b)$ is true, by WHILE1 rule,
$\gamma'=<\!\!c;\;\textrm{while }b\textrm{ do } c \!\!>_k<\!\!\rho'\!\!>_{\mathit{env}}<\!\!m'\!\!>_{\mathit{mem}}$
and
$$<\!\!c;\;\textrm{while }b\textrm{ do } c \!\!>_k<\!\!\rho'\!\!>_{\mathit{env}}<\!\!m'\!\!>_{\mathit{mem}}\xrightarrow{L}^{n_2}\gamma_n$$
Proposition 5 implies there exits some final configuration $<\!\!\cdot\!\!>_k<\!\!\rho''\!\!>_{\mathit{env}}<\!\!m''\!\!>_{\mathit{mem}}$ such that
$$<\!\!c \!\!>_k<\!\!\rho'\!\!>_{\mathit{env}}<\!\!m'\!\!>_{\mathit{mem}}\xrightarrow{L}^*<\!\!\cdot\!\!>_k<\!\!\rho''\!\!>_{\mathit{env}}<\!\!m''\!\!>_{\mathit{mem}}$$
and
$$<\!\!\textrm{while }b\textrm{ do } c\!\!>_k<\!\!\rho''\!\!>_{\mathit{env}}<\!\!m''\!\!>_{\mathit{mem}}\xrightarrow{L}^{n_3}\gamma_n$$
with $n_3<n_2$. Since $\sigma$ is an actual execution, $(<\!\!\textrm{while }b\textrm{ do } c \!\!>_k<\!\!\rho'\!\!>_{\mathit{env}}<\!\!m'\!\!>_{\mathit{mem}},\tau)\models \Gamma$. Since $\textrm{while }b\textrm{ do } c$ is ground and $\rho'(b)$ is true,
$$(<\!\!c \!\!>_k<\!\!\rho'\!\!>_{\mathit{env}}<\!\!m'\!\!>_{\mathit{mem}},\tau)\models \exists X(o=<\!\!c\!\!>_k<\!\!\rho\!\!>_{\mathit{env}}<\!\!m\!\!>_{\mathit{mem}}\land \varphi \land(\rho(b)\;\;is\;\; true))$$
By the induction hypothesis of
$$\exists X(o=<\!\!c\!\!>_k<\!\!\rho\!\!>_{\mathit{env}}<\!\!m\!\!>_{\mathit{mem}}\land \varphi \land(\rho(b)\;\;is\;\; true))\Downarrow\exists X(o=<\!\!\cdot\!\!>_k<\!\!\rho\!\!>_{\mathit{env}}<\!\!m\!\!>_{\mathit{mem}}\land \varphi)$$
we conclude
$(<\!\!\cdot\!\!>_k<\!\!\rho''\!\!>_{\mathit{env}}<\!\!m''\!\!>_{\mathit{mem}},\tau)\models \exists X(o=<\!\!\cdot\!\!>_k<\!\!\rho\!\!>_{\mathit{env}}<\!\!m\!\!>_{\mathit{mem}}\land \varphi)$.
Since $\textrm{while }b\textrm{ do } c$ is ground,
$$(<\!\!\textrm{while }b\textrm{ do } c\!\!>_k<\!\!\rho''\!\!>_{\mathit{env}}<\!\!m''\!\!>_{\mathit{mem}},\tau)\models \exists X(o=<\!\!\textrm{while }b\textrm{ do } c\!\!>_k<\!\!\rho\!\!>_{\mathit{env}}<\!\!m\!\!>_{\mathit{mem}}\land \varphi)$$
By the inner induction hypothesis ($n_3<n$), we conclude
$$(\gamma_n,\tau)\models \exists X(o=<\!\!\cdot\!\!>_k<\!\!\rho\!\!>_{\mathit{env}}<\!\!m\!\!>_{\mathit{mem}}\land \varphi\land(\rho(b)\;\;is\;\; false))$$
If $\rho'(b)$ is false, by WHILE2 rule,
$\gamma'=<\!\!\cdot \!\!>_k<\!\!\rho'\!\!>_{\mathit{env}}<\!\!m'\!\!>_{\mathit{mem}}$.
Since $\sigma$ is an actual execution, $(<\!\!\textrm{while }b\textrm{ do } c \!\!>_k<\!\!\rho'\!\!>_{\mathit{env}}<\!\!m'\!\!>_{\mathit{mem}},\tau)\models \Gamma$. Since $\textrm{while }b\textrm{ do } c$ is ground and $\rho'(b)$ is false,
$$(<\!\!\cdot \!\!>_k<\!\!\rho'\!\!>_{\mathit{env}}<\!\!m'\!\!>_{\mathit{mem}},\tau)\models \exists X(o=<\!\!\cdot\!\!>_k<\!\!\rho\!\!>_{\mathit{env}}<\!\!m\!\!>_{\mathit{mem}}\land \varphi\land(\rho(b)\;\;is\;\; false))$$
Since $\sigma$ is an actual execution, $(\gamma_n,\tau)\models\Gamma'$.\\
\textbf{Case11:} By M-PAR rule
\begin{align*}
&\exists X_1(o=<\!\!c_1\!\!>_k<\!\!\rho\!\!>_{\mathit{env}}<\!\!m\!\!>_{\mathit{mem}}\land \varphi_1)\Downarrow\exists X_2(o=<\!\!\cdot\!\!>_k<\!\!\rho'\!\!>_{\mathit{env}}<\!\!m'\!\!>_{\mathit{mem}}\land \varphi_1')\\
&\exists X_1(o=<\!\!c_2\!\!>_k<\!\!\rho\!\!>_{\mathit{env}}<\!\!m\!\!>_{\mathit{mem}}\land \varphi_2)\Downarrow\exists X_2(o=<\!\!\cdot\!\!>_k<\!\!\rho'\!\!>_{\mathit{env}}<\!\!m'\!\!>_{\mathit{mem}}\land\varphi_2')\\
&\overline{\exists X_1(o\!=\!<\!\!c_1\parallel c_2\!\!>_k<\!\!\rho\!\!>_{env}<\!\!m\!\!>_{\mathit{mem}}\!\land\! \varphi_1 \!\land\! \varphi_2)\!\Downarrow\!\exists X_2(o\!=\!<\!\!\cdot\!\!>_k<\!\!\rho'\!\!>_{env}<\!\!m'\!\!>_{\mathit{mem}}\!\land\!\varphi_1'\!\land \!\varphi_2')}
\end{align*}
where $c_1,\;c_2\; \textrm{are interference-free}$. $\Gamma\equiv \exists X_1(o=<\!\!c_1\parallel c_2\!\!>_k<\!\!\rho\!\!>_{env}<\!\!m\!\!>_{\mathit{mem}}\land \varphi_1 \land \varphi_2)$ and $\Gamma'\equiv \exists X_2(o=<\!\!\cdot\!\!>_k<\!\!\rho'\!\!>_{env}<\!\!m'\!\!>_{\mathit{mem}}\land\varphi_1'\land \varphi_2')$. Suppose there exists $\theta_\tau:\mathit{SVar}\to \mathcal{T}$ with $\theta_\tau\!\!\upharpoonright_{\mathit{SVar}/X_1}=\tau\!\!\upharpoonright_{\mathit{SVar}/X_1}$ such that $\gamma_0=<\!\!c_1\parallel c_2\!\!>_k<\!\!\theta_\tau(\rho)\!\!>_{\mathit{env}}<\!\!\theta_\tau(m)\!\!>_{\mathit{mem}}$  and $\theta_\tau \models \varphi_1 \wedge \varphi_2$ and
$$<\!\!c_1\parallel c_2\!\!>_k<\!\!\theta_\tau(\rho)\!\!>_{\mathit{env}}<\!\!\theta_\tau(m)\!\!>_{\mathit{mem}}\xrightarrow{L}\gamma_1\xrightarrow{L}\cdots\xrightarrow{L}\gamma_i\xrightarrow{L}\gamma_{i+1}\xrightarrow{L}\cdots\xrightarrow{L}\gamma_n$$
Set $\gamma_i=<\!\!c_{1i}\parallel c_{2i}\!\!>_k<\!\!\rho_i\!\!>_{\mathit{env}}<\!\!m_i\!\!>_{\mathit{mem}}$, $0\leq i\leq n$ and $c_{10}=c_1$ and $c_{20}=c_2$ and $\rho_0=\theta_\tau(\rho)$ and $m_0=\theta_\tau(m)$. Proposition 3 implies
$$\gamma_0\xrightarrow{L}\gamma_1\xrightarrow{L}\cdots\xrightarrow{L}\gamma_i\xrightarrow{L}\gamma_{i+1}\xrightarrow{L}\cdots\xrightarrow{L}\gamma_n$$
can be broken down into two terminable
executions
$$\gamma_{10}\xrightarrow{L}\gamma_{11}\xrightarrow{L}\cdots\xrightarrow{L}\gamma_{1i}\xrightarrow{L}\gamma_{1(i+1)}\xrightarrow{L}\cdots\xrightarrow{L}\gamma_{1n}$$
and
$$\gamma_{20}\xrightarrow{L}\gamma_{21}\xrightarrow{L}\cdots\xrightarrow{L}\gamma_{2i}\xrightarrow{L}\gamma_{2(i+1)}\xrightarrow{L}\cdots\xrightarrow{L}\gamma_{2n}$$
where $\gamma_{1i}=<\!\!c_{1i}\!\!>_k<\!\!\rho_i\!\!>_{\mathit{env}}<\!\!m_i\!\!>_{\mathit{mem}}$ and $\gamma_{2i}=<\!\!c_{2i}\!\!>_k<\!\!\rho_i\!\!>_{\mathit{env}}<\!\!m_i\!\!>_{\mathit{mem}}$, $0\leq i<n$.
\begin{claim}
$\gamma_{10}\xrightarrow{L}\gamma_{11}\xrightarrow{L}\cdots\xrightarrow{L}\gamma_{1i}\xrightarrow{L}\gamma_{1(i+1)}\xrightarrow{L}\cdots\xrightarrow{L}\gamma_{1n}$
and
$\gamma_{20}\xrightarrow{L}\gamma_{21}\xrightarrow{L}\cdots\xrightarrow{L}\gamma_{2i}\xrightarrow{L}\gamma_{2(i+1)}\xrightarrow{L}\cdots\xrightarrow{L}\gamma_{2n}$
are actual executions.
\end{claim}
\begin{claimproof}
First, since $\sigma$ is an actual execution and $c_1$ is substructure of $c_1\parallel c_2$,
$c', c_1$ are "interference-free" where $c'\in S$ and $S$ represents a set of computations executed in parallel with $c_1\parallel c_2$. Second, $c_1, c_2$ are also "interference-free". Thus,
$\gamma_{10}\xrightarrow{L}\gamma_{11}\xrightarrow{L}\cdots\xrightarrow{L}\gamma_{1i}\xrightarrow{L}\gamma_{1i+1}\xrightarrow{L}\cdots\xrightarrow{L}\gamma_{1n}$
is an actual execution. Similarly,
$\gamma_{20}\xrightarrow{L}\gamma_{21}\xrightarrow{L}\cdots\xrightarrow{L}\gamma_{2i}\xrightarrow{L}\gamma_{2i+1}\xrightarrow{L}\cdots\xrightarrow{L}\gamma_{2n}$
is also an actual execution.
\end{claimproof}
By the induction  hypothesis of
$$\exists X_1(o=<\!\!c_1\!\!>_k<\!\!\rho\!\!>_{\mathit{env}}<\!\!m\!\!>_{\mathit{mem}}\land \varphi_1)\Downarrow\exists X_2(o=<\!\!\cdot\!\!>_k<\!\!\rho'\!\!>_{\mathit{env}}<\!\!m'\!\!>_{\mathit{mem}}\land \varphi_1')$$
we conclude $(\gamma_{1n},\tau)\models \exists X_2(o=<\!\!\cdot\!\!>_k<\!\!\rho'\!\!>_{\mathit{env}}<\!\!m'\!\!>_{\mathit{mem}}\land \varphi_1')$.\\
By the induction  hypothesis of
$$\exists X_1(o=<\!\!c_2\!\!>_k<\!\!\rho\!\!>_{\mathit{env}}<\!\!m\!\!>_{\mathit{mem}}\land \varphi_2)\Downarrow\exists X_2(o=<\!\!\cdot\!\!>_k<\!\!\rho'\!\!>_{\mathit{env}}<\!\!m'\!\!>_{\mathit{mem}}\land\varphi_2')$$
we conclude $(\gamma_{2n},\tau)\models \exists X_2(o=<\!\!\cdot\!\!>_k<\!\!\rho'\!\!>_{\mathit{env}}<\!\!m'\!\!>_{\mathit{mem}}\land\varphi_2')$.\\
Since $\gamma_{1n}=\gamma_{2n}=\gamma_{n}$, there exits $\theta_\tau^1:\mathit{SVar}\to \mathcal{T}$ with $\theta_\tau^1\!\!\upharpoonright_{\mathit{SVar}/X_2}=\tau\!\!\upharpoonright_{\mathit{SVar}/X_2}$ such that $\gamma_{1n}=\gamma_{2n}=<\!\!\cdot\!\!>_k<\!\!\theta_\tau^1(\rho')\!\!>_{\mathit{env}}<\!\!\theta_\tau^1(m')\!\!>_{\mathit{mem}}$  and $\theta_\tau^1 \models \varphi_1'$ and $\theta_\tau^1 \models \varphi_2'$. Therefor, $(\gamma_{n},\tau)\models \exists X_2(o=<\!\!\cdot\!\!>_k<\!\!\rho'\!\!>_{\mathit{env}}<\!\!m'\!\!>_{\mathit{mem}}\varphi_1'\land\varphi_2')$.\\
\textbf{Case12:} By M-SKIP rule
$$
\dfrac{\cdot}{\exists X(o=<\!\!skip\!\!>_k<\!\!\rho\!\!>_{\mathit{env}}<\!\!m\!\!>_{\mathit{mem}}\land \varphi)\Downarrow\exists X(o=<\!\!\cdot\!\!>_k<\!\!\rho\!\!>_{\mathit{env}}<\!\!m\!\!>_{\mathit{mem}}\land \varphi)}
$$
$\Gamma\equiv\exists X(o=<\!\!skip\!\!>_k<\!\!\rho\!\!>_{\mathit{env}}<\!\!m\!\!>_{\mathit{mem}}\land \varphi)$ and $\Gamma'\equiv\exists  X(o=\!<\!\!\cdot\!\!>_k<\!\!\rho\!\!>_{\mathit{env}}<\!\!m\!\!>_{\mathit{mem}}\land \varphi)$. Suppose there exists $\theta_\tau:\mathit{SVar}\to \mathcal{T}$ with $\theta_\tau\!\!\upharpoonright_{\mathit{SVar}/X}=\tau\!\!\upharpoonright_{\mathit{SVar}/X}$ such that $\gamma_0=<\!\!skip\!\!>_k<\!\!\theta_\tau(\rho)\!\!>_{\mathit{env}}<\!\!\theta_\tau(m)\!\!>_{\mathit{mem}}$  and $\theta_\tau \models \varphi$ and
$$
<\!\!skip\!\!>_k<\!\!\theta_\tau(\rho)\!\!>_{\mathit{env}}<\!\!\theta_\tau(m)\!\!>_{\mathit{mem}}\xrightarrow{L}\gamma_1\xrightarrow{L}\cdots\xrightarrow{L}\gamma_i\xrightarrow{L}\gamma_{i+1}\xrightarrow{L}\cdots\xrightarrow{L}\gamma_n
$$
If $<\!\!skip\!\!>_k<\!\!\theta_\tau(\rho)\!\!>_{\mathit{env}}<\!\!\theta_\tau(m)\!\!>_{\mathit{mem}}\xrightarrow{P}\gamma_1$, by SKIP rule, $\gamma_1=<\!\!\cdot\!\!>_k<\!\!\theta_\tau(\rho)\!\!>_{\mathit{env}}<\!\!\theta_\tau(m)\!\!>_{\mathit{mem}}$ and
$$<\!\!\cdot\!\!>_k<\!\!\theta_\tau(\rho)\!\!>_{\mathit{env}}<\!\!\theta_\tau(m)\!\!>_{\mathit{mem}}\xrightarrow{E}\cdots\xrightarrow{E}\gamma_i\xrightarrow{E}\gamma_{i+1}\xrightarrow{E}\cdots\xrightarrow{E}\gamma_n$$
Since $\theta_\tau \models \varphi$, $(\gamma_1,\tau)\models\Gamma'$. Since $\sigma$ is an actual execution,  $(\gamma_i,\tau)\models\Gamma'$, $2\leq i\leq n$.\\
If $<\!\!skip\!\!>_k<\!\!\theta_\tau(\rho)\!\!>_{\mathit{env}}<\!\!\theta_\tau(m)\!\!>_{\mathit{mem}}\xrightarrow{E}\gamma_1$, since $\sigma$ is an actual execution, $(\gamma_1,\tau)\models\Gamma$ and there exits $j$ ($1< j\leq n-1$) such that
$\gamma_1\xrightarrow{E}\cdots\xrightarrow{E}\gamma_j\xrightarrow{P}\gamma_{j+1}\xrightarrow{E}\cdots\xrightarrow{E}\gamma_n$
with $(\gamma_i,\tau)\models\Gamma$, $1\leq i\leq j$.
Suppose there exists $\theta_\tau^1:\mathit{SVar}\to \mathcal{T}$ with $\theta_\tau^1\!\!\upharpoonright_{\mathit{SVar}/X}=\tau\!\!\upharpoonright_{\mathit{SVar}/X}$ such that $\gamma_j=<\!\!skip\!\!>_k<\!\!\theta_\tau^1(\rho_j)\!\!>_{\mathit{env}}<\!\!\theta_\tau^1(m_j)\!\!>_{\mathit{mem}}$  and $\theta_\tau^1 \models \varphi$. By SKIP rule, we conclude $\gamma_{j+1}=<\!\!\cdot\!\!>_k<\!\!\theta_\tau^1(\rho_j)\!\!>_{\mathit{env}}<\!\!\theta_\tau^1(m_j)\!\!>_{\mathit{mem}}$. Due to $\theta_\tau^1 \models \varphi$, $(\gamma_{j+1},\tau)\models\Gamma'$. Since $\sigma$ is an actual execution,  $(\gamma_i,\tau)\models\Gamma'$, $j+1\leq i\leq n$.
\end{proof}
\section{Matching logic verifier for PIMP}
\begin{figure}
  \scriptsize
  \begin{tabular}{l}
    \\[-2mm]
    \hline
    \hline\\[-2mm]
    {\bf \small Matching logic verifier for PIMP}\\
    \hline\\[-2mm]
    \vspace{1mm}
   \textbf{Abstract Syntax:} \\
     \vspace{1mm}
  the same as Rewrite Theory Of PIMP, but adding:
  $\mathit{C}::=\cdots\mid \textrm{assert ass}$\\
 \textbf{Configurtation:}\\
 \vspace{1mm}
 the same as Rewrite Theory Of PIMP, but adding: \\
 $\mathit{PatternItem}::=<\!\!\mathit{C}\!\!>_k\mid<\!\!\mathit{Env}\!\!>_{env}\mid<\!\!\mathit{Mem}\!\!>_{mem}\mid<\!\!Form\!\!>_{form}\mid<\!\!Set_{\cdot}^{-,-}[SVar]\!\!>_{bnd}$\\
 \vspace{1mm}
 $\mathit{Pattern}::=<\mathit{Bag}_.^{-,-}[\mathit{PatternItem}]>$\\
 \vspace{1mm}
$\mathit{Ass}::=<\mathit{Env}>_{env}<\!\!\mathit{Mem}\!\!>_{mem}<Set_{\cdot}^{-,-}[SVar]>_{bnd}<Form>_{form}$\\
 \vspace{1mm}
$\mathit{Top}::=<\mathit{Set}_.^{-,-}[\mathit{pattern}]>$\\
 \textbf{Semantic Rules:}\\\\
 \textrm{V-NULL:}\\
   $\dfrac{\cdot}{<\!\!<\!\!\cdot\!\!>_k<\!\!\rho\!\!>_{env}<\!\!m\!\!>_{mem}<\!\!X\!\!>_{bnd}<\!\!\varphi\!\!>_{form}\!\!>\rightarrow<\!\!\cdot\!\!>}$\\\\
   \textrm{V-FALSE:}\\
    $\dfrac{\cdot}{<\!\!<\!\!c_1;c_2\!\!>_k<\!\!\rho\!\!>_{env}<\!\!m\!\!>_{mem}<\!\!X\!\!>_{bnd}<\!\!false\!\!>_{form}\!\!>\rightarrow<\!\!\cdot\!\!>}$\\\\
    \textrm{V-SKIP:}\\
    $\dfrac{\cdot}{<\!\!<\!\!skip;c\!\!>_k<\!\!\rho\!\!>_{env}<\!\!m\!\!>_{mem}<\!\!X\!\!>_{bnd}<\!\!\varphi\!\!>_{form}\!\!>\rightarrow<\!\!<\!\!c\!\!>_k<\!\!\rho\!\!>_{env}<\!\!m\!\!>_{mem}<\!\!X\!\!>_{bnd}<\!\!\varphi\!\!>_{form}\!\!>}$\\\\
    \textrm{V-ASSERT:}\\\\
    $\dfrac{\models<\!\!c\!\!>_k ass_2  \Rightarrow<\!\!c\!\!>_k ass_1}{<\!\!<\!\!\textrm{assert } ass_1;\;c\!\!>_k \;ass_2\!\!>\rightarrow<\!\!<\!\!c\!\!>_k ass_1\!\!>}$\\\\
    \textrm{V-ASGN1:}\\
    $\dfrac{\cdot}{<\!\!<\!\!\mathrm{x}\!:=\!e;c\!\!>_k<\!\!\rho\!\!>_{e\!n\!v}<\!\!m\!\!>_{m\!e\!m}<\!\!X\!\!>_{b\!n\!d}<\!\!\varphi\!\!>_{f\!o\!r\!m}\!\!>\!\rightarrow<\!\!<\!\!c\!\!>_k<\!\!\rho[\rho(e)/\mathrm{x}]\!\!>_{e\!n\!v}<\!\!m\!\!>_{m\!e\!m}<\!\!X\!\!>_{b\!n\!d}<\!\!\varphi\!\!>_{f\!o\!r\!m}\!\!>}$\\\\
   \textrm{V-ASGN2:}\\
   $\dfrac{\cdot}{<\!\!<\!\!\mathrm{x}:=A[e];c\!\!>_k<\!\!\rho\!\!>_{env}<\!\!m\!\!>_{mem}<\!\!X\!\!>_{bnd}<\!\!\varphi\!\!>_{form}\!\!>\rightarrow\qquad\qquad\qquad}$\\
   $<\!\!<\!\!c\!\!>_k<\!\!\rho[m(\rho(A)+_\mathit{Int}\rho(e))/\mathrm{x}]\!\!>_{env}<\!\!m\!\!>_{mem}<\!\!X\!\!>_{bnd}<\!\!\varphi\!\!>_{form}\!\!>$\\\\
   \textrm{V-ASGN3:}\\
   $\dfrac{\cdot}{<\!\!<\!\!A[e_1]:=e_2;c\!\!>_k<\!\!\rho\!\!>_{env}<\!\!m\!\!>_{mem}<\!\!X\!\!>_{bnd}<\!\!\varphi\!\!>_{form}\!\!>\rightarrow\qquad\qquad\qquad}$\\
   $<\!\!<\!\!c\!\!>_k<\!\!\rho\!\!>_{env}<\!\!m[\rho(e_2)/(\rho(A)+_\mathit{Int}\rho(e_1))]\!\!>_{mem}<\!\!X\!\!>_{bnd}<\!\!\varphi\!\!>_{form}\!\!>$\\\\
    \textrm{V-ARRAY:}\\
    $\dfrac{\cdot}{<\!\!<\!\!A:=\mathit{Array}(\overline{e});c\!\!>_k<\!\!\rho\!\!>_{env}<\!\!m\!\!>_{mem}<\!\!X\!\!>_{bnd}<\!\!\varphi\!\!>_{form}\!\!>\rightarrow}$\\
    $<\!\!<\!\!c\!\!>_k<\!\!\rho[p/A]\!\!>_{env}<\!\!p\mapsto[\rho(\overline{e})],m\!\!>_{mem}<\!\!X\cup \{p\}\!\!>_{bnd}<\!\!\varphi\!\!>_{form}\!\!>$\\\\
\textrm{V-IF:}\\
$\dfrac{\cdot}{<\!\!<\!\!(\textrm{if}(b)c_1\textrm{else}\,c_2);c_3\!\!>_k<\!\!\rho\!\!>_{env}<\!\!m\!\!>_{mem}<\!\!X\!\!>_{bnd}\!<\!\!\varphi\!\!>_{form}\!\!>\rightarrow\qquad\qquad}$\\
$<<\!\!c_1;c_3\!\!>_k<\!\!\rho\!\!>_{env}<\!\!m\!\!>_{mem}<\!\!X\!\!>_{bnd}\!<\!\!\varphi\wedge (\rho(b)\;\;is\;\;true)\!\!>_{form},$\\
$<\!\!c_2;c_3\!\!>_k<\!\!\rho\!\!>_{env}<\!\!m\!\!>_{mem}<\!\!X\!\!>_{bnd}\!<\!\!\varphi\wedge (\rho(b)\;\;is\;\;false)\!\!>_{form}>$\\\\
\textrm{V-WHILE:}\\
$\dfrac{\cdot}{<\!\!<\!\!(\textrm{while }b\textrm{ do }c_1);c_2\!\!>_k<\!\!\rho\!\!>_{env}<\!\!m\!\!>_{mem}<\!\!X\!\!>_{bnd}\!<\!\!\varphi\!\!>_{form}\!\!>\rightarrow\qquad\qquad\qquad\qquad\qquad\qquad\qquad\qquad}$\\
$<\!\!<\!\!c_1;\textrm{assert } (<\!\!\rho\!\!>_{e\!n\!v}<\!\!m\!\!>_{m\!e\!m}<\!\!X\!\!>_{b\!n\!d}\!<\!\!\varphi\!\!>_{f\!o\!r\!m})\!\!>_k<\!\!\rho\!\!>_{e\!n\!v}<\!\!m\!\!>_{m\!e\!m}<\!\!X\!\!>_{b\!n\!d}\!<\!\!\varphi\!\wedge \! (\rho(b)\;\;is\;\;true)\!\!>_{f\!o\!r\!m},$\\
$<\!\!c_2\!\!>_k<\!\!\rho\!\!>_{env}<\!\!m\!\!>_{mem}<\!\!X\!\!>_{bnd}\!<\!\!\varphi\wedge (\rho(b)\;\;is\;\;false)\!\!>_{form}>$\\\\
\textrm{V-AWAIT:}\\
$\dfrac{\cdot}{<\!\!<\!\!\textrm{await } b\textrm{ then } cc;c_2\!\!>_k<\!\!\rho\!\!>_{env}<\!\!m\!\!>_{mem}<\!\!X\!\!>_{bnd}\!<\!\!\varphi\!\!>_{form}\!\!>\rightarrow\qquad}$\\
$<\!\!<\!\!cc;c_2\!\!>_k<\!\!\rho\!\!>_{env}<\!\!m\!\!>_{mem}<\!\!X\!\!>_{bnd}\!<\!\!\varphi\wedge(\rho(b)\;\;is\;\;true)\!\!>_{form}\!\!>$\\\\
\textrm{V-PAR:}\\
$\dfrac{c_1,\;c_2\; \textrm{are interference-free}}{<\!\!<\!\!(c_1\parallel c_2);c_3\!\!>_k<\!\!\rho\!\!>_{env}<\!\!m\!\!>_{mem}<\!\!X\!\!>_{bnd}\!<\!\!\varphi_1\wedge\varphi_2\!\!>_{form}\!\!>\rightarrow\qquad\qquad\qquad\qquad\qquad\qquad}$\\
$<<\!\!c_1; assert(result(c_1,\rho,m,X,\varphi_1))\!\!>_k<\!\!\rho\!\!>_{env}<\!\!m\!\!>_{mem}<\!\!X\!\!>_{bnd}\!<\!\!\varphi_1\!\!>_{form},$\\
$<\!\!c_2;assert(result(c_2,\rho,m,X,\varphi_2))\!\!>_k<\!\!\rho\!\!>_{env}<\!\!m\!\!>_{mem}<\!\!X\!\!>_{bnd}\!<\!\!\varphi_2\!\!>_{form},$\\
$<<\!\!c_3\!\!>_k result(c_1,\rho,m,X,\varphi_1)\cap result(c_2,\rho,m,X,\varphi_2)>$\\\\
 \\
    \hline
    \hline
  \end{tabular}
  \caption{Matching logic verifier for PIMP}\label{Fig1}
\end{figure}
Intuitively, the matching logic verifier for PIMP is to execute rewriting logic
semantics symbolically on configuration patterns that we define in section 3. Figure 7 shows matching logic verifier for PIMP. First, we add $\mathit{patternItem}$ sort as algebraic infrastructure for pattern. Compared to $\mathit{Cfg}$ sort, two new sub cells are added, one to hold the bound variables set and the other to hold the constraints. Secondly, we add program annotations. $assert \; ass$ acts as program annotation and can be inserted at any place in the computation. Sort $Ass$ don't include the $<\!\!\cdots\!\!>_k$ because the purpose of program annotation is to describe the state of the program. We also add $\mathit{Top}$ sort to wrap a set of patterns such as $<\!\!\Gamma_1,\Gamma_2,\cdots,\Gamma_n\!\!>$ and each $\Gamma_i$ is a pattern with program annotation. Final, we introduce semantic rules. V-NULL and V-FALSE rules are used to dissolve pattern. A pattern $\Gamma_i$ is considered verified either when it translate into $<\!\!\cdot\!\!>_k<\!\!\rho\!\!>_{env}<\!\!m\!\!>_{mem}<\!\!X\!\!>_{bnd}<\!\!\varphi\!\!>_{form}$, meaning that all the program annotations have been validated,
or when it is found that constraint yield false. $<\!\!\Gamma_1,\Gamma_2,\cdots,\Gamma_n\!\!>$ is emptied means that the program is fully verified. V-ASSERT rule replaces pattern's $ass_2$ with $ass_1$, which is particularly important when $ass_1$ is a loop invariant. V-IF rule divides the current pattern into two patterns, corresponding to the two cases that the assumption about $b$ is true or false. V-WHILE rule assumes that ass1 of the current pattern is invariant and generates two branches.
In one branch, $\rho(b)$ is true indicating the invariant holds, and in the other branch, $\rho(b)$ is false ignoring the computation of current pattern. In V-AWAIT rule, $cc$ is also executed as an indivisible action. In V-PARALLEL rule,  $\mathit{result}(\_):\mathit{PatternItem}\rightarrow \mathit{Ass}$ is a function. In later chapters, we abbreviate $\mathit{result}(<\!\!c\!\!>_k<\!\!\rho\!\!>_{env}<\!\!m\!\!>_{mem}<\!\!X\!\!>_{bnd}<\!\!\varphi\!\!>_{form})$ as $\mathit{result}(c,\rho,m,X,\varphi)$. If
$$<\!\!<\!\!c_1\!\!>_k<\!\!\rho\!\!>_{env}<\!\!m\!\!>_{mem}<\!\!X\!\!>_{bnd}<\!\!\varphi_1\!\!>_{form}\!\!>\rightarrow^*<\!\!<\!\!\cdot\!\!>_k<\!\!\rho'\!\!>_{env}<\!\!m'\!\!>_{mem}<\!\!X'\!\!>_{bnd}<\!\!\varphi_1'\!\!>_{form}\!\!>$$
$$<\!\!<\!\!c_2\!\!>_k<\!\!\rho\!\!>_{env}<\!\!m\!\!>_{mem}<\!\!X\!\!>_{bnd}<\!\!\varphi_1\!\!>_{form}\!\!>\rightarrow^*<\!\!<\!\!\cdot\!\!>_k<\!\!\rho'\!\!>_{env}<\!\!m'\!\!>_{mem}<\!\!X'\!\!>_{bnd}<\!\!\varphi_2'\!\!>_{form}\!\!>$$
then
$$\mathit{result}(c_1,\rho,m,X,\varphi_1)=<\!\!\rho'\!\!>_{env}<\!\!m'\!\!>_{mem}<\!\!X'\!\!>_{bnd}<\!\!\varphi_1'\!\!>_{form}$$
$$\mathit{result}(c_2,\rho,m,X,\varphi_2)=<\!\!\rho'\!\!>_{env}<\!\!m'\!\!>_{mem}<\!\!X'\!\!>_{bnd}<\!\!\varphi_2'\!\!>_{form}$$
$$\mathit{result}(c_1,\rho,m,X,\varphi_1)\cap\mathit{result}(c_2,\rho,m,X,\varphi_2)=<\!\!\rho'\!\!>_{env}<\!\!m'\!\!>_{mem}<\!\!X'\!\!>_{bnd}<\!\!\varphi_1'\wedge\varphi_2'\!\!>_{form}$$
\begin{theorem}
Given an annotated computation $k\in \mathrm{C}$, $\overline{k} \in \mathrm{C}$ is the computation obtained by removing all $assert \; ass$ from $k$, the following holds:
\begin{enumerate}
  \item If $<\!\!<\!\!k; assert\; ass_{post}\!\!>_k ass_{pre}\!\!>\to^*<\!\!\cdot\!\!>$, then $<\!\!\overline{k}\!\!>_k ass_{pre}\Downarrow <\!\!\cdot\!\!>_k ass_{post}$ is derivable from matching logic;
  \item If $<\!\!k_1\!\!>_k ass_{pre}\Downarrow <\!\!\cdot\!\!>_k ass_{post}$ is derivable from matching logic, then there is an annotated computation $k_2$ such that $\overline{k_2}=k_1$ and $<\!\!<\!\!k_2; assert\;ass_{post}>_k ass_{pre}\!\!>\to^*<\!\!\cdot\!\!>$.
\end{enumerate}
\end{theorem}
\begin{proof}
Suppose $\Gamma\equiv<\!\!<\!\!k; assert\; ass_{post}\!\!>_k ass_{pre}\!\!>$ and $\Gamma\to^*<\!\!\cdot\!\!>$. The original pattern $\Gamma$ iterative
rewrite it to $<\!\!\Gamma_1,\Gamma_2,\cdots,\Gamma_n\!\!>$. Since rewrite logic allows parallel rewriting, we require $\Gamma_1,\Gamma_2,\cdots,\Gamma_n$ parallel rewriting. Therefor, we get a rewrite tree with $\Gamma$ as root node and $\cdot$ as leaf nodes. Notice that $\Gamma_i(0<i<n+1)$ is also a subtree which  implies $\Gamma_i\to^* <\!\!\cdot\!\!>$. We prove conclusion 1 by induction on the depth of the rewrite tree.
According to the rewriting rule adopted in the first step of $\Gamma\to^*<\!\!\cdot\!\!>$, we distinguish different cases:\\
\textbf{Case1:} By V-FALSE rule
$$<\!\!<\!\!c; assert\; ass_{post}\!\!>_k<\!\!\rho\!\!>_{env}<\!\!m\!\!>_{mem}<\!\!X\!\!>_{bnd}<\!\!false\!\!>_{form}\!\!>\rightarrow<\!\!\cdot\!\!>$$
$\models <\!\! \overline{c}\!\!>_k<\!\!\rho\!\!>_{env}<\!\!m\!\!>_{mem}<\!\!X\!\!>_{bnd}<\!\!false\!\!>_{form} \Rightarrow <\!\!\cdot\!\!>_k ass_{post}$ is tautology. Since $<\!\!\cdot\!\!>_k ass_{post} \Downarrow <\!\!\cdot\!\!>_k ass_{post}$ ,
$<\!\!\overline{c}\!\!>_k<\!\!\rho\!\!>_{env}<\!\!m\!\!>_{mem}<\!\!X\!\!>_{bnd}<\!\!false\!\!>_{form} \Downarrow <\!\!\cdot\!\!>_k ass_{post}$
is derivable from matching logic by M-CONS rule.\\
\textbf{Case2:} By V-ASGN1 rule
$$
<\!\!<\!\!\mathrm{x}:=e;k'; assert\; ass_{post}\!\!>_k<\!\!\rho\!\!>_{env}<\!\!m\!\!>_{mem}<\!\!X\!\!>_{bnd}<\!\!\varphi\!\!>_{form}\!\!>\rightarrow$$
$$<\!\!<\!\!k'; assert\; ass_{post}\!\!>_k<\!\!\rho[\rho(e)/\mathrm{x}]\!\!>_{env}<\!\!m\!\!>_{mem}<\!\!X\!\!>_{bnd}<\!\!\varphi\!\!>_{form}\!\!>\to^*<\!\!\cdot\!\!>$$
By the induction hypothesis of
$$<\!\!<\!\!k'; assert\; ass_{post}\!\!>_k<\!\!\rho[\rho(e)/\mathrm{x}]\!\!>_{env}<\!\!m\!\!>_{mem}<\!\!X\!\!>_{bnd}<\!\!\varphi\!\!>_{form}\!\!>\rightarrow^*<\!\!\cdot\!\!>$$
we conclude $<\!\!\overline{k'}\!\!>_k<\!\!\rho[\rho(e)/\mathrm{x}]\!\!>_{env}<\!\!m\!\!>_{mem}<\!\!X\!\!>_{bnd}<\!\!\varphi\!\!>_{form} \Downarrow <\!\!\cdot\!\!>_k ass_{post}$.
M-ASGN1 rule implies $$<\!\!\mathrm{x}:=e\!\!>_k<\!\!\rho\!\!>_{env}<\!\!m\!\!>_{mem}<\!\!X\!\!>_{bnd}<\!\!\varphi\!\!>_{form} \Downarrow <\!\!\cdot\!\!>_k<\!\!\rho[\rho(e)/\mathrm{x}]\!\!>_{env}<\!\!m\!\!>_{mem}<\!\!X\!\!>_{bnd}<\!\!\varphi\!\!>_{form}$$
Since $\overline{\mathrm{x}:=e;k'}=\mathrm{x}:=e;\overline{k'}$, by M-SEQ rule, we conclude
$$<\!\!\overline{x:=e;k'}\!\!>_k<\!\!\rho\!\!>_{env}<\!\!m\!\!>_{mem}<\!\!X\!\!>_{bnd}<\!\!\varphi\!\!>_{form} \Downarrow <\!\!\cdot\!\!>_k ass_{post}$$
is derivable from matching logic.\\
\textbf{Case3:} By V-ASGN2 rule
$$<\!\!<\!\!\mathrm{x}:=A[e];k'; assert\; ass_{post}\!\!>_k<\!\!\rho\!\!>_{env}<\!\!m\!\!>_{mem}<\!\!X\!\!>_{bnd}<\!\!\varphi\!\!>_{form}\!\!>\rightarrow$$
$$<\!\!<\!\!k'; assert\; ass_{post}\!\!>_k<\!\!\rho[m(\rho(A)+_\mathit{Int}\rho(e))/\mathrm{x}]\!\!>_{env}<\!\!m\!\!>_{mem}<\!\!X\!\!>_{bnd}<\!\!\varphi\!\!>_{form}\!\!>\to^*<\!\!\cdot\!\!>$$
By the induction hypothesis of
$$<\!\!<\!\!k'; assert\; ass_{post}\!\!>_k<\!\!\rho[m(\rho(A)+_\mathit{Int}\rho(e))/\mathrm{x}]\!\!>_{env}<\!\!m\!\!>_{mem}<\!\!X\!\!>_{bnd}<\!\!\varphi\!\!>_{form}\!\!>\rightarrow^*<\!\!\cdot\!\!>$$
we conclude
$$<\!\!\overline{k'}\!\!>_k<\!\!\rho[m(\rho(A)+_\mathit{Int}\rho(e))/\mathrm{x}]\!\!>_{env}<\!\!m\!\!>_{mem}<\!\!X\!\!>_{bnd}<\!\!\varphi\!\!>_{form}\Downarrow <\!\!\cdot\!\!>_k ass_{post}$$
M-ASGN2 rule implies
$$<\!\!\mathrm{x}:=A[e]\!\!>_k<\!\!\rho\!\!>_{env}<\!\!m\!\!>_{mem}<\!\!X\!\!>_{bnd}<\!\!\varphi\!\!>_{form}\Downarrow$$
$$<\!\!\cdot\!\!>_k<\!\!\rho[m(\rho(A)+_\mathit{Int}\rho(e))/\mathrm{x}]\!\!>_{env}<\!\!m\!\!>_{mem}<\!\!X\!\!>_{bnd}<\!\!\varphi\!\!>_{form}$$
Since $\overline{\mathrm{x}:=A[e];k'}=\mathrm{x}:=A[e];\overline{k'}$, by M-SEQ rule, we conclude
$$<\!\!\overline{\mathrm{x}:=A[e];k'}\!\!>_k<\!\!\rho\!\!>_{env}<\!\!m\!\!>_{mem}<\!\!X\!\!>_{bnd}<\!\!\varphi\!\!>_{form}\Downarrow <\!\!\cdot\!\!>_k ass_{post}$$
is derivable from matching logic.\\
\textbf{Case4:} By V-ASGN3 rule
$$<\!\!<\!\!A[e_1]:=e_2;k'; assert\; ass_{post}\!\!>_k<\!\!\rho\!\!>_{env}<\!\!m\!\!>_{mem}<\!\!X\!\!>_{bnd}<\!\!\varphi\!\!>_{form}\!\!>\rightarrow$$
$$<\!\!<\!\!k'; assert\; ass_{post}\!\!>_k<\!\!\rho\!\!>_{env}<\!\!m[\rho(e_2)/(\rho(A)+_\mathit{Int}\rho(e_1))]\!\!>_{mem}<\!\!X\!\!>_{bnd}<\!\!\varphi\!\!>_{form}\!\!>\to^*<\!\!\cdot\!\!>$$
By the induction hypothesis of
$$<\!\!<\!\!k'; assert\; ass_{post}\!\!>_k<\!\!\rho\!\!>_{env}<\!\!m[\rho(e_2)/(\rho(A)+_\mathit{Int}\rho(e_1))]\!\!>_{mem}<\!\!X\!\!>_{bnd}<\!\!\varphi\!\!>_{form}\!\!>\rightarrow^*<\!\!\cdot\!\!>$$
we conclude
$$<\!\!\overline{k'}\!\!>_k<\!\!\rho\!\!>_{env}<\!\!m[\rho(e_2)/(\rho(A)+_\mathit{Int}\rho(e_1))]\!\!>_{mem}<\!\!X\!\!>_{bnd}<\!\!\varphi\!\!>_{form}\Downarrow <\!\!\cdot\!\!>_k ass_{post}$$
M-ASGN3 rule implies
$$<\!\!A[e_1]:=e_2\!\!>_k<\!\!\rho\!\!>_{env}<\!\!m\!\!>_{mem}<\!\!X\!\!>_{bnd}<\!\!\varphi\!\!>_{form}\Downarrow$$
$$<\!\!\cdot\!\!>_k<\!\!\rho\!\!>_{env}<\!\!m[\rho(e_2)/(\rho(A)+_\mathit{Int}\rho(e_1))]\!\!>_{mem}<\!\!X\!\!>_{bnd}<\!\!\varphi\!\!>_{form}
$$
Since $\overline{A[e_1]:=e_2;k'}=A[e_1]:=e_2;\overline{k'}$, by M-SEQ rule, we conclude
$$<\!\!<\!\!\overline{A[e_1]:=e_2;k'}\!\!>_k<\!\!\rho\!\!>_{env}<\!\!m\!\!>_{mem}<\!\!X\!\!>_{bnd}<\!\!\varphi\!\!>_{form}\Downarrow <\!\!\cdot\!\!>_k ass_{post}$$
is derivable from matching logic.\\
\textbf{Case5:} By V-ARRAY rule
$$<\!\!<\!\!A:=\mathit{Array}(\overline{e});k'; assert\; ass_{post}\!\!>_k<\!\!\rho\!\!>_{env}<\!\!m\!\!>_{mem}<\!\!X\!\!>_{bnd}<\!\!\varphi\!\!>_{form}\!\!>\rightarrow$$
$$<\!\!<\!\!k'; assert\; ass_{post}\!\!>_k<\!\!\rho[p/A]\!\!>_{env}<\!\!p\mapsto[\rho(\overline{e})],m\!\!>_{mem}<\!\!X\cup \{p\}\!\!>_{bnd}<\!\!\varphi\!\!>_{form}\!\!>\to^*<\!\!\cdot\!\!>$$
By the induction hypothesis of
$$<\!\!<\!\!k'; assert\; ass_{post}\!\!>_k<\!\!\rho[p/A]\!\!>_{env}<\!\!p\mapsto[\rho(\overline{e})],m\!\!>_{mem}<\!\!X\cup \{p\}\!\!>_{bnd}<\!\!\varphi\!\!>_{form}\!\!>\rightarrow^*<\!\!\cdot\!\!>$$
we conclude
$$<\!\!\overline{k'}\!\!>_k<\!\!\rho[p/A]\!\!>_{env}<\!\!p\mapsto[\rho(\overline{e})],m\!\!>_{mem}<\!\!X\cup \{p\}\!\!>_{bnd}<\!\!\varphi\!\!>_{form}\!\!>\Downarrow <\!\!\cdot\!\!>_k ass_{post}$$
M-ARRAY rule implies
$$<\!\!A:=\mathit{Array}(\overline{e})\!\!>_k<\!\!\rho\!\!>_{env}<\!\!m\!\!>_{mem}<\!\!X\!\!>_{bnd}<\!\!\varphi\!\!>_{form}\Downarrow$$
$$<\!\!\cdot\!\!>_k<\!\!\rho[p/A]\!\!>_{env}<\!\!p\mapsto[\rho(\overline{e})],m\!\!>_{mem}<\!\!X\cup \{p\}\!\!>_{bnd}<\!\!\varphi\!\!>_{form}
$$
Since $\overline{A:=\mathit{Array}(\overline{e});k'}=A:=\mathit{Array}(\overline{e});\overline{k'}$,  by M-SEQ rule, we conclude
$$<\!\!\overline{A:=\mathit{Array}(\overline{e});k'}\!\!>_k<\!\!\rho\!\!>_{env}<\!\!m\!\!>_{mem}<\!\!X\!\!>_{bnd}<\!\!\varphi\!\!>_{form}\Downarrow<\!\!\cdot\!\!>_k ass_{post}$$
is derivable from matching logic.\\
\textbf{Case6:} By V-IF rule
$$<\!\!<\!\!(\textrm{if}(b)c_1\textrm{else}\,c_2);k'; assert\; ass_{post}\!\!>_k<\!\!\rho\!\!>_{env}<\!\!m\!\!>_{mem}<\!\!X\!\!>_{bnd}\!<\!\!\varphi\!\!>_{form}\!\!>\rightarrow$$
$$<<\!\!c_1;k'; assert\; ass_{post}\!\!>_k<\!\!\rho\!\!>_{env}<\!\!m\!\!>_{mem}<\!\!X\!\!>_{bnd}<\!\!\varphi\wedge (\rho(b)\;\;is\;\;true)\!\!>_{form},$$
$$<\!\!c_2;k'; assert\; ass_{post}\!\!>_k<\!\!\rho\!\!>_{env}<\!\!m\!\!>_{mem}<\!\!X\!\!>_{bnd}<\!\!\varphi\wedge (\rho(b)\;\;is\;\;false)\!\!>_{form}>\to^*<\!\!\cdot\!\!>$$
Let $\Gamma_1 \equiv <\!\!k\!\!>_k<\!\!\rho_1\!\!>_{env}<\!\!m_1\!\!>_{mem}<\!\!X_1\!\!>_{bnd}<\!\!\varphi_1\!\!>_{form}$ and
$\Gamma_2 \equiv <\!\!k\!\!>_k<\!\!\rho_2\!\!>_{env}<\!\!m_2\!\!>_{mem}<\!\!X_2\!\!>_{bnd}<\!\!\varphi_2\!\!>_{form}$, $\Gamma_1\nabla\Gamma_2$ is a new pattern:
$$<\!\!k\!\!>_k<\!\!\rho\!\!>_{env}<\!\!m\!\!>_{mem}<\!\!X_1\cup X_2\cup \rho\cup m\!\!>_{bnd}<\!\!(\varphi_1\wedge\rho\!=\!\rho_1\wedge m\!=\!m_1)\vee(\varphi_2\wedge\rho\!=\!\rho_2\wedge m\!=\!m_2)\!\!>_{form}$$
Obviously, $\models \Gamma_1\nabla\Gamma_2 \Leftrightarrow \Gamma_1\vee \Gamma_2$. By the induction hypothesis of
$$<\!\!<\!\!c_1;k'; assert\; ass_{post}\!\!>_k<\!\!\rho\!\!>_{env}<\!\!m\!\!>_{mem}<\!\!X\!\!>_{bnd}<\!\!\varphi\wedge (\rho(b)\;\;is\;\;true)\!\!>_{form}\!\!>\rightarrow^*<\!\!\cdot\!\!>$$
we conclude
$$<\!\!\overline{c_1};\overline{k'}\!\!>_k<\!\!\rho\!\!>_{env}<\!\!m\!\!>_{mem}<\!\!X\!\!>_{bnd}<\!\!\varphi\wedge (\rho(b)\;\;is\;\;true)\!\!>_{form}\Downarrow<\!\!\cdot\!\!>_k ass_{post}$$
and
$$<\!\!\overline{c_2};\overline{k'}\!\!>_k<\!\!\rho\!\!>_{env}<\!\!m\!\!>_{mem}<\!\!X\!\!>_{bnd}<\!\!\varphi\wedge (\rho(b)\;\;is\;\;false)\!\!>_{form}\Downarrow<\!\!\cdot\!\!>_k ass_{post}$$
by the induction hypothesis of
$$<\!\!<\!\!c_2;k'; assert\; ass_{post}\!\!>_k<\!\!\rho\!\!>_{env}<\!\!m\!\!>_{mem}<\!\!X\!\!>_{bnd}<\!\!\varphi\wedge (\rho(b)\;\;is\;\;false)\!\!>_{form}\!\!>\rightarrow^*<\!\!\cdot\!\!>$$
there must have appropriate $ass_1$ and $ass_2$ such that
$$<\!\!\overline{c_1}\!\!>_k<\!\!\rho\!\!>_{env}<\!\!m\!\!>_{mem}<\!\!X\!\!>_{bnd}<\!\!\varphi\wedge (\rho(b)\;\;is\;\;true)\!\!>_{form}\Downarrow<\!\!\cdot\!\!>_k ass_1$$
and
$$<\!\!\overline{k'}\!\!>_k ass_1\Downarrow<\!\!\cdot\!\!>_k ass_{post}$$
and
$$<\!\!\overline{c_2}\!\!>_k<\!\!\rho\!\!>_{env}<\!\!m\!\!>_{mem}<\!\!X\!\!>_{bnd}<\!\!\varphi\wedge (\rho(b)\;\;is\;\;false)\!\!>_{form}\Downarrow<\!\!\cdot\!\!>_k ass_2$$
and
$$<\!\!\overline{k'}\!\!>_k ass_2\Downarrow<\!\!\cdot\!\!>_k ass_{post}$$
Suppose $ass_1 \equiv <\!\!\rho_1\!\!>_{env}<\!\!m_1\!\!>_{mem}<\!\!X_1\!\!>_{bnd}<\!\!\varphi_1\!\!>_{form}$ and
$ass_2 \equiv <\!\!\rho_2\!\!>_{env}<\!\!m_2\!\!>_{mem}<\!\!X_2\!\!>_{bnd}<\!\!\varphi_2\!\!>_{form}$,
$$ass_1 \nabla ass_2 \equiv <\!\!\rho\!\!>_{env}<\!\!X_1\cup X_2\cup \rho\cup m\!\!>_{bnd}<\!\!(\varphi_1\wedge\rho\!=\!\rho_1\wedge m\!=\!m_1)\vee(\varphi_2\wedge\rho\!=\!\rho_2\wedge m\!=\!m_2)\!\!>_{form}$$
Since $\models <\!\! \overline{k'}\!\!>_k ass_1 \nabla ass_2 \Leftrightarrow (<\!\! \overline{k'}\!\!>_k ass_1)\vee (<\!\! \overline{k'}\!\!>_k ass_2)$ and $<\!\!\overline{k'}\!\!>_k ass_1 \Downarrow <\!\!\cdot\!\!>_k ass_{post}$
and $<\!\!\overline{k'}\!\!>_k ass_2 \Downarrow <\!\!\cdot\!\!>_k ass_{post}$, M-CASE rule implies
$$<\!\! \overline{k'}\!\!>_k ass_1 \nabla ass_2 \Downarrow <\!\!\cdot\!\!>_k ass_{post}$$
Since
$\models <\!\!\cdot\!\!>_k ass_1 \nabla ass_2 \Leftrightarrow (<\!\! \cdot\!\!>_k ass_1)\vee (<\!\! \cdot\!\!>_k ass_2)$ and $\models <\!\!\cdot\!\!>_k ass_1 \Rightarrow (<\!\! \cdot\!\!>_k ass_1)\vee (<\!\! \cdot\!\!>_k ass_2)$ and
$$<\!\!\overline{c_1}\!\!>_k<\!\!\rho\!\!>_{env}<\!\!m\!\!>_{mem}<\!\!X\!\!>_{bnd}<\!\!\varphi\wedge (\rho(b)\;\;is\;\;true)\!\!>_{form}\Downarrow<\!\!\cdot\!\!>_k ass_1$$ applying M-CONS rule, we conclude
$$<\!\!\overline{c_1}\!\!>_k<\!\!\rho\!\!>_{env}<\!\!m\!\!>_{mem}<\!\!X\!\!>_{bnd}<\!\!\varphi\wedge (\rho(b)\;\;is\;\;true)\!\!>_{form}\Downarrow<\!\!\cdot\!\!>_k ass_1 \nabla ass_2$$
Since
$\models <\!\!\cdot\!\!>_k ass_1 \nabla ass_2 \Leftrightarrow (<\!\! \cdot\!\!>_k ass_1)\vee (<\!\! \cdot\!\!>_k ass_2)$ and $\models <\!\!\cdot\!\!>_k ass_2 \Rightarrow (<\!\! \cdot\!\!>_k ass_1)\vee (<\!\!\cdot\!\!>_k ass_2)$ and
$$<\!\!\overline{c_2}\!\!>_k<\!\!\rho\!\!>_{env}<\!\!m\!\!>_{mem}<\!\!X\!\!>_{bnd}<\!\!\varphi\wedge (\rho(b)\;\;is\;\;false)\!\!>_{form}\Downarrow<\!\!\cdot\!\!>_k ass_2$$
by M-CONS rule, we conclude
$$<\!\!\overline{c_2}\!\!>_k<\!\!\rho\!\!>_{env}<\!\!m\!\!>_{mem}<\!\!X\!\!>_{bnd}<\!\!\varphi\wedge (\rho(b)\;\;is\;\;false)\!\!>_{form}\Downarrow<\!\!\cdot\!\!>_k ass_1 \nabla ass_2$$
M-IF rule implies
$$<\!\!\overline{\textrm{if }(b)c_1 \textrm{ else }c_2}\!\!>_k<\!\!\rho\!\!>_{env}<\!\!m\!\!>_{mem}<\!\!X\!\!>_{bnd}<\!\!\varphi\!\!>_{form}\Downarrow<\!\!\cdot\!\!>_k ass_1 \nabla ass_2$$
Applying M-SEQ rule, we conclude
$$<\!\!\overline{(\textrm{if }(b)c_1 \textrm{ else }c_2);k'}\!\!>_k<\!\!\rho\!\!>_{env}<\!\!m\!\!>_{mem}<\!\!X\!\!>_{bnd}<\!\!\varphi\!\!>_{form}\Downarrow<\!\!\cdot\!\!>_k ass_{post}$$
is derivable from matching logic.\\
\textbf{Case7:} By V-WHILE rule
$$<\!\!<\!\!(\textrm{while }b\textrm{ do }c);k'; assert\; ass_{post}\!\!>_k<\!\!\rho\!\!>_{env}<\!\!m\!\!>_{mem}<\!\!X\!\!>_{bnd}\!<\!\!\varphi\!\!>_{form}\!\!>\rightarrow$$
$$<<\!\!c;\textrm{assert} (<\!\!\rho\!\!>_{e\!n\!v}<\!\!m\!\!>_{m\!e\!m}<\!\!X\!\!>_{b\!n\!d}\!<\!\!\varphi\!\!>_{f\!o\!r\!m})\!\!>_k<\!\!\rho\!\!>_{e\!n\!v}<\!\!m\!\!>_{m\!e\!m}<\!\!X\!\!>_{b\!n\!d}\!<\!\!\varphi\wedge (\rho(b)\;\;is\;\;true)\!\!>_{f\!o\!r\!m},$$
$$<\!\!k'; assert\; ass_{post}\!\!>_k<\!\!\rho\!\!>_{env}<\!\!m\!\!>_{mem}<\!\!X\!\!>_{bnd}\!<\!\!\varphi\wedge (\rho(b)\;\;is\;\;false)\!\!>_{form}>\to^*<\!\!\cdot\!\!>
$$
By the induction hypothesis of
$$<\!\!<\!\!c;\textrm{assert} (<\!\!\rho\!\!>_{e\!n\!v}<\!\!m\!\!>_{m\!e\!m}<\!\!X\!\!>_{b\!n\!d}\!<\!\!\varphi\!\!>_{f\!o\!r\!m})\!\!>_k<\!\!\rho\!\!>_{e\!n\!v}<\!\!m\!\!>_{m\!e\!m}<\!\!X\!\!>_{b\!n\!d}\!<\!\!\varphi\wedge (\rho(b)\;\;is\;\;true)
\!\!>_{f\!o\!r\!m}\!\!>$$$$\rightarrow^*<\!\!\cdot\!\!>
$$
we conclude
$$<\!\!\overline{c}\!\!>_k<\!\!\rho\!\!>_{env}<\!\!m\!\!>_{mem}<\!\!X\!\!>_{bnd}\!<\!\!\varphi\wedge (\rho(b)\;\;is\;\;true)\!\!>_{form}\Downarrow<\!\!\cdot\!\!>_k<\!\!\rho\!\!>_{env}<\!\!m\!\!>_{mem}<\!\!X\!\!>_{bnd}\!<\!\!\varphi\!\!>_{form}$$
M-WHILE rule implies
$$
<\!\!\overline{\textrm{while }b\textrm{ do }c}\!\!>_k<\!\!\rho\!\!>_{env}<\!\!m\!\!>_{mem}<\!\!X\!\!>_{bnd}\!<\!\!\varphi\!\!>_{form}\Downarrow$$
$$<\!\!\cdot\!\!>_k<\!\!\rho\!\!>_{env}<\!\!m\!\!>_{mem}<\!\!X\!\!>_{bnd}\!<\!\!\varphi\land(\rho(b)\;\;is\;\; false)\!\!>_{form}
$$
By the induction hypothesis of
$$<\!\!<\!\!k'; assert\; ass_{post}\!\!>_k<\!\!\rho\!\!>_{env}<\!\!m\!\!>_{mem}<\!\!X\!\!>_{bnd}\!<\!\!\varphi\wedge (\rho(b)\;\;is\;\;false)\!\!>_{form}\!\!>\rightarrow^*<\!\!\cdot\!\!>$$
we conclude
$$<\!\!\overline{k'}\!\!>_k<\!\!\rho\!\!>_{env}<\!\!m\!\!>_{mem}<\!\!X\!\!>_{bnd}\!<\!\!\varphi\wedge (\rho(b)\;\;is\;\;false)\!\!>_{form}\!\!>\Downarrow<\!\!\cdot\!\!>_k ass_{post}$$
By M-SEQ rule, we conclude
$$<\!\!\overline{\textrm{while }b\textrm{ do }c;k'}\!\!>_k<\!\!\rho\!\!>_{env}<\!\!m\!\!>_{mem}<\!\!X\!\!>_{bnd}\!<\!\!\varphi\!\!>_{form}\Downarrow<\!\!\cdot\!\!>_k ass_{post}$$
is derivable from matching logic.\\
\textbf{Case8:} By V-AWAIT rule
$$<\!\!<\!\!\textrm{await } b\textrm{ then } cc; k'; assert\; ass_{post}\!\!>_k<\!\!\rho\!\!>_{env}<\!\!m\!\!>_{mem}<\!\!X\!\!>_{bnd}\!<\!\!\varphi\!\!>_{form}\!\!>\rightarrow$$
$$<\!\!<\!\!cc;k'; assert\; ass_{post}\!\!>_k<\!\!\rho\!\!>_{env}<\!\!m\!\!>_{mem}<\!\!X\!\!>_{bnd}\!<\!\!\varphi\wedge(\rho(b)\;\;is\;\;true)\!\!>_{form}\!\!>\to^*<\!\!\cdot\!\!>$$
By the induction hypothesis of
$$<\!\!<\!\!cc;k'; assert\; ass_{post}\!\!>_k<\!\!\rho\!\!>_{env}<\!\!m\!\!>_{mem}<\!\!X\!\!>_{bnd}\!<\!\!\varphi\wedge(\rho(b)\;\;is\;\;true)\!\!>_{form}\!\!>\rightarrow^*<\!\!\cdot\!\!>$$
we conclude
$$<\!\!\overline{cc};\overline{k'}\!\!>_k<\!\!\rho\!\!>_{env}<\!\!m\!\!>_{mem}<\!\!X\!\!>_{bnd}\!<\!\!\varphi\wedge(\rho(b)\;\;is\;\;true)\!\!>_{form}\Downarrow<\!\!\cdot\!\!>_k ass_{post}$$
There exits some appropriate $ass_1$ such that
$$<\!\!\overline{cc}\!\!>_k<\!\!\rho\!\!>_{env}<\!\!m\!\!>_{mem}<\!\!X\!\!>_{bnd}\!<\!\!\varphi\wedge(\rho(b)\;\;is\;\;true)\!\!>_{form}\Downarrow<\!\!\cdot\!\!>_k ass_1$$
and
$$<\!\!\overline{k'}\!\!>_k ass_1\Downarrow<\!\!\cdot\!\!>_k ass_{post}$$
M-AWAIT rule implies
$$<\!\!\overline{\textrm{await }b\textrm{ then }\mathit{cc}}\!\!>_k<\!\!\rho\!\!>_{env}<\!\!m\!\!>_{mem}<\!\!X\!\!>_{bnd}\!<\!\!\varphi\!\!>_{form}\Downarrow<\!\!\cdot\!\!>_k ass_1$$
By M-SEQ rule, we conclude
$$<\!\!\overline{\textrm{await }b\textrm{ then }\mathit{cc};k'}\!\!>_k<\!\!\rho\!\!>_{env}<\!\!m\!\!>_{mem}<\!\!X\!\!>_{bnd}\!<\!\!\varphi\!\!>_{form}\Downarrow<\!\!\cdot\!\!>_k ass_{post}$$
is derivable from matching logic.\\
\textbf{Case9:} By V-PAR rule
$$
<\!\!<\!\!(c_1\parallel c_2);k'; assert\; ass_{post}\!\!>_k<\!\!\rho\!\!>_{env}<\!\!m\!\!>_{mem}<\!\!X\!\!>_{bnd}\!<\!\!\varphi_1\wedge\varphi_2\!\!>_{form}\!\!>\rightarrow$$
$$<<\!\!c_1; assert\;result(c_1,\rho,m,X,\varphi_1)\!\!>_k<\!\!\rho\!\!>_{env}<\!\!m\!\!>_{mem}<\!\!X\!\!>_{bnd}\!<\!\!\varphi_1\!\!>_{form},$$
$$<\!\!c_2;assert\;result(c_2,\rho,m,X,\varphi_2)\!\!>_k<\!\!\rho\!\!>_{env}<\!\!m\!\!>_{mem}<\!\!X\!\!>_{bnd}\!<\!\!\varphi_2\!\!>_{form},$$
$$<\!\!k'; assert\; ass_{post}\!\!>_k result(c_1,\rho,m,X,\varphi_1)\cap result(c_2,\rho,m,X,\varphi_2)>\rightarrow^*<\!\!\cdot\!\!>$$
By the induction hypothesis of
$$<\!\!<\!\!c_1\;assert\;result(c_1,\rho,m,X,\varphi_1)\!\!>_k<\!\!\rho\!\!>_{env}<\!\!m\!\!>_{mem}<\!\!X\!\!>_{bnd}\!<\!\!\varphi_1\!\!>_{form}\!\!>\rightarrow^*$$
$$<\!\!<\!\!\cdot\!\!>_k result(c_1,\rho,m,X,\varphi_1)\!\!>\rightarrow^*<\!\!\cdot\!\!>$$
we conclude
$$<\!\!\overline{c_1}\!\!>_k<\!\!\rho\!\!>_{env}<\!\!m\!\!>_{mem}<\!\!X\!\!>_{bnd}\!<\!\!\varphi_1\!\!>_{form}\!\!>\Downarrow<\!\!\cdot\!\!>_k result(c_1,\rho,m,X,\varphi_1)$$
where $result(c_1,\rho,m,X,\varphi_1)=<\!\!\rho'\!\!>_{env}<\!\!m'\!\!>_{mem}<\!\!X'\!\!>_{bnd}\!<\!\!\varphi_1'\!\!>_{form}$.\\
Similarly, we conclude
$$<\!\!\overline{c_2}\!\!>_k<\!\!\rho\!\!>_{env}<\!\!m\!\!>_{mem}<\!\!X\!\!>_{bnd}\!<\!\!\varphi_2\!\!>_{form}\!\!>\Downarrow<\!\!\cdot\!\!>_k result(c_2,\rho,m,X,\varphi_2)$$
where $result(c_2,\rho,m,X,\varphi_2)=<\!\!\rho'\!\!>_{env}<\!\!m'\!\!>_{mem}<\!\!X'\!\!>_{bnd}\!<\!\!\varphi_2'\!\!>_{form}$.\\
Since $\overline{c_1},\;\overline{c_2}$ are interference-free, M-PAR rule implies
$$<\!\!\overline{c_1\parallel c_2}\!\!>_k<\!\!\rho\!\!>_{env}<\!\!m\!\!>_{mem}<\!\!X\!\!>_{bnd}\!<\!\!\varphi_1\wedge\varphi_2\!\!>_{form}\!\!>\Downarrow$$
$$<\!\!\cdot\!\!>_k result(c_1,\rho,m,X,\varphi_1)\cap result(c_2,\rho,m,X,\varphi_2)$$
By the induction hypothesis of
$$<\!\!<\!\!k'; assert\; ass_{post}\!\!>_k result(c_1,\rho,m,X,\varphi_1)\cap result(c_2,\rho,m,X,\varphi_2)\!\!>\rightarrow^*<\!\!\cdot\!\!>$$
we conclude
$$<\!\!\overline{k'}\!\!>_k result(c_1,\rho,m,X,\varphi_1)\cap result(c_2,\rho,m,X,\varphi_2)\Downarrow <\!\!\cdot\!\!>_k ass_{post}$$
By M-SEQ rule, we conclude
$$<\!\!\overline{(c_1\parallel c_2);k'}\!\!>_k<\!\!\rho\!\!>_{env}<\!\!m\!\!>_{mem}<\!\!X\!\!>_{bnd}\!<\!\!\varphi_1\wedge\varphi_2\!\!>_{form}\!\!>\Downarrow<\!\!\cdot\!\!>_k ass_{post} $$
is derivable from matching logic.\\
\textbf{Case10:} By V-SKIP rule
$$<\!\!<\!\!skip;k'; assert\; ass_{post}\!\!>_k ass_1\!\!>\rightarrow<\!\!<\!\!k'; assert\; ass_{post}\!\!>_k ass_1\!\!>\to^*<\!\!\cdot\!\!>$$
By the induction hypothesis of
$$<\!\!<\!\!k'; assert\; ass_{post}\!\!>_k ass_1\!\!>\rightarrow^*<\!\!\cdot\!\!>$$
we conclude $<\!\!\overline{k'}\!\!>_k ass_1 \Downarrow <\!\!\cdot\!\!>_k ass_{post}$. Since $<\!\!skip\!\!>_k ass_1 \Downarrow <\!\!\cdot\!\!>_k ass_1$ and $\overline{skip;k'}=skip;\overline{k'}$, by M-SEQ rule, we conclude $<\!\!\overline{skip;k'}\!\!>_k ass_1 \Downarrow <\!\!\cdot\!\!>_k ass_{post}$ is derivable from matching logic.\\
We annotate the program with assertions which are the pre-condition and post-condition of the corresponding code segment in the matching logic, and then verify the annotated program. Let $\hbar$ be a mapping from matching logic rules to annotated programs.\\
If $\pi_{skip}\equiv <\!\!skip\!\!>_k ass\Downarrow <\!\!\cdot\!\!>_k ass$, then
$\hbar(\pi_{skip})= assert\;ass; skip; assert\;ass$;\\
If $\pi_{asgn1}\equiv <\!\!\mathrm{x}:=e\!\!>_k ass_{pre}\Downarrow <\!\!\cdot\!\!>_k ass_{post}$, then $\hbar(\pi_{asgn1})= assert\;ass_{pre}; \mathrm{x}:=e; assert\;ass_{post}$;\\
If $\pi_{asgn2}\equiv <\!\!\mathrm{x}:=A[e]\!\!>_k ass_{pre}\Downarrow <\!\!\cdot\!\!>_k ass_{post}$, then $\hbar(\pi_{asgn2})= assert\;ass_{pre}; \mathrm{x}:=A[e]; assert\;ass_{post}$;\\
If $\pi_{asgn3}\equiv <\!\!A[e_1]:=e_2\!\!>_k ass_{pre}\Downarrow <\!\!\cdot\!\!>_k ass_{post}$, then $\hbar(\pi_{asgn3})= assert\;ass_{pre}; A[e_1]:=e_2; assert\;ass_{post}$;\\
If $\pi_{array}\equiv <\!\!A:=Array(\overline{e})\!\!>_k ass_{pre}\Downarrow <\!\!\cdot\!\!>_k ass_{post}$, then $$\hbar(\pi_{array})= assert\;ass_{pre}; A:=Array(\overline{e}); assert\;ass_{post}$$
If $\pi_{seq}\equiv <\!\!c_1;c_2\!\!>_k ass_{pre}\Downarrow <\!\!\cdot\!\!>_k ass_{post}$, then $$\hbar(\pi_{seq})= assert\;ass_{pre};\hbar(\pi_{c_1});\hbar(\pi_{c_2}) ; assert\;ass_{post}$$
where $\pi_{c_1}\equiv <\!\!c_1\!\!>_k ass_{pre}\Downarrow <\!\!\cdot\!\!>_k ass_1$ and $\pi_{c_2}\equiv <\!\!c_2\!\!>_k ass_1\Downarrow <\!\!\cdot\!\!>_k ass_{post}$ for appropriate $ass_1$;\\
If $\pi_{if}\equiv <\!\!\textrm{if }(b)\;c_1\textrm{ else }c_2\!\!>_k<\!\!\rho\!\!>_{env}<\!\!m\!\!>_{mem}<\!\!X\!\!>_{bnd}<\!\!\varphi\!\!>_{form}\Downarrow <\!\!\cdot\!\!>_k ass_{post}$, then
$$\hbar(\pi_{if})= assert(<\!\!\rho\!\!>_{env}<\!\!m\!\!>_{mem}<\!\!X\!\!>_{bnd}<\!\!\varphi\!\!>_{form});(\textrm{if }(b)\;\hbar(\pi_{if_1}) \textrm{ else } \hbar(\pi_{if_2})) ; assert\;ass_{post}$$
where
$\pi_{if_1}\equiv <\!\!c_1\!\!>_k<\!\!\rho\!\!>_{env}<\!\!m\!\!>_{mem}<\!\!X\!\!>_{bnd}<\!\!\varphi\wedge (\rho(b)\;is\;true)\!\!>_{form}\Downarrow <\!\!\cdot\!\!>_k ass_{post}$ and
$\pi_{if_2}\equiv <\!\!c_2\!\!>_k<\!\!\rho\!\!>_{env}<\!\!m\!\!>_{mem}<\!\!X\!\!>_{bnd}<\!\!\varphi\wedge (\rho(b)\;is\;false)\!\!>_{form}\Downarrow <\!\!\cdot\!\!>_k ass_{post}$;\\
If $\pi_{while}\equiv <\!\!\textrm{while }b\textrm{ do  }c\!\!>_k<\!\!\rho\!\!>_{env}<\!\!m\!\!>_{mem}<\!\!X\!\!>_{bnd}<\!\!\varphi\!\!>_{form}\Downarrow <\!\!\cdot\!\!>_k ass_{post}$, then
$$\hbar(\pi_{while})= assert(<\!\!\rho\!\!>_{env}<\!\!m\!\!>_{mem}<\!\!X\!\!>_{bnd}<\!\!\varphi\!\!>_{form});(\textrm{while }b\textrm{ do }\hbar(\pi_{body})) ; assert\;ass_{post}$$
where $\pi_{body}\equiv <\!\!c\!\!>_k<\!\!\rho\!\!>_{env}<\!\!m\!\!>_{mem}<\!\!X\!\!>_{bnd}<\!\!\varphi\wedge (\rho(b)\;is\;true)\!\!>_{form}\Downarrow <\!\!\cdot\!\!>_k<\!\!\rho\!\!>_{env}<\!\!m\!\!>_{mem}<\!\!X\!\!>_{bnd}<\!\!\varphi\!\!>_{form}$;\\
If $\pi_{await}\equiv <\!\!\textrm{await }b \textrm{ then } cc\!\!>_k<\!\!\rho\!\!>_{env}<\!\!m\!\!>_{mem}<\!\!X\!\!>_{bnd}<\!\!\varphi\!\!>_{form}\Downarrow <\!\!\cdot\!\!>_k ass_{post}$, then
$$\hbar(\pi_{await})= assert(<\!\!\rho\!\!>_{env}<\!\!X\!\!>_{bnd}<\!\!\varphi\!\!>_{form});
(\textrm{await }b \textrm{ then }\;\hbar(\pi_{body})) ; assert\;ass_{post}$$
where $\pi_{body}\equiv <\!\!cc\!\!>_k<\!\!\rho\!\!>_{env}<\!\!m\!\!>_{mem}<\!\!X\!\!>_{bnd}<\!\!\varphi\wedge (\rho(b)\;is\;true)\!\!>_{form}\Downarrow <\!\!\cdot\!\!>_k ass_{post}$;\\
If $\pi_{par}\equiv <\!\!c_1\parallel c_2\!\!>_k<\!\!\rho\!\!>_{env}<\!\!m\!\!>_{mem}<\!\!X_1\!\!>_{bnd}<\!\!\varphi_1\wedge \varphi_2\!\!>_{form}\Downarrow <\!\!\cdot\!\!>_k ass_{post}$, then
$$\hbar(\pi_{par})= assert(<\!\!\rho\!\!>_{env}<\!\!m\!\!>_{mem}<\!\!X_1\!\!>_{bnd}<\!\!\varphi_1\wedge \varphi_2 \!\!>_{form});
(\hbar(\pi_{c_1})\parallel\hbar(\pi_{c_2})) ; assert\;ass_{post}$$
where $\pi_{c_1}\equiv <\!\!c_1\!\!>_k<\!\!\rho\!\!>_{env}<\!\!m\!\!>_{mem}<\!\!X_1\!\!>_{bnd}<\!\!\varphi_1\!\!>_{form}\Downarrow <\!\!\cdot\!\!>_k<\!\!\rho'\!\!>_{env}<\!\!m'\!\!>_{mem}<\!\!X_2\!\!>_{bnd}<\!\!\varphi_1'\!\!>_{form}$ and $\pi_{c_2}\equiv <\!\!c_2\!\!>_k<\!\!\rho\!\!>_{env}<\!\!m\!\!>_{mem}<\!\!X_1\!\!>_{bnd}<\!\!\varphi_2\!\!>_{form}\Downarrow <\!\!\cdot\!\!>_k<\!\!\rho'\!\!>_{env}<\!\!m'\!\!>_{mem}<\!\!X_2\!\!>_{bnd}<\!\!\varphi_2'\!\!>_{form}$.\\
For any $<\!\!c\!\!>_k ass_{pre}\Downarrow <\!\!\cdot\!\!>_k ass_{post}$, say $\pi_{c}$, $ \overline{\hbar(\pi_{c})}=c$. All we need to do is that
$$<\!\!<\!\!\hbar(\pi_{c})\!\!>_k ass_{pre}\!\!>\rightarrow^*<\!\!\cdot\!\!>$$
We prove by structural induction on $\pi_{c}$:\\
\textbf{Case1:} $\pi_{c}\equiv <\!\!\mathrm{x}:=e\!\!>_k ass_{pre}\Downarrow <\!\!\cdot\!\!>_k ass_{post}$ and $\hbar(\pi_{c})\equiv assert\;ass_{pre}; \mathrm{x}:=e; assert\;ass_{post}$ where $ass_{pre} \equiv <\!\!\rho\!\!>_{env}<\!\!m\!\!>_{mem}<\!\!X\!\!>_{bnd}<\!\!\varphi\!\!>_{form}$ and $ass_{post} \equiv <\!\!\rho[\rho(e)/x]\!\!>_{env}<\!\!m\!\!>_{mem}<\!\!X\!\!>_{bnd}<\!\!\varphi\!\!>_{form}$.\\
By V-ASSERT and V-ASGN1 rules,
$$<\!\!<\!\!assert\;ass_{pre}; \mathrm{x}:=e; assert\;ass_{post}\!\!>_k ass_{pre}\!\!>\rightarrow<\!\!<\!\!\mathrm{x}:=e; assert\;ass_{post}\!\!>_k ass_{pre}\!\!>$$
$$\rightarrow<\!\!<\!\!assert\;ass_{post}\!\!>_k ass_{post}\!\!>
$$
By V-ASSERT rule again, we conclude
$$<\!\!<\!\!assert\;ass_{post}\!\!>_k ass_{post}\!\!>\rightarrow<\!\!<\!\!\cdot\!\!>_k ass_{post}\!\!>$$
The V-NULL rule implies $<\!\!<\!\!\cdot\!\!>_k ass_{post}\!\!>\rightarrow<\!\!\cdot\!\!>$.\\
\textbf{Case2:} $\pi_{c}\equiv <\!\!\mathrm{x}:=A[e]\!\!>_k ass_{pre}\Downarrow <\!\!\cdot\!\!>_k ass_{post}$ and $\hbar(\pi_{c})\equiv assert\;ass_{pre}; \mathrm{x}:=A[e]; assert\;ass_{post}$ where $ass_{pre} \equiv <\!\!\rho\!\!>_{env}<\!\!m\!\!>_{mem}<\!\!X\!\!>_{bnd}<\!\!\varphi\!\!>_{form}$ and $ass_{post} \equiv <\!\!\rho[m(\rho(A)+_\mathit{Int}\rho(e))/\mathrm{x}]\!\!>_{env}<\!\!m\!\!>_{mem}<\!\!X\!\!>_{bnd}<\!\!\varphi\!\!>_{form}$.\\
By V-ASSERT and V-ASGN2 rules,
$$<\!\!<\!\!assert\;ass_{pre}; \mathrm{x}:=A[e]; assert\;ass_{post}\!\!>_k ass_{pre}\!\!>\rightarrow<\!\!<\!\!\mathrm{x}:=A[e]; assert\;ass_{post}\!\!>_k ass_{pre}\!\!>$$
$$\rightarrow<\!\!<\!\!assert\;ass_{post}\!\!>_k ass_{post}\!\!>$$
By V-ASSERT and V-NULL rules,
$$<\!\!<\!\!assert\;ass_{post}\!\!>_k ass_{post}\!\!>\rightarrow<\!\!<\!\!\cdot\!\!>_k ass_{post}\!\!>\rightarrow<\!\!\cdot\!\!>$$
\textbf{Case3:} $\pi_{c}\equiv <\!\!A[e_1]:=e_2\!\!>_k ass_{pre}\Downarrow <\!\!\cdot\!\!>_k ass_{post}$ and
$\hbar(\pi_{c})\equiv assert\;ass_{pre}; A[e_1]:=e_2; assert\;ass_{post}$ where $ass_{pre} \equiv <\!\!\rho\!\!>_{env}<\!\!m\!\!>_{mem}<\!\!X\!\!>_{bnd}<\!\!\varphi\!\!>_{form}$ and $ass_{post} \equiv <\!\!\rho\!\!>_{env}<\!\!m[\rho(e_2)/(\rho(A)+_\mathit{Int}\rho(e_1))]\!\!>_{mem}<\!\!X\!\!>_{bnd}<\!\!\varphi\!\!>_{form}$.\\
By V-ASSERT and V-ASGN3 rules,
$$
<\!\!<\!\!assert\;ass_{pre}; A[e_1]:=e_2; assert\;ass_{post}\!\!>_k ass_{pre}\!\!>\rightarrow$$
$$<\!\!<\!\!A[e_1]:=e_2; assert\;ass_{post}\!\!>_k ass_{pre}\!\!>\rightarrow<\!\!<\!\!assert\;ass_{post}\!\!>_k ass_{post}\!\!>
$$
By V-ASSERT and V-NULL rules,
$$<\!\!<\!\!assert\;ass_{post}\!\!>_k ass_{post}\!\!>\rightarrow<\!\!<\!\!\cdot\!\!>_k ass_{post}\!\!>\rightarrow<\!\!\cdot\!\!>$$
\textbf{Case4:} $\pi_{c}\equiv <\!\!A:=Array(\overline{e})\!\!>_k ass_{pre}\Downarrow <\!\!\cdot\!\!>_k ass_{post}$ and $\hbar(\pi_{c})\equiv assert\;ass_{pre}; A:=Array(\overline{e}); assert\;ass_{post}$ where
$ass_{pre} \equiv <\!\!\rho\!\!>_{env}<\!\!m\!\!>_{mem}<\!\!X\!\!>_{bnd}<\!\!\varphi\!\!>_{form}$ and $ass_{post} \equiv <\!\!\rho[p/A]\!\!>_{env}<\!\!p\mapsto[\rho(\overline{e})],m\!\!>_{mem}<\!\!X\cup \{p\}\!\!>_{bnd}<\!\!\varphi\!\!>_{form}$.\\
By V-ASSERT and V-ARRAY rules,
$$
<\!\!<\!\!assert\;ass_{pre}; A:=Array(\overline{e}); assert\;ass_{post}\!\!>_k ass_{pre}\!\!>\rightarrow$$
$$<\!\!<\!\!A:=Array(\overline{e}); assert\;ass_{post}\!\!>_k ass_{pre}\!\!>
\rightarrow<\!\!<\!\!assert\;ass_{post}\!\!>_k ass_{post}\!\!>
$$
By V-ASSERT and V-NULL rules,
$$<\!\!<\!\!assert\;ass_{post}\!\!>_k ass_{post}\!\!>\rightarrow<\!\!<\!\!\cdot\!\!>_k ass_{post}\!\!>\rightarrow<\!\!\cdot\!\!>$$
\textbf{Case5:} $\pi_{c}\equiv <\!\!c_1;c_2\!\!>_k ass_{pre}\Downarrow <\!\!\cdot\!\!>_k ass_{post}$ and $\hbar(\pi_{c})\equiv assert\;ass_{pre};\hbar(\pi_{c_1});\hbar(\pi_{c_2}) ; $\\ $assert\;ass_{post}$ where $\pi_{c_1}\equiv <\!\!c_1\!\!>_k ass_{pre}\Downarrow <\!\!\cdot\!\!>_k ass_1$ and $\pi_{c_2}\equiv <\!\!c_2\!\!>_k ass_1\Downarrow <\!\!\cdot\!\!>_k ass_{post}$ for appropriate $ass_1$.\\
By V-ASSERT rule
$$<\!\!<\!\!assert\;ass_{pre};\hbar(\pi_{c_1});\hbar(\pi_{c_2}) ; assert\;ass_{post}\!\!>_k ass_{pre}\!\!>\rightarrow$$
$$<\!\!<\!\!\hbar(\pi_{c_1});\hbar(\pi_{c_2}) ; assert\;ass_{post}\!\!>_k ass_{pre}\!\!>$$
By the induction hypothesis of $\pi_{c_1}$, we conclude
$$<\!\!<\!\!\hbar(\pi_{c_1})\!\!>_k ass_{pre}\!\!>\rightarrow^*<\!\!<\!\!\cdot\!\!>_k ass_1\!\!>$$
which implies
$$<\!\!<\!\!\hbar(\pi_{c_1});\hbar(\pi_{c_2}) ; assert\;ass_{post}\!\!>_k ass_{pre}\!\!>\rightarrow^*<\!\!<\!\!\hbar(\pi_{c_2}) ; assert\;ass_{post}\!\!>_k ass_1\!\!>$$
By the induction hypothesis of $\pi_{c_2}$, we conclude
$$<\!\!<\!\!\hbar(\pi_{c_2})\!\!>_k ass_1\!\!>\rightarrow^*<\!\!<\!\!\cdot\!\!>_k ass_{post}\!\!>$$
which implies
$$<\!\!<\!\!\hbar(\pi_{c_2}) ; assert\;ass_{post}\!\!>_k ass_1\!\!>\rightarrow^*<\!\!<\!\!assert\;ass_{post}\!\!>_k ass_{post}\!\!>$$
Applying V-ASSERT rule and V-NULL rule, we conclude
$$<\!\!<\!\!assert\;ass_{post}\!\!>_k ass_{post}\!\!>\rightarrow<\!\!<\!\!\cdot\!\!>_k ass_{post}\!\!>\rightarrow<\!\!\cdot\!\!>$$
\textbf{Case6:} $\pi_{c}\equiv <\!\!\textrm{if }(b)\;c_1\textrm{ else }c_2\!\!>_k<\!\!\rho\!\!>_{env}<\!\!m\!\!>_{mem}<\!\!X\!\!>_{bnd}<\!\!\varphi\!\!>_{form}\Downarrow <\!\!\cdot\!\!>_k ass_{post}$ and $\hbar(\pi_{c})\equiv assert(<\!\!\rho\!\!>_{env}<\!\!m\!\!>_{mem}<\!\!X\!\!>_{bnd}<\!\!\varphi\!\!>_{form});(\textrm{if }(b)\;\hbar(\pi_{if_1}) \textrm{ else } \hbar(\pi_{if_2})) ; assert\;ass_{post}$ where
$\pi_{if_1}\equiv <\!\!c_1\!\!>_k<\!\!\rho\!\!>_{env}<\!\!m\!\!>_{mem}<\!\!X\!\!>_{bnd}<\!\!\varphi\wedge (\rho(b)\;is\;true)\!\!>_{form}\Downarrow <\!\!\cdot\!\!>_k ass_{post}$ and
$\pi_{if_2}\equiv <\!\!c_2\!\!>_k<\!\!\rho\!\!>_{env}<\!\!m\!\!>_{mem}<\!\!X\!\!>_{bnd}<\!\!\varphi\wedge (\rho(b)\;is\;false)\!\!>_{form}\Downarrow <\!\!\cdot\!\!>_k ass_{post}$.\\
By V-ASSERT and V-IF rules,
$$
<\!\!<\!\!assert(<\!\!\rho\!\!>_{env}<\!\!m\!\!>_{mem}<\!\!X\!\!>_{bnd}<\!\!\varphi\!\!>_{form});(\textrm{if }(b)\;\hbar(\pi_{if_1}) \textrm{ else } \hbar(\pi_{if_2})) ; assert\;ass_{post}\!\!>_k $$
$$<\!\!\rho\!\!>_{env}<\!\!m\!\!>_{mem}<\!\!X\!\!>_{bnd}<\!\!\varphi\!\!>_{form}\!\!>\rightarrow$$$$<\!\!<\!\!(\textrm{if }(b)\;\hbar(\pi_{if_1}) \textrm{ else } \hbar(\pi_{if_2})) ; assert\;ass_{post}\!\!>_k <\!\!\rho\!\!>_{env}
<\!\!m\!\!>_{mem}<\!\!X\!\!>_{bnd}<\!\!\varphi\!\!>_{form}\!\!>\rightarrow$$
$$<<\!\!\hbar(\pi_{if_1});assert\;ass_{post}\!\!>_k<\!\!\rho\!\!>_{env}<\!\!m\!\!>_{mem}<\!\!X\!\!>_{bnd}<\!\!\varphi\wedge(\rho(b)\;is\;true)\!\!>_{form},$$
$$ <\!\!\hbar(\pi_{if_2});assert\;ass_{post}\!\!>_k<\!\!\rho\!\!>_{env}<\!\!m\!\!>_{mem}<\!\!X\!\!>_{bnd}<\!\!\varphi\wedge(\rho(b)\;is\;false)\!\!>_{form}>
$$
By the induction hypothesis of $\pi_{if_1}$, we conclude
$$<\!\!<\!\!\hbar(\pi_{if_1})\!\!>_k<\!\!\rho\!\!>_{env}<\!\!m\!\!>_{mem}<\!\!X\!\!>_{bnd}<\!\!\varphi\wedge(\rho(b)\;is\;true)\!\!>_{form}\!\!>\rightarrow^*<\!\!<\!\!\cdot\!\!>_k ass_{psot}\!\!>$$
which implies
$$<\!\!<\!\!\hbar(\pi_{if_1});assert\;ass_{post}\!\!>_k<\!\!\rho\!\!>_{env}<\!\!m\!\!>_{mem}<\!\!X\!\!>_{bnd}<\!\!\varphi\wedge(\rho(b)\;is\;true)\!\!>_{form}\!\!>$$
$$\rightarrow^*<\!\!<\!\!assert\;ass_{post}\!\!>_k ass_{psot}\!\!>
$$
By the induction hypothesis of $\pi_{if_2}$, we conclude
$$<\!\!<\!\!\hbar(\pi_{if_2})\!\!>_k<\!\!\rho\!\!>_{env}<\!\!m\!\!>_{mem}<\!\!X\!\!>_{bnd}<\!\!\varphi\wedge(\rho(b)\;is\;false)\!\!>_{form}\!\!>\rightarrow^*<\!\!<\!\!\cdot\!\!>_k ass_{psot}\!\!>$$
which implies
$$<\!\!<\!\!\hbar(\pi_{if_2});assert\;ass_{post}\!\!>_k<\!\!\rho\!\!>_{env}<\!\!m\!\!>_{mem}<\!\!X\!\!>_{bnd}<\!\!\varphi\wedge(\rho(b)\;is\;false)\!\!>_{form}\!\!>$$
$$\rightarrow^*<\!\!<\!\!assert\;ass_{post}\!\!>_k ass_{psot}\!\!>
$$
Therefor,
$$
<<\!\!\hbar(\pi_{if_1});assert\;ass_{post}\!\!>_k<\!\!\rho\!\!>_{env}<\!\!m\!\!>_{mem}<\!\!X\!\!>_{bnd}<\!\!\varphi\wedge(\rho(b)\;is\;true)\!\!>_{form},$$
$$<\!\!\hbar(\pi_{if_2});assert\;ass_{post}\!\!>_k<\!\!\rho\!\!>_{env}<\!\!m\!\!>_{mem}<\!\!X\!\!>_{bnd}<\!\!\varphi\wedge(\rho(b)\;is\;false)\!\!>_{form}>$$
$$\rightarrow^*<<\!\!assert\;ass_{post}\!\!>_k aas_{post},\;<\!\!assert\;ass_{post}\!\!>_k aas_{post}>
$$
Applying V-ASSERT rule and V-NULL rule, we conclude
$$<<\!\!assert\;ass_{post}\!\!>_k aas_{post},\;<\!\!assert\;ass_{post}\!\!>_k aas_{post}>\rightarrow<\!\!\cdot\!\!>$$
\textbf{Case7:} $\pi_{c}\equiv <\!\!\textrm{while }b\textrm{ do  }c\!\!>_k<\!\!\rho\!\!>_{env}<\!\!m\!\!>_{mem}<\!\!X\!\!>_{bnd}<\!\!\varphi\!\!>_{form}\Downarrow <\!\!\cdot\!\!>_k ass_{post}$ and $\hbar(\pi_{c})\equiv assert(<\!\!\rho\!\!>_{env}<\!\!m\!\!>_{mem}<\!\!X\!\!>_{bnd}<\!\!\varphi\!\!>_{form});(\textrm{while }b\textrm{ do }\hbar(\pi_{body})) ; assert\;ass_{post}$ where $\pi_{body}\equiv <\!\!c\!\!>_k<\!\!\rho\!\!>_{env}<\!\!m\!\!>_{mem}<\!\!X\!\!>_{bnd}<\!\!\varphi\wedge (\rho(b)\;is\;true)\!\!>_{form}\Downarrow <\!\!\cdot\!\!>_k<\!\!\rho\!\!>_{env}<\!\!m\!\!>_{mem}<\!\!X\!\!>_{bnd}<\!\!\varphi\!\!>_{form}$ and $aas_{post}\equiv <\!\!\rho\!\!>_{env}<\!\!m\!\!>_{mem}<\!\!X\!\!>_{bnd}<\!\!\varphi\wedge(\rho(b)\;is\;false)\!\!>_{form}$.\\
By V-ASSERT and V-WHILE rules,
$$
<\!\!<\!\!assert(<\!\!\rho\!\!>_{env}<\!\!m\!\!>_{mem}<\!\!X\!\!>_{bnd}<\!\!\varphi\!\!>_{form});(\textrm{while }b\textrm{ do }\hbar(\pi_{body})) ; assert\;ass_{post}\!\!>_k <\!\!\rho\!\!>_{env}$$
$$<\!\!m\!\!>_{mem}<\!\!X\!\!>_{bnd}<\!\!\varphi\!\!>_{form}\!\!>\rightarrow$$
$$<\!\!<\!\!(\textrm{while }b\textrm{ do }\hbar(\pi_{body})) ; assert\;ass_{post}\!\!>_k <\!\!\rho\!\!>_{env}<\!\!m\!\!>_{mem}<\!\!X\!\!>_{bnd}<\!\!\varphi\!\!>_{form}\!\!>\rightarrow$$
$$<<\!\!\hbar(\pi_{body});assert(<\!\!\rho\!\!>_{env}<\!\!m\!\!>_{mem}<\!\!X\!\!>_{bnd}<\!\!\varphi\!\!>_{form})\!\!>_k <\!\!\rho\!\!>_{env}<\!\!m\!\!>_{mem}<\!\!X\!\!>_{bnd}$$
$$<\!\!\varphi\wedge(\rho(b)\;is\;true)\!\!>_{form},<\!\!assert\;ass_{post}\!\!>_k <\!\!\rho\!\!>_{env}<\!\!m\!\!>_{mem}<\!\!X\!\!>_{bnd}<\!\!\varphi\wedge(\rho(b)\;is\;false)\!\!>_{form}>
$$
By the induction hypothesis of $\pi_{body}$,  we conclude
$$<\!\!<\!\!\hbar(\pi_{body})\!\!>_k <\!\!\rho\!\!>_{env}<\!\!m\!\!>_{mem}<\!\!X\!\!>_{bnd}<\!\!\varphi\wedge(\rho(b)\;is\;true)\!\!>_{form}\!\!>$$
$$\rightarrow^*<\!\!<\!\!\cdot\!\!>_k <\!\!\rho\!\!>_{env}<\!\!m\!\!>_{mem}<\!\!X\!\!>_{bnd}<\!\!\varphi\!\!>_{form}\!\!>$$
which implies
$$<\!\!<\!\!\hbar(\pi_{body});assert(<\!\!\rho\!\!>_{e\!n\!v}<\!\!m\!\!>_{m\!e\!m}<\!\!X\!\!>_{b\!n\!d}$$
$$<\!\!\varphi\!\!>_{f\!o\!r\!m})\!\!>_k <\!\!\rho\!\!>_{e\!n\!v}<\!\!m\!\!>_{m\!e\!m}<\!\!X\!\!>_{b\!n\!d}<\!\!\varphi\wedge(\rho(b)\;is\;true)\!\!>_{f\!o\!r\!m}\!\!>\rightarrow^*$$
$$<\!\!<\!\!assert(<\!\!\rho\!\!>_{env}<\!\!m\!\!>_{mem}<\!\!X\!\!>_{bnd}<\!\!\varphi\!\!>_{form})\!\!>_k <\!\!\rho\!\!>_{env}<\!\!m\!\!>_{mem}<\!\!X\!\!>_{bnd}<\!\!\varphi\!\!>_{form}\!\!>$$
By V-ASSERT and V-NULL rules
$$
<\!\!<\!\!assert(<\!\!\rho\!\!>_{env}<\!\!m\!\!>_{mem}<\!\!X\!\!>_{bnd}<\!\!\varphi\!\!>_{form})\!\!>_k <\!\!\rho\!\!>_{env}<\!\!m\!\!>_{mem}<\!\!X\!\!>_{bnd}<\!\!\varphi\!\!>_{form},$$
$$<\!\!assert\;ass_{post}\!\!>_k <\!\!\rho\!\!>_{env}<\!\!m\!\!>_{mem}<\!\!X\!\!>_{bnd}<\!\!\varphi\wedge(\rho(b)\;is\;false)\!\!>_{form}>\rightarrow^*<\!\!\cdot\!\!>
$$
\textbf{Case8:} $\pi_{c}\equiv <\!\!\textrm{await }b \textrm{ then } cc\!\!>_k<\!\!\rho\!\!>_{env}<\!\!m\!\!>_{mem}<\!\!X\!\!>_{bnd}<\!\!\varphi\!\!>_{form}\Downarrow <\!\!\cdot\!\!>_k ass_{post}$ and $\hbar(\pi_{c})\equiv assert(<\!\!\rho\!\!>_{env}<\!\!X\!\!>_{bnd}<\!\!\varphi\!\!>_{form});
(\textrm{await }b \textrm{ then }\;\hbar(\pi_{body})) ; assert\;ass_{post}$ where $\pi_{body}\equiv <\!\!cc\!\!>_k<\!\!\rho\!\!>_{env}<\!\!m\!\!>_{mem}<\!\!X\!\!>_{bnd}<\!\!\varphi\wedge (\rho(b)\;is\;true)\!\!>_{form}\Downarrow <\!\!\cdot\!\!>_k ass_{post}$.\\
By V-ASSERT and V-AWAIT rules,
$$
<\!\!<\!\!assert(<\!\!\rho\!\!>_{env}<\!\!X\!\!>_{bnd}<\!\!\varphi\!\!>_{form});
(\textrm{await }b \textrm{ then }\;\hbar(\pi_{body})) ; assert\;ass_{post}\!\!>_k <\!\!\rho\!\!>_{env}$$
$$<\!\!m\!\!>_{mem}<\!\!X\!\!>_{bnd}<\!\!\varphi\!\!>_{form}\!\!>\rightarrow^*$$
$$<\!\!<\!\!\hbar(\pi_{body}); assert\;ass_{post}\!\!>_k <\!\!\rho\!\!>_{env}<\!\!m\!\!>_{mem}<\!\!X\!\!>_{bnd}<\!\!\varphi\wedge(\rho(b)\;is\;true)\!\!>_{f\!o\!r\!m}\!\!>
$$
By the induction hypothesis of $\pi_{body}$, we conclude
$$<\!\!<\!\!\hbar(\pi_{body})\!\!>_k <\!\!\rho\!\!>_{env}<\!\!m\!\!>_{mem}<\!\!X\!\!>_{bnd}<\!\!\varphi\wedge(\rho(b)\;is\;true)\!\!>_{form}\!\!>\rightarrow^*<\!\!<\!\!\cdot\!\!>_k aas_{psot}\!\!>$$
Therefor,
$$<\!\!<\!\!\hbar(\pi_{body}); assert\;ass_{post}\!\!>_k <\!\!\rho\!\!>_{env}<\!\!m\!\!>_{mem}<\!\!X\!\!>_{bnd}<\!\!\varphi\wedge(\rho(b)\;is\;true)\!\!>_{form}\!\!>\rightarrow^*$$
$$<\!\!<\!\!assert\;ass_{post}\!\!>_k aas_{psot}\!\!>
$$
By V-ASSERT and V-NULL rules
$$<\!\!<\!\!assert\;ass_{post}\!\!>_k aas_{psot}\!\!>\rightarrow<\!\!<\!\!\cdot\!\!>_k aas_{psot}\!\!>\rightarrow<\!\!\cdot\!\!>$$
\textbf{Case9:} $\pi_{c}\equiv <\!\!c_1\parallel c_2\!\!>_k<\!\!\rho\!\!>_{env}<\!\!m\!\!>_{mem}<\!\!X_1\!\!>_{bnd}<\!\!\varphi_1\wedge \varphi_2\!\!>_{form}\Downarrow <\!\!\cdot\!\!>_k ass_{post}$ and $\hbar(\pi_{c})\equiv assert(<\!\!\rho\!\!>_{env}<\!\!m\!\!>_{mem}<\!\!X_1\!\!>_{bnd}<\!\!\varphi_1\wedge \varphi_2 \!\!>_{form});
(\hbar(\pi_{c_1})\parallel\hbar(\pi_{c_2})) ; assert\;ass_{post}$ where $\pi_{c_1}\equiv <\!\!c_1\!\!>_k<\!\!\rho\!\!>_{env}<\!\!m\!\!>_{mem}<\!\!X_1\!\!>_{bnd}<\!\!\varphi_1\!\!>_{form}\Downarrow <\!\!\cdot\!\!>_k<\!\!\rho'\!\!>_{env}<\!\!m'\!\!>_{mem}<\!\!X_2\!\!>_{bnd}<\!\!\varphi_1'\!\!>_{form}$ and $\pi_{c_2}\equiv <\!\!c_2\!\!>_k<\!\!\rho\!\!>_{env}<\!\!m\!\!>_{mem}<\!\!X_1\!\!>_{bnd}<\!\!\varphi_2\!\!>_{form}\Downarrow <\!\!\cdot\!\!>_k<\!\!\rho'\!\!>_{env}<\!\!m'\!\!>_{mem}<\!\!X_2\!\!>_{bnd}<\!\!\varphi_2'\!\!>_{form}$ and $aas_{post}\equiv <\!\!\rho'\!\!>_{env}<\!\!m'\!\!>_{mem}<\!\!X_2\!\!>_{bnd}<\!\!\varphi_1'\wedge\varphi_2'\!\!>_{form}$\\
By V-ASSERT and V-PAR rules,
$$
<\!\!<\!\!assert(<\!\!\rho\!\!>_{env}<\!\!m\!\!>_{mem}<\!\!X_1\!\!>_{bnd}<\!\!\varphi_1\wedge \varphi_2 \!\!>_{form});
(\hbar(\pi_{c_1})\parallel\hbar(\pi_{c_2})) ; assert\;ass_{post}\!\!>_k <\!\!\rho\!\!>_{env}$$
$$<\!\!m\!\!>_{mem}<\!\!X_1\!\!>_{bnd}<\!\!\varphi_1\wedge \varphi_2 \!\!>_{form}\!\!>\rightarrow<\!\!<\!\!(\hbar(\pi_{c_1})\parallel\hbar(\pi_{c_2})) ; assert\;ass_{post}\!\!>_k <\!\!\rho\!\!>_{env}<\!\!m\!\!>_{mem}$$
$$<\!\!X_1\!\!>_{bnd}<\!\!\varphi_1\wedge \varphi_2 \!\!>_{form}\!\!>\rightarrow<<\!\!\hbar(\pi_{c_1})\!\!>_k <\!\!\rho\!\!>_{env}<\!\!m\!\!>_{mem}<\!\!X_1\!\!>_{bnd}<\!\!\varphi_1\!\!>_{form},$$
$$<\!\!\hbar(\pi_{c_2})\!\!>_k <\!\!\rho\!\!>_{env}<\!\!m\!\!>_{mem}<\!\!X_1\!\!>_{bnd}<\!\!\varphi_2\!\!>_{form},$$
$$<\!\! assert\;ass_{post}\!\!>_k result(\hbar(\pi_{c_1}),\rho,m,X_1,\varphi_1)\cap result(\hbar(\pi_{c_2}),\rho,m,X_1,\varphi_2)>$$
By the induction hypothesis of $\pi_{c_1}$, we conclude
$$<\!\!<\!\!\hbar(\pi_{c_1})\!\!>_k <\!\!\rho\!\!>_{env}<\!\!m\!\!>_{mem}<\!\!X_1\!\!>_{bnd}<\!\!\varphi_1\!\!>_{form}\!\!>\rightarrow^*$$
$$<\!\!<\!\!\cdot\!\!>_k<\!\!\rho'\!\!>_{env}<\!\!m'\!\!>_{mem}<\!\!X_2\!\!>_{bnd}<\!\!\varphi_1'\!\!>_{form}\!\!>$$
which implies
$$result(\hbar(\pi_{c_1}),\rho,m,X_1,\varphi_1)=<\!\!\rho'\!\!>_{env}<\!\!m'\!\!>_{mem}<\!\!X_2\!\!>_{bnd}<\!\!\varphi_1'\!\!>_{form}\!\!>
$$
By the induction hypothesis of $\pi_{c_2}$, we conclude
$$<\!\!<\!\!\hbar(\pi_{c_2})\!\!>_k <\!\!\rho\!\!>_{env}<\!\!m\!\!>_{mem}<\!\!X_1\!\!>_{bnd}<\!\!\varphi_2\!\!>_{form}\!\!>\rightarrow^*$$
$$<\!\!<\!\!\cdot\!\!>_k<\!\!\rho'\!\!>_{env}<\!\!m'\!\!>_{mem}<\!\!X_2\!\!>_{bnd}<\!\!\varphi_2'\!\!>_{form}\!\!>$$
which implies
$$result(\hbar(\pi_{c_2}),\rho,m,X_1,\varphi_2)=<\!\!\rho'\!\!>_{env}<\!\!m'\!\!>_{mem}<\!\!X_2\!\!>_{bnd}<\!\!\varphi_2'\!\!>_{form}\!\!>
$$
Therefor,
$$<\!\!<\!\!\hbar(\pi_{c_1})\!\!>_k <\!\!\rho\!\!>_{env}<\!\!m\!\!>_{mem}<\!\!X_1\!\!>_{bnd}<\!\!\varphi_1\!\!>_{form},$$ $$<\!\!\hbar(\pi_{c_2})\!\!>_k <\!\!\rho\!\!>_{env}<\!\!m\!\!>_{mem}<\!\!X_1\!\!>_{bnd}<\!\!\varphi_2\!\!>_{form}>,$$
$$<\!\! assert\;ass_{post}\!\!>_k result(\hbar(\pi_{c_1}),\rho,m,X_1,\varphi_1)\cap result(\hbar(\pi_{c_2}),\rho,m,X_1,\varphi_2)>\rightarrow^*$$
$$<<\!\!\cdot\!\!>,<\!\!\cdot\!\!>,<\!\!assert\;ass_{post}\!\!>_k<\!\!\rho'\!\!>_{env}<\!\!m'\!\!>_{mem}<\!\!X_2\!\!>_{bnd}<\!\!\varphi_1'\wedge\varphi_2'\!\!>_{form}>
$$
By V-ASSERT and V-NULL rules, we conclude
$$
<<\!\!\cdot\!\!>,<\!\!\cdot\!\!>,<\!\!assert\;ass_{post}\!\!>_k<\!\!\rho'\!\!>_{env}<\!\!m'\!\!>_{mem}<\!\!X_2\!\!>_{bnd}<\!\!\varphi_1'\wedge\varphi_2'\!\!>_{form}>\rightarrow^*<\!\!\cdot\!\!>
$$
\textbf{Case10:} $\pi_{c}\equiv <\!\!skip\!\!>_k ass\Downarrow <\!\!\cdot\!\!>_k ass$ and
$\hbar(\pi_{c})\equiv assert\;ass; skip; assert\;ass$.\\
By V-ASSERT ,V-SKIP and V-NULL rules ,
$$
<\!\!<\!\!assert\;ass; skip; assert\;ass\!\!>_k ass\!\!>\rightarrow<\!\!<\!\!skip; assert\;ass\!\!>_k ass\!\!>\rightarrow<\!\!<\!\! assert\;ass\!\!>_k ass\!\!>$$
$$\rightarrow<\!\!<\!\!\cdot\!\!>_k ass\!\!>\rightarrow<\!\!\cdot\!\!>$$
\end{proof}
\section{Application Example}
Let's take a standard problem in parallel programming as an example. The producer process provides  products (e.g. values) to the consumer process. Because the speed of producer process is different from that of consumer, it is profitable to set up a buffer between producer and consumer. But the storage is limited, the buffer can only store $N$ products. Figure 5 uses this solution to copy values of array $A[1:M]$ into array $B[1:M]$. Figure 6 gives proof outlines of the producer process and Figure 7 gives proof outlines of the consumer process. The yellow text is the code, the rest are assertions. $<\!\!\rho\!\!>_{env}<\!\!m\!\!>_{mem}<\!\!X\!\!>_{bnd}<\!\!\varphi\!\!>_{form}\;c\;<\!\!\rho'\!\!>_{env}<\!\!m'\!\!>_{mem}<\!\!X'\!\!>_{bnd}<\!\!\varphi'\!\!>_{form}$
implies the existence of a proof of the corresponding correctness pair, using the rules in Figure 3. Two consecutive assertions  denote a use of the rule of M-CONS.\\
\begin{figure}
  \small
  \begin{tabular}{l}
    \hline
    $\textbf{Comment:}$\\
    $C \equiv\textrm{the }\textrm{shared } \textrm{buffer }\textrm{of }\textrm{size }N$;\\
    $\mathrm{in} \equiv \textrm{number } \textrm{of } \textrm{elements } \textrm{added } \textrm{to } \textrm{the } \textrm{buffer}$;\\
    $\mathrm{out}\equiv  \textrm{number } \textrm{of } \textrm{elements } \textrm{removed } \textrm{from } \textrm{the } \textrm{buffer}$;\\
    $\mathrm{in}-\mathrm{out} \equiv \textrm{the } \textrm{total } \textrm{number } \textrm{of } \textrm{elements } \textrm{in } \textrm{the } \textrm{buffer}, \textrm{those } \textrm{are } \textrm{in } \textrm{order}:\mathit{C}[\mathrm{out}\;mod\;N],\cdots,$\\
    $\mathit{C}[\mathrm{(out+in-out-1)}\;mod\;N]$\\\\
    $\textbf{Begin:}$\\
     $\mathrm{in}:=0;\;\; \mathrm{out}:=0;\;\; \mathrm{i}:=1;\;\; \mathrm{j}:=1;$\\
     $\textbf{producer:}$\\
     $\qquad\qquad\quad\textbf{while }\mathrm{i}<M+1\;\; \textbf{do }$\\
     $\qquad\qquad\qquad\qquad \mathrm{x}:=A[\mathrm{i}];$\\
     $\qquad\qquad\qquad\qquad \textbf{await } \mathrm{in}-\mathrm{out}<N \textbf{ then } skip;$\\
     $\qquad\qquad\qquad\qquad C[\mathrm{in}\; mod\; N]:=\mathrm{x};$\\
     $\qquad\qquad\qquad\qquad \mathrm{in}:=\mathrm{in}+1;$\\
     $\qquad\qquad\qquad\qquad \mathrm{i}:=\mathrm{i}+1;$\\
     $\parallel$\\
     $\textbf{consumer:}$\\
     $\qquad\qquad\quad\textbf{while }\mathrm{j}<M+1\;\; \textbf{do }$\\
     $\qquad\qquad\qquad\qquad \textbf{await } \mathrm{in}-\mathrm{out}>0 \textbf{ then } skip;$\\
     $\qquad\qquad\qquad\qquad \mathrm{y}:=C[\mathrm{out}\; mod\; N];$\\
     $\qquad\qquad\qquad\qquad \mathrm{out}:=\mathrm{out}+1;$\\
     $\qquad\qquad\qquad\qquad B[\mathrm{j}]:=\mathrm{y};$\\
     $\qquad\qquad\qquad\qquad \mathrm{j}:=\mathrm{j}+1;$\\
     $\textbf{End}$\\
\hline
\end{tabular}
  \caption{Producer and Consumer}\label{fig2}
\end{figure}
\begin{figure}
  \small
  \begin{tabular}{l}
    \hline\\
     $1.<\!\!\mathrm{in}\mapsto in, \mathrm{out}\mapsto out, \mathrm{i}\mapsto i, \mathrm{j}\mapsto j, \mathrm{x}\mapsto x, \mathrm{y}\mapsto y, \mathrm{A}\mapsto a, \mathrm{B}\mapsto b, \mathrm{C}\mapsto c, \rho\!\!>_{env}<\!\!a\mapsto[\overline{p}], b\mapsto[\overline{q}],$\\
     $c\mapsto[\overline{k}], m\!\!>_{mem}<\!\!in,out,i,j,x,y,\overline{q},\overline{k}\!\!>_{bnd}<\!\!\varphi\!\!>_{form}$\\\\
     \colorbox{yellow}{$\mathrm{in}:=0;\;\; \mathrm{out}:=0;\;\; \mathrm{i}:=1;\;\; \mathrm{j}:=1;$}\\\\
     $2.<\!\!\mathrm{in}\mapsto in, \mathrm{out}\mapsto out, \mathrm{i}\mapsto i, \mathrm{j}\mapsto j, \mathrm{x}\mapsto x, \mathrm{y}\mapsto y, \mathrm{A}\mapsto a, \mathrm{B}\mapsto b, \mathrm{C}\mapsto c,\rho \!\!>_{env}<\!\!a\mapsto[\overline{p}], b\mapsto[\overline{q}],$\\
     $c\mapsto[\overline{k}], m\!\!>_{mem}<\!\!in,out,i,j,x,y,\overline{q},\overline{k}\!\!>_{bnd}<\!\!\varphi\wedge i=in+1\!\!>_{form}$\\\\
     \colorbox{yellow}{$\textbf{producer:}$}\\
     \colorbox{yellow}{$\qquad\qquad\quad\textbf{while }\mathrm{i}<M+1\;\; \textbf{do }$}\\\\
     $3.<\!\!\mathrm{in}\mapsto in, \mathrm{out}\mapsto out, \mathrm{i}\mapsto i, \mathrm{j}\mapsto j, \mathrm{x}\mapsto x, \mathrm{y}\mapsto y, \mathrm{A}\mapsto a, \mathrm{B}\mapsto b, \mathrm{C}\mapsto c,\rho \!\!>_{env}<\!\!a\mapsto[\overline{p}], b\mapsto[\overline{q}], $\\
     $c\mapsto[\overline{k}], m\!\!>_{mem}<\!\!in,out,i,j,x,y,\overline{q},\overline{k}\!\!>_{bnd}<\!\!\varphi\wedge i=in+1\wedge i<M+1\!\!>_{form}$\\\\
     \colorbox{yellow}{$\qquad\qquad\qquad\qquad \mathrm{x}:=A[\mathrm{i}];$}\\\\
$4.<\!\!\mathrm{in}\mapsto in, \mathrm{out}\mapsto out, \mathrm{i}\mapsto i, \mathrm{j}\mapsto j, \mathrm{x}\mapsto p_i, \mathrm{y}\mapsto y, \mathrm{A}\mapsto a, \mathrm{B}\mapsto b, \mathrm{C}\mapsto c,\rho \!\!>_{env}<\!\!a\mapsto[\overline{p}], b\mapsto[\overline{q}],$\\
     $c\mapsto[\overline{k}], m\!\!>_{mem}<\!\!in,out,i,j,x,y,\overline{q},\overline{k}\!\!>_{bnd}<\!\!\varphi\wedge i=in+1\wedge i<M+1\!\!>_{form}$\\\\
     $5.<\!\!\mathrm{in}\mapsto in, \mathrm{out}\mapsto out, \mathrm{i}\mapsto i, \mathrm{j}\mapsto j, \mathrm{x}\mapsto x, \mathrm{y}\mapsto y, \mathrm{A}\mapsto a, \mathrm{B}\mapsto b, \mathrm{C}\mapsto c,\rho \!\!>_{env}<\!\!a\mapsto[\overline{p}], b\mapsto[\overline{q}],$\\
     $c\mapsto[\overline{k}], m\!\!>_{mem}<\!\!in,out,i,j,x,y,\overline{q},\overline{k}\!\!>_{bnd}<\!\!\varphi\wedge i=in+1\wedge i<M+1\wedge x=p_i\!\!>_{form}$\\\\
     \colorbox{yellow}{$\qquad\qquad\qquad\qquad \textbf{await } \mathrm{in}-\mathrm{out}<N \textbf{ then } skip;$}\\\\
 $6.<\!\!\mathrm{in}\mapsto in, \mathrm{out}\mapsto out, \mathrm{i}\mapsto i, \mathrm{j}\mapsto j, \mathrm{x}\mapsto x, \mathrm{y}\mapsto y, \mathrm{A}\mapsto a, \mathrm{B}\mapsto b, \mathrm{C}\mapsto c,\rho \!\!>_{env}<\!\!a\mapsto[\overline{p}], b\mapsto[\overline{q}],$\\
 $c\mapsto[\overline{k}], m\!\!>_{mem}<\!\!in,out,i,j,x,y,\overline{q},\overline{k}\!\!>_{bnd}<\!\!\varphi\wedge i=in+1\wedge i<M+1\wedge x=p_i\wedge in-out<N$\\
 $>_{form}$\\\\
     \colorbox{yellow}{$\qquad\qquad\qquad\qquad C[\mathrm{in}\; mod\; N]:=x;$}\\\\
$7.<\!\!\mathrm{in}\mapsto in, \mathrm{out}\mapsto out, \mathrm{i}\mapsto i, \mathrm{j}\mapsto j, \mathrm{x}\mapsto x, \mathrm{y}\mapsto y, \mathrm{A}\mapsto a, \mathrm{B}\mapsto b, \mathrm{C}\mapsto c,\rho \!\!>_{env}<\!\!a\mapsto[\overline{p}], b\mapsto[\overline{q}],$\\
$c\mapsto[\overline{k}], m\!\!>_{mem}<\!\!in,out,i,j,x,y,\overline{q},\overline{k}\!\!>_{bnd}<\!\!\varphi\wedge i=in+1\wedge i<M+1\wedge in-out<N\wedge$\\
$k_{in\;mod\;N}=p_i\!\!>_{form}$\\\\
     \colorbox{yellow}{$\qquad\qquad\qquad\qquad \mathrm{in}:=\mathrm{in}+1;$}\\\\
     $8.<\!\!\mathrm{in}\mapsto in, \mathrm{out}\mapsto out, \mathrm{i}\mapsto i, \mathrm{j}\mapsto j, \mathrm{x}\mapsto x, \mathrm{y}\mapsto y, \mathrm{A}\mapsto a, \mathrm{B}\mapsto b, \mathrm{C}\mapsto c,\rho \!\!>_{env}<\!\!a\mapsto[\overline{p}], b\mapsto[\overline{q}], $\\
     $c\mapsto[\overline{k}], m\!\!>_{mem}<\!\!in,out,i,j,x,y,\overline{q},\overline{k}\!\!>_{bnd}<\!\!\varphi\wedge i=in\wedge i<M+1\!\!>_{form}$\\\\
     \colorbox{yellow}{$\qquad\qquad\qquad\qquad \mathrm{i}:=\mathrm{i}+1;$}\\\\
     $9.<\!\!\mathrm{in}\mapsto in, \mathrm{out}\mapsto out, \mathrm{i}\mapsto i, \mathrm{j}\mapsto j, \mathrm{x}\mapsto x, \mathrm{y}\mapsto y, \mathrm{A}\mapsto a, \mathrm{B}\mapsto b, \mathrm{C}\mapsto c,\rho \!\!>_{env}<\!\!a\mapsto[\overline{p}], b\mapsto[\overline{q}],$\\
     $c\mapsto[\overline{k}], m\!\!>_{mem}<\!\!in,out,i,j,x,y,\overline{q},\overline{k}\!\!>_{bnd}<\!\!\varphi\wedge i=in+1\wedge i<M+1\!\!>_{form}$\\\\
      $10.<\!\!\mathrm{in}\mapsto in, \mathrm{out}\mapsto out, \mathrm{i}\mapsto i, \mathrm{j}\mapsto j, \mathrm{x}\mapsto x, \mathrm{y}\mapsto y, \mathrm{A}\mapsto a, \mathrm{B}\mapsto b, \mathrm{C}\mapsto c,\rho \!\!>_{env}<\!\!a\mapsto[\overline{p}], b\mapsto[\overline{q}],$\\
     $c\mapsto[\overline{k}], m\!\!>_{mem}<\!\!in,out,i,j,x,y,\overline{q},\overline{k}\!\!>_{bnd}<\!\!\varphi\wedge i=in+1\wedge i=M+1\!\!>_{form}$\\\\
     $\textrm{where } \varphi=(k_{(n-1)\;mod\;N}=p_n,\; out<n< in+1)\wedge (0\leq in-out\leq N)\wedge(1\leq i \leq M+1)\wedge$\\ $(1\leq j\leq M+1)$ and
     $\overline{p}=p_0,p_1,p_2,\cdots,p_M$ and $a \mapsto[\overline{p}]=a \mapsto p_0, a+1 \mapsto p_1,\cdots, a+M \mapsto p_M$ \\ and $p_n\in Int, 0<n<M+1$\\\\
\hline
\end{tabular}
  \caption{Proof outlines of Producer}\label{fig2}
\end{figure}
\begin{figure}
  \small
  \begin{tabular}{l}
    \hline\\
     $1.<\!\!\mathrm{in}\mapsto in, \mathrm{out}\mapsto out, \mathrm{i}\mapsto i, \mathrm{j}\mapsto j, \mathrm{x}\mapsto x, \mathrm{y}\mapsto y, \mathrm{A}\mapsto a, \mathrm{B}\mapsto b, \mathrm{C}\mapsto c, \rho\!\!>_{env}<\!\!a\mapsto[\overline{p}], b\mapsto[\overline{q}],$\\
     $c\mapsto[\overline{k}], m\!\!>_{mem}<\!\!in,out,i,j,x,y,\overline{q},\overline{k}\!\!>_{bnd}<\!\!\varphi\!\!>_{form}$\\\\
     \colorbox{yellow}{$\mathrm{in}:=0;\;\; \mathrm{out}:=0;\;\; \mathrm{i}:=1;\;\; \mathrm{j}:=1;$}\\\\
     $2.<\!\!\mathrm{in}\mapsto in, \mathrm{out}\mapsto out, \mathrm{i}\mapsto i, \mathrm{j}\mapsto j, \mathrm{x}\mapsto x, \mathrm{y}\mapsto y, \mathrm{A}\mapsto a, \mathrm{B}\mapsto b, \mathrm{C}\mapsto c,\rho \!\!>_{env}<\!\!a\mapsto[\overline{p}], b\mapsto[\overline{q}], $\\
     $c\mapsto[\overline{k}] , m\!\!>_{mem}<\!\!in,out,i,j,x,y,\overline{q},\overline{k}\!\!>_{bnd}<\!\!\varphi\wedge\phi\wedge j=out+1\!\!>_{form}$\\\\
     \colorbox{yellow}{$\textbf{consumer:}$}\\\\
     \colorbox{yellow}{$\qquad\qquad\quad\textbf{while }\mathrm{j}<M+1\;\; \textbf{do }$}\\\\
     $3.<\!\!\mathrm{in}\mapsto in, \mathrm{out}\mapsto out, \mathrm{i}\mapsto i, \mathrm{j}\mapsto j, \mathrm{x}\mapsto x, \mathrm{y}\mapsto y, \mathrm{A}\mapsto a, \mathrm{B}\mapsto b, \mathrm{C}\mapsto c,\rho \!\!>_{env}<\!\!a\mapsto[\overline{p}], b\mapsto[\overline{q}], $\\
     $c\mapsto[\overline{k}] , m\!\!>_{mem}<\!\!in,out,i,j,x,y,\overline{q},\overline{k}\!\!>_{bnd}<\!\!\varphi\wedge\phi\wedge j=out+1\wedge j<M+1\!\!>_{form}$\\\\
     \colorbox{yellow}{$\qquad\qquad\qquad\qquad \textbf{await } \mathrm{in}-\mathrm{out}>0 \textbf{ then } skip;$}\\\\
      $4.<\!\!\mathrm{in}\mapsto in, \mathrm{out}\mapsto out, \mathrm{i}\mapsto i, \mathrm{j}\mapsto j, \mathrm{x}\mapsto x, \mathrm{y}\mapsto y, \mathrm{A}\mapsto a, \mathrm{B}\mapsto b, \mathrm{C}\mapsto c,\rho \!\!>_{env}<\!\!a\mapsto[\overline{p}], b\mapsto[\overline{q}], $\\
     $c\mapsto[\overline{k}] , m\!\!>_{mem}<\!\!in,out,i,j,x,y,\overline{q},\overline{k}\!\!>_{bnd}<\!\!\varphi\wedge\phi\wedge j=out+1\wedge j<M+1\wedge in-out>0\!\!>_{form}$\\\\
      \colorbox{yellow}{$\qquad\qquad\qquad\qquad \mathrm{y}:=C[\mathrm{out}\; mod\; N];$}\\\\
      $5.<\!\!\mathrm{in}\mapsto in, \mathrm{out}\mapsto out, \mathrm{i}\mapsto i, \mathrm{j}\mapsto j, \mathrm{x}\mapsto x, \mathrm{y}\mapsto k_{out\;mod\;N}, \mathrm{A}\mapsto a, \mathrm{B}\mapsto b, \mathrm{C}\mapsto c,\rho \!\!>_{env}<\!\!a\mapsto[\overline{p}], $\\
     $b\mapsto[\overline{q}],  c\mapsto[\overline{k}] ,m\!\!>_{mem}<\!\!in,out,i,j,x,y,\overline{q},\overline{k}\!\!>_{bnd}<\!\!\varphi\wedge\phi\wedge j=out+1\wedge j<M+1\wedge in-out>0\!\!$\\$>_{form}$\\\\
     $6.<\!\!\mathrm{in}\mapsto in, \mathrm{out}\mapsto out, \mathrm{i}\mapsto i, \mathrm{j}\mapsto j, \mathrm{x}\mapsto x, \mathrm{y}\mapsto y, \mathrm{A}\mapsto a, \mathrm{B}\mapsto b, \mathrm{C}\mapsto c,\rho \!\!>_{env}<\!\!a\mapsto[\overline{p}], b\mapsto[\overline{q}], $\\
     $c\mapsto[\overline{k}] , m\!\!>_{mem}<\!\!in,out,i,j,x,y,\overline{q},\overline{k}\!\!>_{bnd}<\!\!\varphi\wedge\phi\wedge j=out+1\wedge j<M+1\wedge in-out>0\wedge$\\$ y=p_j\!\!>_{form}$\\\\
     \colorbox{yellow}{$\qquad\qquad\qquad\qquad \mathrm{out}:=\mathrm{out}+1;$}\\\\
     $7.<\!\!\mathrm{in}\mapsto in, \mathrm{out}\mapsto out, \mathrm{i}\mapsto i, \mathrm{j}\mapsto j, \mathrm{x}\mapsto x, \mathrm{y}\mapsto y, \mathrm{A}\mapsto a, \mathrm{B}\mapsto b, \mathrm{C}\mapsto c,\rho \!\!>_{env}<\!\!a\mapsto[\overline{p}], b\mapsto[\overline{q}], $\\
     $c\mapsto[\overline{k}] , m\!\!>_{mem}<\!\!in,out,i,j,x,y,\overline{q},\overline{k}\!\!>_{bnd}<\!\!\varphi\wedge\phi\wedge j=out\wedge j<M+1\wedge y=p_j\!\!>_{form}$\\\\
     \colorbox{yellow}{$\qquad\qquad\qquad\qquad B[\mathrm{j}]:=\mathrm{y};$}\\\\
     $8.<\!\!\mathrm{in}\mapsto in, \mathrm{out}\mapsto out, \mathrm{i}\mapsto i, \mathrm{j}\mapsto j, \mathrm{x}\mapsto x, \mathrm{y}\mapsto y, \mathrm{A}\mapsto a, \mathrm{B}\mapsto b, \mathrm{C}\mapsto c,\rho \!\!>_{env}<\!\!a\mapsto[\overline{p}], b\mapsto[\overline{q}], $\\
     $c\mapsto[\overline{k}] , m\!\!>_{mem}<\!\!in,out,i,j,x,y,\overline{q},\overline{k}\!\!>_{bnd}<\!\!\varphi\wedge\phi\wedge j=out\wedge j<M+1\wedge q_j=p_j\!\!>_{form}$\\\\
     \colorbox{yellow}{$\qquad\qquad\qquad\qquad \mathrm{j}:=\mathrm{j}+1;$}\\\\
     $9.<\!\!\mathrm{in}\mapsto in, \mathrm{out}\mapsto out, \mathrm{i}\mapsto i, \mathrm{j}\mapsto j, \mathrm{x}\mapsto x, \mathrm{y}\mapsto y, \mathrm{A}\mapsto a, \mathrm{B}\mapsto b, \mathrm{C}\mapsto c,\rho \!\!>_{env}<\!\!a\mapsto[\overline{p}], b\mapsto[\overline{q}],$\\
     $ c\mapsto[\overline{k}] , m\!\!>_{mem}<\!\!in,out,i,j,x,y,\overline{q},\overline{k}\!\!>_{bnd}<\!\!\varphi\wedge\phi\wedge j=out+1\wedge j<M+1\!\!>_{form}$\\\\
      $10.<\!\!\mathrm{in}\mapsto in, \mathrm{out}\mapsto out, \mathrm{i}\mapsto i, \mathrm{j}\mapsto j, \mathrm{x}\mapsto x, \mathrm{y}\mapsto y, \mathrm{A}\mapsto a, \mathrm{B}\mapsto b, \mathrm{C}\mapsto c,\rho \!\!>_{env}<\!\!a\mapsto[\overline{p}], b\mapsto[\overline{q}], $\\
     $c\mapsto[\overline{k}] , m\!\!>_{mem}<\!\!in,out,i,j,x,y,\overline{q},\overline{k}\!\!>_{bnd}<\!\!\varphi\wedge\phi\wedge j=out+1\wedge j=M+1\!\!>_{form}$\\\\
     $\textrm{where } \varphi=(k_{(n-1)\;mod\;N}=p_n,\; out<n< in+1)\wedge (0\leq in-out\leq N)\wedge(1\leq i \leq M+1)\wedge$\\$(1\leq j \leq M+1)$ and
     $\phi=q_n=p_n,\; 0<n<j$ and
     $\overline{p}=p_0,p_1,p_2,\cdots,p_M$ and $a \mapsto[\overline{p}]=a \mapsto p_0,$\\$ a+1 \mapsto p_1,\cdots, a+M \mapsto p_M$ and
     $p_n\in Int, 0<n<M+1$\\\\
\hline
\end{tabular}
  \caption{Proof outlines of Consumer}\label{fig2}
\end{figure}
First of all, let's show that producer and consumer are "interference-free".
The only operation in consumer process that might invalidate the producer process's assertions is $\mathrm{out}:=\mathrm{out}+1$. The assertions of the producer which may be possibly invalidated are:
\begin{align*}
&6.<\!\!\mathrm{in}\mapsto in, \mathrm{out}\mapsto out, \mathrm{i}\mapsto i, \mathrm{j}\mapsto j, \mathrm{x}\mapsto x, \mathrm{y}\mapsto y, \mathrm{A}\mapsto a, \mathrm{B}\mapsto b, \mathrm{C}\mapsto c,\rho \!\!>_{env}<\!\!a\mapsto[\overline{p}], \\
&b\mapsto[\overline{q}], c\mapsto[\overline{k}],m\!\!>_{mem}<\!\!in,out,i,j,x,y,\overline{q},\overline{k}\!\!>_{bnd}<\!\!\varphi\wedge i=in+1\wedge i<M+1\wedge x=p_i\\
&\wedge in-out<N\!\!>_{form}\\
&7.<\!\!\mathrm{in}\mapsto in, \mathrm{out}\mapsto out, \mathrm{i}\mapsto i, \mathrm{j}\mapsto j, \mathrm{x}\mapsto x, \mathrm{y}\mapsto y, \mathrm{A}\mapsto a, \mathrm{B}\mapsto b, \mathrm{C}\mapsto c,\rho \!\!>_{env}<\!\!a\mapsto[\overline{p}], \\
&b\mapsto[\overline{q}], c\mapsto[\overline{k}] , m\!\!>_{mem}<\!\!in,out,i,j,x,y,\overline{q},\overline{k}\!\!>_{bnd}<\!\!\varphi\wedge i=in+1\wedge i<M+1\\
&\wedge in-out<N\wedge k_{in\;mod\;N}=p_i\!\!>_{form}
\end{align*}
Suppose $\gamma$ is a concrete configuration and
\begin{align*}
&\gamma\models<\!\!\mathrm{out}:=\mathrm{out}+1\!\!>_k<\!\!\mathrm{in}\mapsto in, \mathrm{out}\mapsto out, \mathrm{i}\mapsto i, \mathrm{j}\mapsto j, \mathrm{x}\mapsto x, \mathrm{y}\mapsto y, \mathrm{A}\mapsto a, \mathrm{B}\mapsto b, \\
&\mathrm{C}\mapsto c, \rho \!\!>_{env}<\!\!a\mapsto[\overline{p}], b\mapsto[\overline{q}], c\mapsto[\overline{k}] ,m\!\!>_{mem}<\!\!in,out,i,j,x,y,\overline{q},\overline{k}\!\!>_{bnd}<\!\!\varphi\wedge i=in+1\\
&\wedge i<M+1\wedge x=p_i\wedge in-out<N\!\!>_{form}
\end{align*}
and
\begin{align*}
&\gamma\models<\!\!\mathrm{out}:=\mathrm{out}+1\!\!>_k<\!\!\mathrm{in}\mapsto in, \mathrm{out}\mapsto out, \mathrm{i}\mapsto i, \mathrm{j}\mapsto j, \mathrm{x}\mapsto x, \mathrm{y}\mapsto y, \mathrm{A}\mapsto a, \mathrm{B}\mapsto b, \\
&\mathrm{C}\mapsto c,\rho \!\!>_{env}<\!\!a\mapsto[\overline{p}], b\mapsto[\overline{q}], c\mapsto[\overline{k}] ,m\!\!>_{mem}<\!\!in,out,i,j,x,y,\overline{q},\overline{k}\!\!>_{bnd}<\!\!\varphi\wedge i=in+1\\
&\wedge i<M+1\wedge in-out<N\wedge k_{in\;mod\;N}=p_i\!\!>_{form}
\end{align*}
and
\begin{align*}
&\gamma\models<\!\!\mathrm{out}:=\mathrm{out}+1\!\!>_k<\!\!\mathrm{in}\mapsto in, \mathrm{out}\mapsto out, \mathrm{i}\mapsto i, \mathrm{j}\mapsto j, \mathrm{x}\mapsto x, \mathrm{y}\mapsto y, \mathrm{A}\mapsto a, \mathrm{B}\mapsto b, \\
&\mathrm{C}\mapsto c,\rho \!\!>_{env}<\!\!a\mapsto[\overline{p}], b\mapsto[\overline{q}], c\mapsto[\overline{k}] ,m\!\!>_{mem}<\!\!in,out,i,j,x,y,\overline{q},\overline{k}\!\!>_{bnd}<\!\!\varphi\wedge\phi\\
&\wedge j=out+1\wedge j<M+1\wedge in-out>0\wedge y=p_j\!\!>_{form}
\end{align*}
Then there exit $in_{int},out_{int},i_{int},j_{int},x_{int},y_{int},\overline{q}_{int},\overline{k}_{int}$ such that
\begin{align*}
&\gamma=<\!\!\mathrm{out}:=\mathrm{out}+1\!\!>_k<\!\!\mathrm{in}\mapsto in_{int}, \mathrm{out}\mapsto out_{int}, \mathrm{i}\mapsto i_{int}, \mathrm{j}\mapsto j_{int}, \mathrm{x}\mapsto x_{int}, \mathrm{y}\mapsto y_{int}, \\
&\mathrm{A}\mapsto a, \mathrm{B}\mapsto b, \mathrm{C}\mapsto c,\rho \!\!>_{env}<\!\!a\mapsto[\overline{p}], b\mapsto[\overline{q}_{int}], c\mapsto[\overline{k}_{int}] ,m\!\!>_{mem}
\end{align*}
and $i_{int}=in_{int}+1$ and $i_{int}<M+1$ and $x_{int}=p_{i_{int}}$ and $in_{int}-out_{int}<N$ and $k_{in_{int}\;mod\;N}=p_{i_{int}}$ and $j_{int}=out_{int}+1$ and $j_{int}<M+1$ and $in_{int}-out_{int}>0$ and $y_{int}=p_{j_{int}}$ and $k_{(n-1)\;mod\;N}=p_{n},out_{int}<n< in_{int}+1$ and $0\leq in_{int}-out_{int}\leq N$ and $1\leq i_{int} \leq M+1$ and $1\leq j_{int} \leq M+1$ and $q_{n}=p_{n}, 0<n<j_{int}$.\\
By ASGN1 rule, we get
\begin{align*}
&\gamma'=<\!\!\cdot\!\!>_k<\!\!\mathrm{in}\mapsto in_{int}, \mathrm{out}\mapsto out_{int}+1, \mathrm{i}\mapsto i_{int}, \mathrm{j}\mapsto j_{int}, \mathrm{x}\mapsto x_{int}, \mathrm{y}\mapsto y_{int}, \mathrm{A}\mapsto a, \\
&\mathrm{B}\mapsto b, \mathrm{C}\mapsto c,\rho \!\!>_{env}<\!\!a\mapsto[\overline{p}], b\mapsto[\overline{q}_{int}], c\mapsto[\overline{k}_{int}] ,m\!\!>_{mem}
\end{align*}
Set $in=in_{int}, out=out_{int}+1, i=i_{int}, j=j_{int}, x=x_{int}, y=y_{int},\overline{q}=\overline{q}_{int},\overline{k}=\overline{k}_{int}$, then
\begin{align*}
&\gamma'\models<\!\!\cdot\!\!>_k<\!\!\mathrm{in}\mapsto in, \mathrm{out}\mapsto out, \mathrm{i}\mapsto i, \mathrm{j}\mapsto j, \mathrm{x}\mapsto x, \mathrm{y}\mapsto y, \mathrm{A}\mapsto a, \mathrm{B}\mapsto b, \mathrm{C}\mapsto c,\rho \!\!>_{env}\\
&<\!\!a\mapsto[\overline{p}], b\mapsto[\overline{q}], c\mapsto[\overline{k}] ,m\!\!>_{mem}<\!\!in,out,i,j,x,y,\overline{q},\overline{k}\!\!>_{bnd}<\!\!\varphi\wedge i=in+1\wedge i<M+1\\
&\wedge x=p_i\wedge in-out<N\!\!>_{form}
\end{align*}
and
\begin{align*}
&\gamma'\models<\!\!\cdot\!\!>_k<\!\!\mathrm{in}\mapsto in, \mathrm{out}\mapsto out, \mathrm{i}\mapsto i, \mathrm{j}\mapsto j, \mathrm{x}\mapsto x, \mathrm{y}\mapsto y, \mathrm{A}\mapsto a, \mathrm{B}\mapsto b, \mathrm{C}\mapsto c,\rho \!\!>_{env}\\
&<\!\!a\mapsto[\overline{p}], b\mapsto[\overline{q}], c\mapsto[\overline{k}] ,m\!\!>_{mem}<\!\!in,out,i,j,x,y,\overline{q},\overline{k}\!\!>_{bnd}<\!\!\varphi\wedge i=in+1\wedge i<M+1\\
&\wedge in-out<N\wedge k_{in\;mod\;N}=p_i\!\!>_{form}
\end{align*}
Thus, the consumer does not interfere with the producer; Similarly, the
producer does not interfere with the consumer. The proof outlines of $producer \parallel consumer$ is:
\begin{align*}
&<\!\!\mathrm{in}\mapsto in, \mathrm{out}\mapsto out, \mathrm{i}\mapsto i, \mathrm{j}\mapsto j, \mathrm{x}\mapsto x, \mathrm{y}\mapsto y, \mathrm{A}\mapsto a, \mathrm{B}\mapsto b, \mathrm{C}\mapsto c,\rho \!\!>_{env}<\!\!a\mapsto[\overline{p}], \\
&b\mapsto[\overline{q}], c\mapsto[\overline{k}] , m\!\!>_{mem}<\!\!in,out,i,j,x,y,\overline{q},\overline{k}\!\!>_{bnd}<\!\!\varphi\wedge j=out+1\wedge i=in+1\!\!>_{form}\\
&\colorbox{yellow}{$\qquad\qquad\qquad\qquad\qquad\qquad\qquad producer \parallel consumer;$}\\
& <\!\!\mathrm{in}\mapsto in, \mathrm{out}\mapsto out, \mathrm{i}\mapsto i, \mathrm{j}\mapsto j, \mathrm{x}\mapsto x, \mathrm{y}\mapsto y, \mathrm{A}\mapsto a, \mathrm{B}\mapsto b, \mathrm{C}\mapsto c,\rho \!\!>_{env}<\!\!a\mapsto[\overline{p}], \\
&b\mapsto[\overline{q}], c\mapsto[\overline{k}], m\!\!>_{mem}<\!\!in,out,i,j,x,y,\overline{q},\overline{k}\!\!>_{bnd}<\!\!\varphi\wedge\phi\wedge i=M+1\wedge j=M+1\wedge \\ &j=out+1\wedge i=in+1\!\!>_{form}
\end{align*}
Secondary, let us illustrate our verifier on Producer-Consumer. Figure 8 shows the validation process starting with $\Gamma_{start}\equiv <<\!\!\mathrm{in}:=0; \mathrm{out}:=0; \mathrm{i}:=1;\mathrm{j}:=1; assert\;ass_1;(producer \parallel consumer); assert\;ass_{post}\!\!>_k ass_{pre}>$. $<\Gamma_{start}>\rightarrow^*<\cdot>$ means Producer-Consumer is validated.
\begin{figure}
  \small
  \begin{tabular}{l}
    \hline
    \\
    $<\Gamma_{start}>\rightarrow^*<\Gamma_1>\rightarrow<\Gamma_2>\rightarrow<\Gamma_3, \Gamma_4, \Gamma_5>\rightarrow^*<\Gamma_6, \Gamma_7, \Gamma_8, \Gamma_9, \Gamma_5>\rightarrow<\Gamma_{10}, \Gamma_7, \Gamma_{11}, \Gamma_9, \Gamma_5>\rightarrow$\\
  $<\Gamma_{12}, \Gamma_7, \Gamma_{13}, \Gamma_9, \Gamma_5>\rightarrow<\Gamma_{14}, \Gamma_7, \Gamma_{15}, \Gamma_9, \Gamma_5>
  \rightarrow<\Gamma_{16}, \Gamma_7, \Gamma_{17}, \Gamma_9, \Gamma_5>\rightarrow<\Gamma_{18}, \Gamma_7, \Gamma_{19}, \Gamma_9, \Gamma_5>\rightarrow$\\
  $<\Gamma_{20}, \Gamma_7, \Gamma_{21}, \Gamma_9, \Gamma_5>\rightarrow<\Gamma_{22}, \Gamma_7, \Gamma_{23}, \Gamma_9, \Gamma_5>\rightarrow^*<\cdot>$\\\\
   $producer \equiv assert\;ass_{p1};\textbf{while} \;\;\mathrm{i}<M+1\;\; \textbf{do}\;\;(\mathrm{x}:=A[\mathrm{i}];\;\textbf{await}\;\; \mathrm{in}-\mathrm{out}<N \;\; \textbf{then} \;\; skip;$\\$C[\mathrm{in}\; mod\; N]:=x;\mathrm{in}:=\mathrm{in}+1;
  \mathrm{i}:=\mathrm{i}+1;)$\\\\
   $consumer\equiv assert\;ass_{c1};\;\;\textbf{while} \;\;\mathrm{j}<M+1\;\; \textbf{do} \;\;(\textbf{await} \;\; \mathrm{in}-\mathrm{out}>0 \;\; \textbf{then} \;\; skip;\mathrm{y}:=C[\mathrm{out}\; mod\; N];$\\
  $\mathrm{out}:=\mathrm{out}+1;B[\mathrm{j}]:=\mathrm{y};\mathrm{j}:=\mathrm{j}+1;)$\\\\
  $\Gamma_{start}\equiv <<\!\!\mathrm{in}:=0; \mathrm{out}:=0; \mathrm{i}:=1;\mathrm{j}:=1; assert\;ass_1;(producer \parallel consumer); assert\;ass_{post}\!\!>_k$\\ $ass_{pre}>$\\\\
   $\Gamma_1\equiv <<\!\!assert\;ass_1; (producer \parallel consumer); assert\;ass_{post}\!\!>_k ass_0>$\\\\
   $\Gamma_2\equiv <<\!\!producer \parallel consumer; assert\;ass_{post}\!\!>_k ass_1>$\\\\
  $\Gamma_3\equiv <<\!\!producer; assert\;result(<\!\!producer\!\!>_k ass_{p1})\!\!>_k ass_{p1}>$\\\\
  $\Gamma_4\equiv <<\!\!consumer; assert\;result(<\!\!consumer\!\!>_k ass_{c1})\!\!>_k ass_{c1}>$\\\\
  $\Gamma_5\equiv <<\!\!assert\;ass_{post}\!\!>_k result(<\!\!producer\!\!>_k ass_{p1})\cap result(<\!\!consumer\!\!>_k ass_{c1})>$\\\\
  $\Gamma_6\equiv <<\!\!(\mathrm{x}:=A[\mathrm{i}];\;\;\textbf{await}\;\;\mathrm{in}-\mathrm{out}<N \;\; \textbf{then}\;\; skip;C[\mathrm{in}\; mod\; N]:=x;\mathrm{in}:=\mathrm{in}+1;
  \mathrm{i}:=\mathrm{i}+1); $\\ $assert\;ass_{p1}\!\!>_k ass_{p2}>$\\\\
  $\Gamma_7\equiv <<\!\!assert\;result(<\!\!producer\!\!>_k ass_{p1})\!\!>_k ass_{p3}>$\\\\
  $\Gamma_8\equiv <<\!\!(\textbf{await} \;\; \mathrm{in}-\mathrm{out}>0 \;\; \textbf{then} \;\; skip;\mathrm{y}:=C[\mathrm{out}\; mod\; N];\mathrm{out}:=\mathrm{out}+1;B[\mathrm{j}]:=\mathrm{y};\mathrm{j}:=\mathrm{j}+1); $\\ $assert\;ass_{c1}\!\!>_k ass_{c2}>$\\\\
  $\Gamma_9\equiv <<\!\!assert\;result(<\!\!consumer\!\!>_k ass_{c1})\!\!>_k ass_{c3}>$\\\\
  $\Gamma_{10}\equiv <<\!\!(\textbf{await} \;\; \mathrm{in}-\mathrm{out}<N \;\; \textbf{then}\;\; skip;C[\mathrm{in}\; mod\; N]:=x;\mathrm{in}:=\mathrm{in}+1;
  \mathrm{i}:=\mathrm{i}+1); assert\;ass_{p1}\!\!>_k$\\ $ass_{p4}>$\\\\
  $\Gamma_{11}\equiv <<\!\!(skip;\mathrm{y}:=C[\mathrm{out}\; mod\; N];\mathrm{out}:=\mathrm{out}+1;B[\mathrm{j}]:=\mathrm{y};\mathrm{j}:=\mathrm{j}+1); assert\;ass_{c1}\!\!>_k ass_{c4}>$\\\\
  $\Gamma_{12}\equiv <<\!\!(skip;C[\mathrm{in}\; mod\; N]:=x;\mathrm{in}:=\mathrm{in}+1;
  \mathrm{i}:=\mathrm{i}+1); assert\;ass_{p1}\!\!>_k ass_{p5}>$\\\\
  $\Gamma_{13}\equiv <<\!\!(\mathrm{y}:=C[\mathrm{out}\; mod\; N];\mathrm{out}:=\mathrm{out}+1;B[\mathrm{j}]:=\mathrm{y};\mathrm{j}:=\mathrm{j}+1); assert\;ass_{c1}\!\!>_k ass_{c4}>$\\\\
  $\Gamma_{14}\equiv <<\!\!(C[\mathrm{in}\; mod\; N]:=x;\mathrm{in}:=\mathrm{in}+1;
  \mathrm{i}:=\mathrm{i}+1); assert\;ass_{p1}\!\!>_k ass_{p5}>$\\\\
  $\Gamma_{15}\equiv <<\!\!(\mathrm{out}:=\mathrm{out}+1;B[\mathrm{j}]:=\mathrm{y};\mathrm{j}:=\mathrm{j}+1); assert\;ass_{c1}\!\!>_k ass_{c5}>$\\\\
  $\Gamma_{16}\equiv <<\!\!(\mathrm{in}:=\mathrm{in}+1;
  \mathrm{i}:=\mathrm{i}+1); assert\;ass_{p1}\!\!>_k ass_{p6}>$\\\\
  $\Gamma_{17}\equiv <<\!\!(B[\mathrm{j}]:=\mathrm{y};\mathrm{j}:=\mathrm{j}+1); assert\;ass_{c1}\!\!>_k ass_{c6}>$\\\\
  $\Gamma_{18}\equiv <<\!\!
  \mathrm{i}:=\mathrm{i}+1; assert\;ass_{p1}\!\!>_k ass_{p7}>$\\\\
  $\Gamma_{19}\equiv <<\!\!\mathrm{j}:=\mathrm{j}+1; assert\;ass_{c1}\!\!>_k ass_{c7}>$\\\\
  $\Gamma_{20}\equiv <<\!\!
  assert\;ass_{p1}\!\!>_k ass_{p8}>$\\\\
  $\Gamma_{21}\equiv <<\!\!assert\;ass_{c1}\!\!>_k ass_{c8}>$\\\\
  $\Gamma_{22}\equiv <<\!\!\cdot\!\!>_k ass_{p1}>$\\\\
    \hline
\end{tabular}
\end{figure}
\begin{figure}
  \small
  \begin{tabular}{l}
    \hline
    \\
  $\Gamma_{23}\equiv <<\!\!\cdot\!\!>_k ass_{c1}>$\\\\
  $ass_{pre} \equiv<\!\!\mathrm{in}\mapsto in, \mathrm{out}\mapsto out, \mathrm{i}\mapsto i, \mathrm{j}\mapsto j, \mathrm{x}\mapsto x, \mathrm{y}\mapsto y, \mathrm{A}\mapsto a, \mathrm{B}\mapsto b, \mathrm{C}\mapsto c, \rho\!\!>_{env}<\!\!a\mapsto[\overline{p}], b\mapsto[\overline{q}], $\\
 $c\mapsto[\overline{k}], m\!\!>_{mem}<\!\!in,out,i,j,x,y,\overline{q},\overline{k}\!\!>_{bnd}<\!\!\varphi\wedge\phi\!\!>_{form}$\\
 \\
 $ass_{post} \equiv<\!\!\mathrm{in}\mapsto in, \mathrm{out}\mapsto out, \mathrm{i}\mapsto i, \mathrm{j}\mapsto j, \mathrm{x}\mapsto x, \mathrm{y}\mapsto y, \mathrm{A}\mapsto a, \mathrm{B}\mapsto b, \mathrm{C}\mapsto c,\rho \!\!>_{env}<\!\!a\mapsto[\overline{p}], b\mapsto[\overline{q}], $\\
 $c\mapsto[\overline{k}], m\!\!>_{mem}<\!\!in,out,i,j,x,y,\overline{q},\overline{k}\!\!>_{bnd}<\!\!\varphi\wedge\phi \wedge j\!=\!out\!+\!1\wedge i\!=\!in\!+\!1\!\!>_{f\!o\!r\!m}$\\
 \\
 $ass_0 \equiv<\!\!\mathrm{in}\mapsto 0, \mathrm{out}\mapsto 0, \mathrm{i}\mapsto 1, \mathrm{j}\mapsto 1, \mathrm{x}\mapsto x, \mathrm{y}\mapsto y, \mathrm{A}\mapsto a, \mathrm{B}\mapsto b, \mathrm{C}\mapsto c,\rho \!\!>_{env}<\!\!a\mapsto[\overline{p}], b\mapsto[\overline{q}], $\\
     $c\mapsto[\overline{k}] , m\!\!>_{mem}<\!\!in,out,i,j,x,y,\overline{q},\overline{k}\!\!>_{bnd}<\!\!\varphi\wedge\phi\wedge i=in+1\wedge i=1\wedge j=out+1\wedge j=1\!\!>_{form}$\\\\
     $ass_1\equiv<\!\!\mathrm{in}\mapsto in, \mathrm{out}\mapsto out, \mathrm{i}\mapsto i, \mathrm{j}\mapsto j, \mathrm{x}\mapsto x, \mathrm{y}\mapsto y, \mathrm{A}\mapsto a, \mathrm{B}\mapsto b, \mathrm{C}\mapsto c,\rho \!\!>_{env}<\!\!a\mapsto[\overline{p}], b\mapsto[\overline{q}], $\\
     $c\mapsto[\overline{k}] , m\!\!>_{mem}<\!\!in,out,i,j,x,y,\overline{q},\overline{k}\!\!>_{bnd}<\!\!\varphi\wedge\phi\wedge i=in+1\wedge j=out+1\!\!>_{form}$\\\\
     $ass_{p1}\equiv<\!\!\mathrm{in}\mapsto in, \mathrm{out}\mapsto out, \mathrm{i}\mapsto i, \mathrm{j}\mapsto j, \mathrm{x}\mapsto x, \mathrm{y}\mapsto y, \mathrm{A}\mapsto a, \mathrm{B}\mapsto b, \mathrm{C}\mapsto c,\rho \!\!>_{env}<\!\!a\mapsto[\overline{p}], b\mapsto[\overline{q}], $\\
     $c\mapsto[\overline{k}] , m\!\!>_{mem}<\!\!in,out,i,j,x,y,\overline{q},\overline{k}\!\!>_{bnd}<\!\!\varphi\wedge i=in+1\!\!>_{form}$\\\\
     $ass_{c1}\equiv<\!\!\mathrm{in}\mapsto in, \mathrm{out}\mapsto out, \mathrm{i}\mapsto i, \mathrm{j}\mapsto j, \mathrm{x}\mapsto x, \mathrm{y}\mapsto y, \mathrm{A}\mapsto a, \mathrm{B}\mapsto b, \mathrm{C}\mapsto c,\rho \!\!>_{env}<\!\!a\mapsto[\overline{p}], b\mapsto[\overline{q}], $\\
     $c\mapsto[\overline{k}] , m\!\!>_{mem}<\!\!in,out,i,j,x,y,\overline{q},\overline{k}\!\!>_{bnd}<\!\!\varphi\wedge\phi\wedge j=out+1\!\!>_{form}$\\\\
     $ass_{p2}\equiv<\!\!\mathrm{in}\mapsto in, \mathrm{out}\mapsto out, \mathrm{i}\mapsto i, \mathrm{j}\mapsto j, \mathrm{x}\mapsto x, \mathrm{y}\mapsto y, \mathrm{A}\mapsto a, \mathrm{B}\mapsto b, \mathrm{C}\mapsto c,\rho \!\!>_{env}<\!\!a\mapsto[\overline{p}], b\mapsto[\overline{q}], $\\
     $c\mapsto[\overline{k}] , m\!\!>_{mem}<\!\!in,out,i,j,x,y,\overline{q},\overline{k}\!\!>_{bnd}<\!\!\varphi\wedge i=in+1\wedge i<M+1\!\!>_{form}$\\\\
      $ass_{c2}\equiv<\!\!\mathrm{in}\mapsto in, \mathrm{out}\mapsto out, \mathrm{i}\mapsto i, \mathrm{j}\mapsto j, \mathrm{x}\mapsto x, \mathrm{y}\mapsto y, \mathrm{A}\mapsto a, \mathrm{B}\mapsto b, \mathrm{C}\mapsto c,\rho \!\!>_{env}<\!\!a\mapsto[\overline{p}], b\mapsto[\overline{q}], $\\
     $c\mapsto[\overline{k}] , m\!\!>_{mem}<\!\!in,out,i,j,x,y,\overline{q},\overline{k}\!\!>_{bnd}<\!\!\varphi\wedge\phi\wedge j=out+1\wedge j<M+1\!\!>_{form}$\\\\
      $ass_{p3}\equiv<\!\!\mathrm{in}\mapsto in, \mathrm{out}\mapsto out, \mathrm{i}\mapsto i, \mathrm{j}\mapsto j, \mathrm{x}\mapsto x, \mathrm{y}\mapsto y, \mathrm{A}\mapsto a, \mathrm{B}\mapsto b, \mathrm{C}\mapsto c,\rho \!\!>_{env}<\!\!a\mapsto[\overline{p}], b\mapsto[\overline{q}], $\\
     $c\mapsto[\overline{k}] , m\!\!>_{mem}<\!\!in,out,i,j,x,y,\overline{q},\overline{k}\!\!>_{bnd}<\!\!\varphi\wedge i=in+1\wedge i=M+1\!\!>_{form}$\\\\
      $ass_{c3}\equiv<\!\!\mathrm{in}\mapsto in, \mathrm{out}\mapsto out, \mathrm{i}\mapsto i, \mathrm{j}\mapsto j, \mathrm{x}\mapsto x, \mathrm{y}\mapsto y, \mathrm{A}\mapsto a, \mathrm{B}\mapsto b, \mathrm{C}\mapsto c,\rho \!\!>_{env}<\!\!a\mapsto[\overline{p}], b\mapsto[\overline{q}], $\\
     $c\mapsto[\overline{k}] , m\!\!>_{mem}<\!\!in,out,i,j,x,y,\overline{q},\overline{k}\!\!>_{bnd}<\!\!\varphi\wedge\phi\wedge j=out+1\wedge j=M+1\!\!>_{form}$\\\\
     $ass_{p4}\equiv<\!\!\mathrm{in}\mapsto in, \mathrm{out}\mapsto out, \mathrm{i}\mapsto i, \mathrm{j}\mapsto j, \mathrm{x}\mapsto x, \mathrm{y}\mapsto y, \mathrm{A}\mapsto a, \mathrm{B}\mapsto b, \mathrm{C}\mapsto c,\rho \!\!>_{env}<\!\!a\mapsto[\overline{p}], b\mapsto[\overline{q}], $\\
     $c\mapsto[\overline{k}] , m\!\!>_{mem}<\!\!in,out,i,j,x,y,\overline{q},\overline{k}\!\!>_{bnd}<\!\!\varphi\wedge i=in+1\wedge i<M+1\wedge x=p_i\!\!>_{form}$\\\\
      $ass_{c4}\equiv<\!\!\mathrm{in}\mapsto in, \mathrm{out}\mapsto out, \mathrm{i}\mapsto i, \mathrm{j}\mapsto j, \mathrm{x}\mapsto x, \mathrm{y}\mapsto y, \mathrm{A}\mapsto a, \mathrm{B}\mapsto b, \mathrm{C}\mapsto c,\rho \!\!>_{env}<\!\!a\mapsto[\overline{p}], b\mapsto[\overline{q}], $\\
     $c\mapsto[\overline{k}] , m\!\!>_{mem}<\!\!in,out,i,j,x,y,\overline{q},\overline{k}\!\!>_{bnd}<\!\!\varphi\wedge\phi\wedge j=out+1\wedge j<M+1\wedge in-out>0\!\!>_{form}$\\\\
      $ass_{p5}\equiv<\!\!\mathrm{in}\mapsto in, \mathrm{out}\mapsto out, \mathrm{i}\mapsto i, \mathrm{j}\mapsto j, \mathrm{x}\mapsto x, \mathrm{y}\mapsto y, \mathrm{A}\mapsto a, \mathrm{B}\mapsto b, \mathrm{C}\mapsto c,\rho \!\!>_{env}<\!\!a\mapsto[\overline{p}], b\mapsto[\overline{q}], $\\
     $c\mapsto[\overline{k}] , m\!\!>_{mem}<\!\!in,out,i,j,x,y,\overline{q},\overline{k}\!\!>_{bnd}<\!\!\varphi\wedge i=in+1\wedge i<M+1\wedge x=p_i\wedge in-out<N\!\!>_{form}$\\\\
     $ass_{c5}\equiv<\!\!\mathrm{in}\mapsto in, \mathrm{out}\mapsto out, \mathrm{i}\mapsto i, \mathrm{j}\mapsto j, \mathrm{x}\mapsto x, \mathrm{y}\mapsto y, \mathrm{A}\mapsto a, \mathrm{B}\mapsto b, \mathrm{C}\mapsto c,\rho \!\!>_{env}<\!\!a\mapsto[\overline{p}], b\mapsto[\overline{q}], $\\
     $c\mapsto[\overline{k}] , m\!\!>_{mem}<\!\!in,out,i,j,x,y,\overline{q},\overline{k}\!\!>_{bnd}<\!\!\varphi\wedge\phi\wedge j\!=\!out\!+\!1\wedge j\!<\!M\!+\!1\wedge in\!-\!out\!>\!0\wedge y\!=\!p_j\!\!>_{form}$\\\\
     $ass_{p6}\equiv<\!\!\mathrm{in}\mapsto in, \mathrm{out}\mapsto out, \mathrm{i}\mapsto i, \mathrm{j}\mapsto j, \mathrm{x}\mapsto x, \mathrm{y}\mapsto y, \mathrm{A}\mapsto a, \mathrm{B}\mapsto b, \mathrm{C}\mapsto c,\rho \!\!>_{env}<\!\!a\mapsto[\overline{p}], b\mapsto[\overline{q}], $\\
     $c\mapsto[\overline{k}] , m\!\!>_{mem}<\!\!in,out,i,j,x,y,\overline{q},\overline{k}\!\!>_{bnd}<\!\!\varphi\wedge i\!=\!in\!+\!1\wedge i\!<\!\!M\!+\!1\wedge in\!-\!out\!<\!\!N\wedge k_{in\;mod\;N}\!\!=\!p_i\!\!>_{form}$\\\\
     $ass_{c6}\equiv<\!\!\mathrm{in}\mapsto in, \mathrm{out}\mapsto out, \mathrm{i}\mapsto i, \mathrm{j}\mapsto j, \mathrm{x}\mapsto x, \mathrm{y}\mapsto y, \mathrm{A}\mapsto a, \mathrm{B}\mapsto b, \mathrm{C}\mapsto c,\rho \!\!>_{env}<\!\!a\mapsto[\overline{p}], b\mapsto[\overline{q}], $\\
     $c\mapsto[\overline{k}] , m\!\!>_{mem}<\!\!in,out,i,j,x,y,\overline{q},\overline{k}\!\!>_{bnd}<\!\!\varphi\wedge\phi\wedge j=out\wedge j<M+1\wedge y=p_j\!\!>_{form}$\\\\
      $ass_{p7}\equiv<\!\!\mathrm{in}\mapsto in, \mathrm{out}\mapsto out, \mathrm{i}\mapsto i, \mathrm{j}\mapsto j, \mathrm{x}\mapsto x, \mathrm{y}\mapsto y, \mathrm{A}\mapsto a, \mathrm{B}\mapsto b, \mathrm{C}\mapsto c,\rho \!\!>_{env}<\!\!a\mapsto[\overline{p}], b\mapsto[\overline{q}], $\\
     $c\mapsto[\overline{k}] , m\!\!>_{mem}<\!\!in,out,i,j,x,y,\overline{q},\overline{k}\!\!>_{bnd}<\!\!\varphi\wedge i=in\wedge i<M+1\!\!>_{form}$\\\\
    \hline
\end{tabular}
\end{figure}
\begin{figure}
  \small
  \begin{tabular}{l}
    \hline
    \\
    $ass_{c7}\equiv<\!\!\mathrm{in}\mapsto in, \mathrm{out}\mapsto out, \mathrm{i}\mapsto i, \mathrm{j}\mapsto j, \mathrm{x}\mapsto x, \mathrm{y}\mapsto y, \mathrm{A}\mapsto a, \mathrm{B}\mapsto b, \mathrm{C}\mapsto c,\rho \!\!>_{env}<\!\!a\mapsto[\overline{p}], b\mapsto[\overline{q}], $\\
     $c\mapsto[\overline{k}] , m\!\!>_{mem}<\!\!in,out,i,j,x,y,\overline{q},\overline{k}\!\!>_{bnd}<\!\!\varphi\wedge\phi\wedge j=out\wedge j<M+1\wedge q_j=p_j\!\!>_{form}$\\\\
      $ass_{p8}\equiv<\!\!\mathrm{in}\mapsto in, \mathrm{out}\mapsto out, \mathrm{i}\mapsto i, \mathrm{j}\mapsto j, \mathrm{x}\mapsto x, \mathrm{y}\mapsto y, \mathrm{A}\mapsto a, \mathrm{B}\mapsto b, \mathrm{C}\mapsto c,\rho \!\!>_{env}<\!\!a\mapsto[\overline{p}], b\mapsto[\overline{q}], $\\
     $c\mapsto[\overline{k}] , m\!\!>_{mem}<\!\!in,out,i,j,x,y,\overline{q},\overline{k}\!\!>_{bnd}<\!\!\varphi\wedge i=in+1\wedge i<M+1\!\!>_{form}$\\\\
     $ass_{c8}\equiv<\!\!\mathrm{in}\mapsto in, \mathrm{out}\mapsto out, \mathrm{i}\mapsto i, \mathrm{j}\mapsto j, \mathrm{x}\mapsto x, \mathrm{y}\mapsto y, \mathrm{A}\mapsto a, \mathrm{B}\mapsto b, \mathrm{C}\mapsto c,\rho \!\!>_{env}<\!\!a\mapsto[\overline{p}], b\mapsto[\overline{q}], $\\
     $c\mapsto[\overline{k}] , m\!\!>_{mem}<\!\!in,out,i,j,x,y,\overline{q},\overline{k}\!\!>_{bnd}<\!\!\varphi\wedge\phi\wedge j=out+1\wedge j<M+1\!\!>_{form}$\\\\
      $\textrm{where } \varphi=(k_{(n-1)\;mod\;N}=p_n,\; out<n< in+1)\wedge (0\leq in-out\leq N)\wedge(1\leq i \leq M+1)\wedge$ \\
     $(1\leq j \leq M+1) \textrm{ and } \phi=q_n=p_n,\; 0<n<j \textrm{ and }\overline{p}=p_0,p_1,p_2,\cdots,p_M \textrm{ and } a \mapsto[\overline{p}]=a \mapsto p_0, $ \\
     $a+1 \mapsto p_1,\cdots, a+M \mapsto p_M \textrm{ and } p_n\in Int, 0<n<M+1$\\\\
    \hline
\end{tabular}
 \caption{Verification of Producer-Consumer}\label{fig2}
\end{figure}
\section{Conclusion and Future Work}
\label{sec:typesetting-summary}
In this paper, following Grigore Ro{\c{s}}u et al's work, we consider matching logic for PIMP. In our matching logic, we  redefine "interference-free" to character parallel rule and prove the soundness of matching logic to the operational semantics of PIMP. We also link PIMP's operational semantics and PIMP's verification formally by constructing a matching logic verifier for PIMP which executes
rewriting logic semantics symbolically on configuration patterns and is sound and complete to matching logic for PIMP. That is our matching logic verifier for PIMP is sound to the operational semantics of PIMP. The state-of-the-art in mechanical program verification
is to develop and prove its proof system soundness to a trusted operational semantics. So far, we have achieved this goal in theory. Finally, we also verify the matching logic verifier through an example which is a standard
problem in parallel programming. \\
Matching logic for PIMP requires $c_1, c_2$ is "interference-free". Our further work is to generalize our results in this
paper so that they dono't depend on ``interference-free". In theory, matching logic verifier meets the state-of-the-art in mechanical program verification. Unfortunately, there is no practical executable matching logic verifier in this paper. Although it needs huge effort to implement the executable matching logic verifier, this is another work we need to do further.\\



\bibliography{lipics-v2019-sample-article}
\end{document}